\documentclass[11pt]{article}
\pdfoutput=1

\usepackage{mathtools}
\usepackage{amssymb}
\usepackage{amsmath}
\usepackage{amsthm}
\usepackage{amsfonts}

\usepackage{fullpage}
\usepackage{comment}
\includecomment{JournalOnly}
\excludecomment{ConferenceOnly}
\usepackage{graphicx}
\usepackage{algorithm}
\usepackage{algpseudocode}

\usepackage{breqn}
\usepackage{subfigure}
\usepackage{float}
\usepackage{multirow}
\newcommand{\tabincell}[2]{\begin{tabular}{@{}#1@{}}#2\end{tabular}}
\usepackage{psfrag}
\usepackage[numbers]{natbib}
\usepackage{url}
\usepackage{hyperref}
\usepackage[nameinlink,capitalize]{cleveref}
\usepackage{booktabs}
\usepackage{diagbox} %For using \diagbox
\usepackage{rotating}
\usepackage{colortbl}
\usepackage{paralist}
\usepackage{epstopdf}
\usepackage{nag}
\usepackage{microtype}
\usepackage{siunitx}
\usepackage{nicefrac}
\usepackage{lipsum}
\usepackage[english]{babel}
\usepackage[pangram]{blindtext}

\newtheorem{theorem}{Theorem}
\newtheorem{lemma}{Lemma}

\newtheorem{definition}{Definition}

\begin{document}

\title{Private Rank Aggregation under\\ Local Differential Privacy}
\author{Ziqi Yan\thanks{Beijing Key Laboratory of Security and Privacy in Intelligent Transportation,
Beijing Jiaotong University, Beijing $100044$, China. \texttt{\{zichiyen,jqliu\}@bjtu.edu.cn}.}
\and Gang Li\thanks{Centre for Cyber Security Research and Innovation, Deakin University, Geelong, VIC 3216, Australia. \texttt{gang.li@deakin.edu.au}.}
\and Jiqiang Liu\footnotemark[1]\ {~}
}

\maketitle

\begin{abstract}
As a method for answer aggregation in crowdsourced data management,
rank aggregation aims to combine different agents' answers or preferences
over the given alternatives into an aggregate ranking
which agrees the most with the preferences.
However,
since the aggregation procedure relies on a data curator,
the privacy within the agents' preference data could be compromised
when the curator is untrusted.
Existing works that guarantee differential privacy in rank aggregation
all assume that the data curator is trusted.
In this paper,
we formulate and address the problem of
\emph{locally differentially private rank aggregation},
in which the agents have no trust in the data curator.
By leveraging the approximate rank aggregation algorithm \texttt{KwikSort},
the \emph{Randomized Response} mechanism,
and the \emph{Laplace} mechanism,
we propose an effective and efficient protocol \texttt{LDP-KwikSort}.
Theoretical and empirical results show that
the solution \texttt{LDP-KwikSort:RR} can achieve the acceptable trade-off
between the utility of aggregate ranking and
the privacy protection of agents' pairwise preferences.

\end{abstract}

%=================================================================
\section{Introduction}
\label{sec-ldp-ra-introduction}

Aggregation is the process of combining multiple inputs into a single output
which represents all inputs in some sense~\cite{Beliakov16Book}.
In \emph{crowdsourced data management},
aggregation plays a crucial role:
by aggregating the answers (can be seen as preferences) from crowd agents,
the crowdsourcing platforms have able to address some ¡°computer-hard¡± tasks
such as entity resolution,
sentiment analysis,
and image recognition~\cite{LiWZF16TKDE}.
In the research of \emph{computational social choice},
one main research issue is on
how to better aggregate the preferences of individual agents,
or the participating decision-makers~\cite{Brandt16Book}.
It provides voting-based solutions to the answer aggregation problems
which often involve multiple individual preferences that could be conflicting.
Since the preferences of individual agents are often represented as rank data
where the alternatives are ranked in order,
the \emph{rank aggregation} has been a topic with broad interests in related applications.

Since most preferences data are inevitably involved with
sensitive information of individual agents,
the collecting,
analyzing,
and publishing of these data
would be a potential threat to the individual's privacy.
For instance,
due to its appealing properties,
Amazon's crowdsourcing platform Mechanical Turk
has been an important research tool for
social sciences such as psychology and sociology~\cite{Behrend11BRE,Bates13NAJP,Baker16AJFP,Shank16AS,Miller17PD,Peer17JESP,PengLN18}.
The researchers can design and post questionnaires
through the platform and recruit examinees to finish online testings.
However,
existing studies show that there are potential risks of data disclosure
within Mechanical Turk~\cite{Lease13,KangBDK14SOUPS,XiaWHS17PACMHCI,SannonC18CHI}.
Even though the disclosure of an individual's preferences is not always embarrassing,
the ability to deduce them may make those agents susceptible to coercion.
These factors prevent the contributions of accurate preferences from individual agents
and inhibit the performance of the aggregation from being fully realized.
However,
this concern cannot be comprehensively addressed
by traditional privacy-preserving methods such as anonymization,
as evidenced by the Hugo Awards $2015$ incident~\cite{HugoAwards15News,HayEM17SDM}
in which the adversary can conduct a \emph{linkage attack}
when s/he has gained unexpected background knowledge of victims.

Considering the above weakness of anonymization techniques
and especially in the scenario of aggregate ranking release,
two recent works~\cite{ShangWCK14FUSION,HayEM17SDM}
adopted the rigorous \emph{differential privacy} (DP) framework~\cite{DworkR14,LiLSY16Morgan,ZhuLZY17a,ZhuLZY17b},
and proposed several differentially private rank aggregation algorithms.
Based on the properties of the \emph{central} model of DP,
the data curator is assumed to be fully trusted
and can access all agents' ranking preference profiles,
while the adversary with any background knowledge
could not confidently infer the existence of an agent's profile
from the aggregate ranking released by the curator.

However,
with the increased awareness of privacy preservation in data collection,
both the academic and industrial communities
are getting more interested in the \emph{local} model of DP (LDP)
\cite{KasiviswanathanLNRS08FOCS,DuchiJW12NIPS,DuchiJW13FOCS}
where the curator is assumed to be untrusted.
An intuitive comparison between these two models
is shown in~\Cref{fig:ldp-diagram}.
Within the LDP model,
the agents could add noise by using the designed \emph{local randomizer}
before reporting their preference profiles $L_{i}$ to the curator,
who could also estimate the population statistics from the noisy data.
More specifically,
in this paper,
we are addressing the following problem
called \emph{locally differentially private rank aggregation} (LDP-RA):
the agents have their ranking preferences over the given alternatives
and the curator with the authority and computing capability
will collect and aggregate those preference profiles
into a final overall ranking list.
When the agents are not trusting the curator's
capability of preventing their privacy from potential attacks,
the challenges for a private rank aggregation are then
on a) how to enable the agents
to avoid sharing their original ranking preference profiles with the untrusted curator,
and b) how to enable the curator
to approximately aggregate those rankings with an acceptable utility.

\begin{figure}[htbp]
\centering
\includegraphics[width=0.8\textwidth]{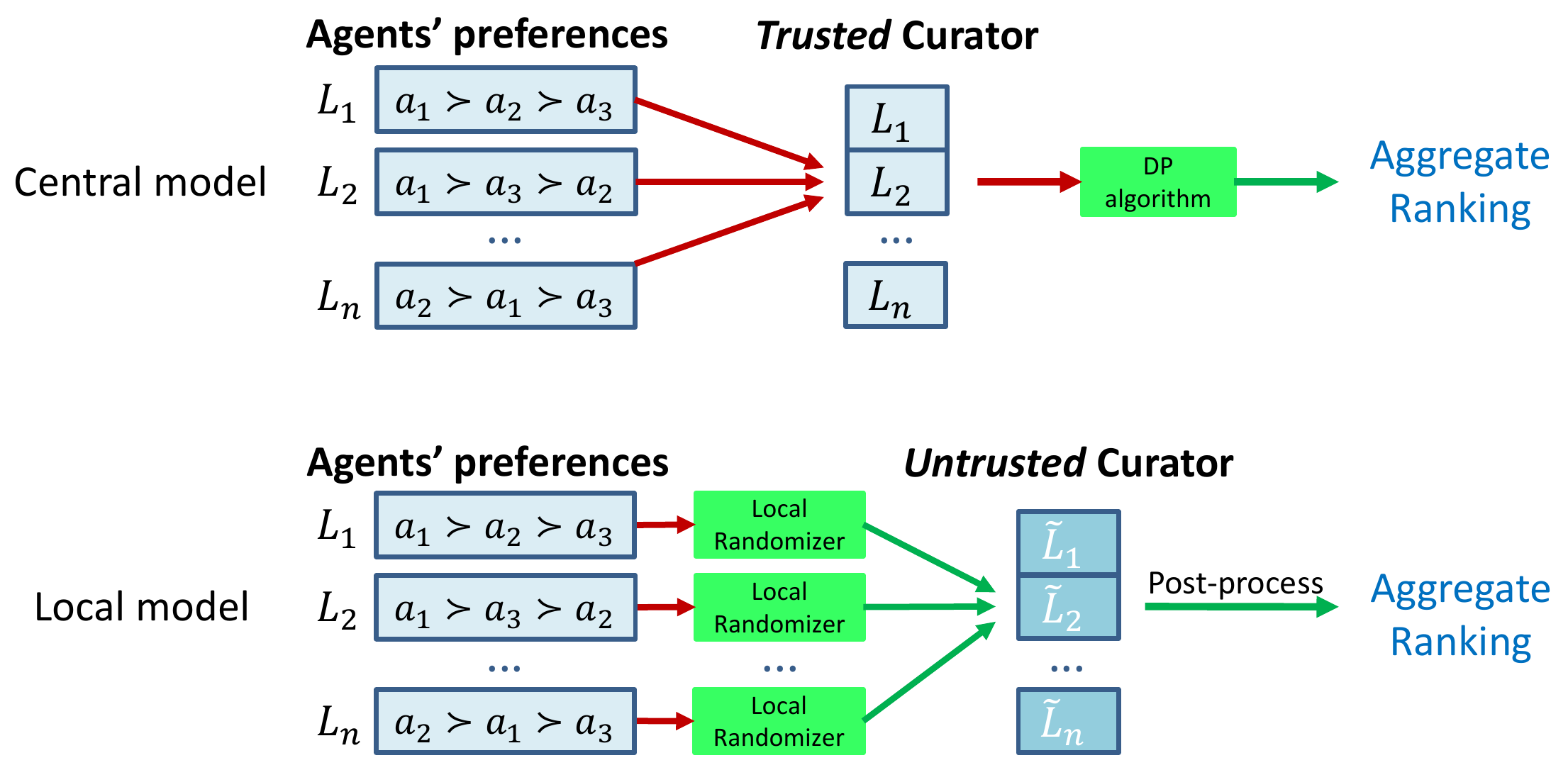}
\caption{Comparison of two DP models}
\label{fig:ldp-diagram}
\end{figure}

The main contribution of our work is \texttt{LDP-KwikSort},
an effective and efficient LDP protocol for the LDP-RA problem.
\begin{enumerate}
\item
Instead of adding noise into the whole ranking list
on each agent's side,
we focus on protecting the pairwise comparison preferences
within the ranking list.
To achieve that,
we leverage on the approximate rank aggregation algorithm
\texttt{KwikSort}~\cite{AilonCN05STOC,AilonCN08JACM},
which only requires the input as agents' pairwise preferences.
Based on this,
the protocol allows the untrusted curator to ask the agents
with pairwise comparison queries
and lets the agents report their differentially private answers
with the RR mechanism or the Laplace mechanism.
By the post-processing algorithm,
the untrusted data curator can approximately
estimate the useful frequencies for rank aggregation
and further output an aggregate ranking.

\item
When adopting the RR mechanism and Laplace mechanism
as the local randomizer
for constructing the local perturbation algorithm,
it comes up with the question of
how should we choose an appropriate number of queries $K$
which is also the times of invoking the local randomizer.
By analyzing the estimation error bound
of aggregate pairwise comparison profile $\texttt{cmp}(\textbf{L})$
in the \texttt{LDP-KwikSort} protocol,
we show that the utility
can achieve the approximate maximum value around $K=\frac{\epsilon}{2}$,
which is then further verified by extensive experiments.

\item
For performance evaluation,
we conduct experiments on three real-world datasets
(\textsf{TurkDots},
\textsf{TurkPuzzle} and \textsf{SUSHI})
and several synthetic datasets generated from the \textsf{Mallows} model.
By observing the error rate and the average Kendall tau distance
resulted from the two solutions of \texttt{LDP-KwikSort},
the central model based solution \texttt{DP-KwikSort}~\cite{HayEM17SDM}
and the non-private \texttt{KwikSort},
it shows that our protocol especially
the solution \texttt{LDP-KwikSort:RR}
can achieve strong local privacy protection
while maintaining an acceptable utility.
\end{enumerate}

The rest of this paper is organized as follows.
\Cref{sec-ldp-ra-background}
provides the background on non-private and private rank aggregation,
local differential privacy and its relaxed definition.
\Cref{sec-ldp-ra-related-work}
reviews the related work on differentially private voting mechanisms.
\Cref{sec-locally-private-rank-aggregation}
formalizes the LDP-RA problem
and proposes the \texttt{LDP-KwikSort} protocol,
followed by theoretical analysis
in~\cref{sec-ldp-ra-theoretical-analysis},
empirical analysis in~\cref{sec-ldp-ra-empirical-analysis},
and conclusions in~\cref{sec-ldp-ra-conclusion}.

\section{Preliminaries}
\label{sec-ldp-ra-background}

In this section,
we introduce the relevant concepts of rank aggregation,
private rank aggregation,
the building blocks of LDP protocol,
and a relaxation of differential privacy.
\Cref{tab:notations} lists the notations used in this paper.

\begin{table}[h]
\centering
\small
\caption{Notations}\label{tab:notations}
  \begin{tabular}{l|l}
\hline
    $N$
    & set of agent $i$, where $N=\{1,...,n\}$ \\
%\hline
    $A$
    & set of alternative $a_{j}$, where $A=\{a_{1},...,a_{m}\}$ \\
%\hline
    $\mathcal{L}(A)$
    & all the possible permutations of elements in $A$ \\
%\hline
    $L_{i}$
    & ranking preference profile of agent $i$, where $L_{i} \in \mathcal{L}(A)$ \\
%\hline
    $L_{i}^{-1}(a_{j})$
    & ranking index of alternative $a_{j}$ in $L_{i}$ \\
%\hline
    $\textbf{L}$
    & combined profile of all the agents' preferences \\
%\hline
    $L_{(\textbf{L})}$
    & aggregate ranking over the given combined profile $\textbf{L}$ \\
%\hline
    $\textbf{K}(L_{(\textbf{L})},L_{i})$
    & Kendall tau distance between $L_{(\textbf{L})}$ and $L_{i}$ \\
%\hline
    $\overline{\textbf{K}}(L_{(\textbf{L})},\textbf{L})$
    & average Kendall tau distance between $L_{(\textbf{L})}$ and $\textbf{L}$ \\
%\hline
    $C_{a_{j} a_{l}}(\textbf{L})$
    & count of the times that $L^{-1}(a_{j})<L^{-1}(a_{l})$ in $\textbf{L}$ \\
%\hline
    $\texttt{cmp}_{\textbf{L}}(a_{j},a_{l})$
    & computing $C_{a_{j} a_{l}}(\textbf{L}) - C_{a_{l} a_{j}}(\textbf{L})$ \\
%\hline
    $\texttt{cmp}(\textbf{L})$
    & set of all the values of $\texttt{cmp}_{\textbf{L}}(a_{j},a_{l})$ \\
%\hline
    $\widetilde{L}_{\mathcal{P}}$
    & aggregate ranking by the proposed protocol \\
%\hline
    $\epsilon$
    & overall privacy budget for each agent \\
%\hline
    $K$
    & number of queries from the curator to each agent \\
%\hline
    $\theta$
    & dispersion parameter of the Mallows model \\
\hline
  \end{tabular}
\end{table}

\subsection{Rank Aggregation}
\label{sec-rank-aggregation}

Given a set of $m$ alternatives $A=\{a_{1},...,a_{m}\}$
and $n$ agents participated in a preference aggregation procedure,
the preference profile of an agent $i$ is represented as
a permutation or a ranking
$L_{i} \in \mathcal{L}(A)$ of those $m$ alternatives.
Hence,
the ranking index of alternative $a_{j}$ in $L_{i}$
is denoted by $L_{i}^{-1} (a_{j})$ with a value
between $0$ (the best) to $m-1$ (the worst).
Then a \emph{combined profile} of these ranking preference profiles
is denoted by $\textbf{L} = (L_{1},...,L_{n}) \in \mathcal{L}(A)^{n}$,
based on which,
the \emph{rank aggregation} algorithm
generates a representative ranking $L_{(\textbf{L})}$
that sufficiently summarizes $\textbf{L}$.

To evaluate the quality of the aggregate ranking $L_{(\textbf{L})}$,
\emph{Kendall tau distance} is commonly used to
count the number of pairwise disagreements between two rankings:
$\textbf{K} (L_{(\textbf{L})},L_{i})
= |\{ (a_{j},a_{l}): a_{j}<a_{l}, L_{(\textbf{L})}^{-1}(a_{j})<L_{(\textbf{L})}^{-1}(a_{l})
\ but \ L^{-1}_{i}(a_{j})>L^{-1}_{i}(a_{l}) \}|$.
Then given the combined profile $\textbf{L}$,
its \emph{average Kendall tau distance}
with an aggregate ranking $L_{(\textbf{L})}$
is defined as:
$\overline{\textbf{K}} (L_{(\textbf{L})},\textbf{L}) = \frac{1}{n} \sum_{i \in N}
\textbf{K} (L_{(\textbf{L})},L_{i})$.
When the above distance achieves the minimum,
the relevant $L_{(\textbf{L})}$ is referred to
as the \emph{Kemeny optimal aggregate ranking}.
However,
it is NP-hard to compute this kind of ranking when $n > 3$,
and various approximate algorithms have been proposed
as summarized in~\cite{BrancotteYBBDH15VLDB}.

In this paper,
we leverage on a Kendall tau distance based algorithm,
$\texttt{KwikSort}$,
which could achieve $11/7$-approximation
by adopting the \texttt{QuickSort} strategy~\cite{AilonCN08JACM}.
Specifically,
the sorting of any two alternatives $a_{j},a_{l} \in A$
is based on the counts of how many times
$a_{j}$ (resp. $a_{l}$) is preferred over
$a_{l}$ (resp. $a_{j}$) among the rankings in $\textbf{L}$,
which is formally defined as
$C_{a_{j} a_{l}}(\textbf{L}) =
	|\{ for \ all \ L_{i} \in \textbf{L} | L_{i}^{-1}(a_{j}) < L_{i}^{-1}(a_{l}) \}|$
and
$C_{a_{l} a_{j}}(\textbf{L}) =
	|\{ for \ all \ L_{i} \in \textbf{L} | L_{i}^{-1}(a_{l}) < L_{i}^{-1}(a_{j}) \}|$.
Upon execution,
$\texttt{KwikSort}$
would first randomly pick an alternative $a_{p} \in A$ as the pivot,
then classify all alternatives $a_{q} \in A \setminus \{a_{p}\}$
using the comparison function
$\texttt{cmp}_{\textbf{L}}(a_{p},a_{q}) =
	\left(C_{a_{p} a_{q}}(\textbf{L}) - C_{a_{q} a_{p}}(\textbf{L})\right)
	\in \texttt{cmp}(\textbf{L})$.
That is,
if $\texttt{cmp}_{\textbf{L}}(a_{p},a_{q}) < 0$,
the alternative $a_{q}$ would be classified into
the left side of the pivot $a_{p}$,
and vice versa.
In particular,
when $\texttt{cmp}_{\textbf{L}}(a_{p},a_{q}) = 0$,
this placement can be done randomly.
This procedure repeats until all alternatives have been sorted
into a ranking ${L}_{KS}$.
In the experiments of this paper,
we adopt the Python package \texttt{pwlistorder}
to help implement $\texttt{KwikSort}$ algorithm,
and the randomized placement method is adopted instead of
its default method that directly places $a_{q}$ after $a_{p}$
if $\texttt{cmp}_{\textbf{L}}(a_{p},a_{q}) = 0$.

\subsection{Private Rank Aggregation}
\label{sec-private-rank-aggregation}

Rank aggregation under the differential privacy framework
is relatively new in the relevant community.
\cite{ShangWCK14FUSION,HayEM17SDM}
considered the central model of DP
in which the trusted data curator has access to
all agents' ranking preference profiles and
would release the aggregate ranking
by applying a differentially private algorithm $\mathcal{M}$ on them:

\begin{definition}[$\epsilon$-differential privacy~\cite{Dwork06ICALP}]
\label{def:DP}
A randomized algorithm $\mathcal{M}$ satisfies
$\epsilon$-differential privacy
if for all $\mathcal{O} \subseteq Range(\mathcal{M})$
and for all neighboring datasets $\textbf{L}$ and $\textbf{L}'$
differing on at most one record $L_{i}$
(i.e., the ranking preference profile of agent $i$),
we have
\begin{equation*}
Pr[\mathcal{M}(\textbf{L}) \in \mathcal{O}] \leq e^{\epsilon}
\cdot Pr[\mathcal{M}(\textbf{L}') \in \mathcal{O}].
\end{equation*}
\end{definition}

The intuition behind the above definition is that
the adversary cannot confidently distinguish two outputs (aggregate rankings)
of the differentially private algorithm $\mathcal{M}$
on a dataset $\textbf{L}$ and its neighboring dataset $\textbf{L}'$,
where $\textbf{L},\textbf{L}' \in \mathcal{D}$.
Then the present or absent status of any agent's ranking as a record
within the input dataset is rigorously protected with
the uncertainty of the algorithm's outputs,
which is measured by the privacy budget $\epsilon$.

To satisfy the definition of differential privacy
and for those query functions $f$ with numeric output,
the \emph{Laplace} mechanism is usually utilized.
Relying on the strategy of adding
the Laplacian random variables (noise)
to the query result,
the Laplace mechanism can be formally defined as follows:

\begin{definition}[Laplace mechanism~\cite{Dwork06ICALP}]
\label{def:lap-mech}
Given a function $f: \mathcal{D} \rightarrow \mathbb{R}^{k}$,
the Laplace mechanism is defined as
\begin{center}
$\mathcal{M}_{Lap}(\textbf{L})=f(\textbf{L})+(X_{1},...,X_{k})$,
\end{center}
where $X_{i}$ is i.i.d random variables drawn from
$\textsf{Lap}(\frac{\Delta_{g} f}{\epsilon})$,
and the global sensitivity of $f$ is
$\Delta_{g} f=\max \limits_{\substack{\textbf{L},\textbf{L}' \in \mathcal{D} \\
\Vert{\textbf{L}-\textbf{L}'}\Vert_{1}=1}} \Vert{f(\textbf{L})-f(\textbf{L}')}\Vert_{1}$.
\end{definition}

\subsection{Local Differential Privacy}
\label{sec-local-differential-privacy}

Unlike the central model,
in the local model of DP
\cite{KasiviswanathanLNRS08FOCS,DuchiJW12NIPS,DuchiJW13FOCS},
each agent would first locally perturb his/her data
by adopting a randomized algorithm,
referred to as the local randomizer $\mathcal{M}_R$,
which satisfies $\epsilon$-differential privacy.
Then,
the agent would upload the perturbed data to the untrusted curator,
who should not infer the sensitive information of each agent
but can post-process those data
to obtain the population statistics for further analysis.
The local differential privacy is formally defined as follows:

\begin{definition}
[$\epsilon$-local differential privacy~\cite{DworkR14,BassilyNST17NIPS}]
\label{def:LDP}
For a protocol $\mathcal{P}$ and dataset $D = (D_{1},...,D_{n})$,
if $\mathcal{P}$ accesses $D$
only via $K$ invocations of a local randomizer $\mathcal{M}_R$
and each invocation satisfies $\epsilon_{k}$-differential privacy,
which means that for any neighboring pair of datasets
$D_{i}$, $D_{i}'$ $\in D$
and $\forall \mathcal{O} \subseteq Range(\mathcal{M}_R)$,
if $\mathcal{M}_R$ satisfies
\begin{equation*}
Pr[\mathcal{M}_R(D_{i}) \in \mathcal{O}]
\leq e^{\epsilon_{k}} \cdot Pr[\mathcal{M}_R(D_{i}') \in \mathcal{O}]
\end{equation*}
and $\sum_{k=1}^{K} \epsilon_{k} \leq \epsilon$,
then the protocol
$\mathcal{P}(\cdot)
\triangleq
\left(\mathcal{M}_R^{(1)}(\cdot),\mathcal{M}_R^{(2)}(\cdot)
,\ldots,\mathcal{M}_R^{(K)}(\cdot)\right)$
satisfies $\epsilon$-local differential privacy ($\epsilon$-LDP).
\end{definition}

To design LDP protocols,
the \emph{randomized response}
(RR) mechanism~\cite{Warner65,GreenbergASH69,Chaudhuri16}
has been widely adopted.
As an indirect questioning mechanism for sensitive questionnaires,
RR allows the participating agents $i \in N$
to answer questions with \emph{plausible deniability}.
Specifically,
assume that the true answers for $v_{i}$ is binary $v_{i} \in \{0,1\}$,
each agent answers truthfully $u_{i}=v_{i}$
with probability $\frac{e^{\epsilon}}{e^{\epsilon} + 1}$,
and falsely $u_{i}=\overline{v_{i}}$ with probability
$\frac{1}{e^{\epsilon}+1}$~\cite{WangWH16EDBT/ICDT,HolohanLM17TIFS,WangBLJ17SS}.
These probabilities $p_{u_{i},v_{i}}$ then constitute a transformation matrix as
\[ \textbf{M} = \begin{pmatrix}
  p_{0,0} & p_{0,1} \\
  p_{1,0} & p_{1,1}
\end{pmatrix} \
= \begin{pmatrix}
  \frac{e^{\epsilon}}{e^{\epsilon} + 1} & \frac{1}{e^{\epsilon} + 1} \\
  \frac{1}{e^{\epsilon} + 1} & \frac{e^{\epsilon}}{e^{\epsilon} + 1}
\end{pmatrix}. \]

Based on the knowledge of \textbf{M},
once obtaining $y_{0}$ (or $y_{1}$),
the number of agents who answer `\num{0}' (or `\num{1}'),
the curator can use the unbiased \textit{maximum likelihood estimate} (MLE)
\cite{HuangD08ICDE,GroatEHHF13PMC,SeiO17TIFS}
to obtain an estimation $\widehat{x}_{0}$ (or $\widehat{x}_{1}$),
which approximates the number of agents
whose true answer are `\num{0}' (or `\num{1}').
\begin{equation}\label{equ:naive-reconstruction}
\overrightarrow{\widehat{X}} = \textbf{M}^{-1} \overrightarrow{Y},
\end{equation}
where
$\overrightarrow{\widehat{X}} = (\widehat{x}_{0}, \widehat{x}_{1})^{\tau}$,
$\overrightarrow{Y} = (y_{0}, y_{1})^{\tau}$,
and $\textbf{M}^{-1}$ is the inverse matrix of $\textbf{M}$.

\subsection{A Relaxation of Differential Privacy}
\label{sec-a-relaxation-of-differential-privacy}

Since the standard definition of differential privacy
requires the strict distortion of analytical results
which may lead to the significant reduction of utility
in practical implementations,
several relaxed definitions of differential privacy
have been proposed such as
the $(\epsilon,\delta)$-differential privacy,
the individual differential privacy,
and the R\'{e}nyi differential privacy.
Among them,
the individual differential privacy (iDP)
focuses on the indistinguishability
between the actual dataset and its neighboring datasets
instead of any pair of neighboring datasets,
which is defined as follows:

\begin{definition}[$\epsilon$-individual differential privacy~\cite{Soria-ComasDSM17TIFS}]
\label{def:iDP}
Given a dataset $\textbf{L}$,
a randomized algorithm $\mathcal{M}$ satisfies
$\epsilon$-individual differential privacy
if for all $\mathcal{O} \subseteq Range(\mathcal{M})$
and for any dataset $\textbf{L}'$ that is a neighbor of $\textbf{L}$,
we have
\begin{equation*}
Pr[\mathcal{M}(\textbf{L}) \in \mathcal{O}] \leq e^{\epsilon}
\cdot Pr[\mathcal{M}(\textbf{L}') \in \mathcal{O}].
\end{equation*}
\end{definition}

To satisfy the $\epsilon$-iDP,
the Laplace mechanism can be also applied,
but the sensitivity of a given function $f$
in~\Cref{def:lap-mech}
should be replaced with the \emph{local} sensitivity:

\begin{definition}[Local sensitivity~\cite{NissimRS07STOC}]
\label{def:local-sensitivity}
Given a function $f: \mathcal{D} \rightarrow \mathbb{R}^{d}$,
its local sensitivity at $\textbf{L}$ is
\begin{equation*}
\Delta_{l} f(\textbf{L})=\max \limits_{y:\Vert{\textbf{L}-y}\Vert_{1}=1}
\Vert{f(\textbf{L})-f(y)}\Vert_{1}.
\end{equation*}
\end{definition}

If we consider the individual differential privacy
under the local model,
it is natural to obtain the definition
of the \emph{local individual differential privacy} (LiDP)
as follows:

\begin{definition}[$\epsilon$-local individual differential privacy]
\label{def:LiDP}
For a protocol $\mathcal{P}$ and dataset $D = (D_{1},...,D_{n})$,
if $\mathcal{P}$ accesses $D$
only via $K$ invocations of a local randomizer $\mathcal{M}_R$
and each invocation satisfies $\epsilon_{k}$-individual differential privacy,
which means that for any neighboring dataset $D_{i}'$ of the given dataset $D_{i}$
and $\forall \mathcal{O} \subseteq Range(\mathcal{M}_R)$,
if $\mathcal{M}_R$ satisfies
\begin{equation*}
Pr[\mathcal{M}_R(D_{i}) \in \mathcal{O}] \leq e^{\epsilon_{k}} \cdot
Pr[\mathcal{M}_R(D_{i}') \in \mathcal{O}]
\end{equation*}
and $\sum_{k=1}^{K} \epsilon_{k} \leq \epsilon$,
then the protocol
$\mathcal{P}(\cdot)
\triangleq \left(\mathcal{M}_R^{(1)}(\cdot),
\mathcal{M}_R^{(2)}(\cdot),...,\mathcal{M}_R^{(K)}(\cdot)\right)$
satisfies $\epsilon$-local differential privacy ($\epsilon$-LiDP).
\end{definition}

\section{Related Work}
\label{sec-ldp-ra-related-work}

In this section,
we review some representative works
on differentially private voting mechanisms.
Under the central model of DP,
\citet{ChenCKMV13EC,ChenCKMV16TEC} considered
to model privacy in players¡¯ utility functions for
achieving both of privacy preservation and truthfulness,
and proposed mechanism for private two-candidate election;
\citet{Lee15IJCAI} proposed an algorithm that satisfies
both $\epsilon$-DP
and $\epsilon$-strategyproof for tournament voting rules.
Considering the rank aggregation scenario,
\citet{ShangWCK14FUSION} introduced a relaxed DP definition,
$(\epsilon,\delta)$-DP,
into the differentially private algorithm to
release the histogram of rankings with the utility bounds analyzed.
Then in the work of~\citet{HayEM17SDM},
three differentially private algorithms were proposed
by separately concentrating on the approximate
and the optimal rank aggregation.
Among them,
\texttt{DP-KwikSort} algorithm
extends the approximate rank aggregation algorithm \texttt{KwikSort}
by the Laplace mechanism.

Under the local model of DP,
the private algorithms for majority voting and truth discovery~\cite{LiMSGLDQ018KDD,LiXQMSGRD18arXiv},
weighted voting~\cite{YanLL19EIS},
and positional voting~\cite{WangDYDLNHX19arXiv} were designed.
Besides,
\citet{LiuLXZ18arXiv} explored the theoretical relationship
between the internal randomness of certain voting rules and
the privacy-preserving level.
However,
to the best of our knowledge,
no existing work has been devoted to
rank aggregation under the local model of DP
which is introduced in the next section.

\section{Locally Private Rank Aggregation}
\label{sec-locally-private-rank-aggregation}

In this section,
we formalize the problem of \emph{locally differentially private rank aggregation}
(LDP-RA)
and then propose a solution called \texttt{LDP-KwikSort} protocol.

\subsection{Problem Formalization}
\label{sec-problem-formalization}

There are $n$ agents that own their ranking profile $L_{i}$ over $m$ alternatives.
In the context of LDP-RA,
the combined profile $\textbf{L} = (L_{1},...,L_{n})$
can be considered as an instantiation of
$D= (D_{1},...,D_{n})$
in~\Cref{def:LDP}.
The task is then to
design a local randomizer $\mathcal{M}_R$
which locally perturbs each agent's ranking preference profile $L_{i}$
before reporting to the untrusted curator.
On the other hand,
we expect that
the constituted $\epsilon$-LDP protocol (or $\epsilon$-LiDP protocol) $\mathcal{P}$
outputs the aggregate ranking $\widetilde{L}_{\mathcal{P}}$
with an acceptable utility as measured by
the average Kendall tau distance
$\overline{\textbf{K}} (\widetilde{L}_{\mathcal{P}},\textbf{L})
= \frac{1}{n} \sum_{i \in N}
\textbf{K} (\widetilde{L}_{\mathcal{P}},L_{i})$.

\subsection{Protocol Overview}
\label{sec-protocol-overview}

To solve the LDP-RA problem,
we propose the \emph{locally differentially private KwikSort}
(\texttt{LDP-KwikSort}) protocol
which contains two solutions:
\texttt{LDP-KwikSort:RR} and \texttt{LDP-KwikSort:Lap},
and the rationale of the former
can be summarized in~\Cref{fig:protocol-rationale}.

When executing the $\texttt{LDP-KwikSort}$ protocol,
$K$ rounds of interactions exist between every agent and the curator.
In each interaction,
the curator randomly selects
$K$ different pairs of alternatives for querying,
and the agent reports the answer.
These queries are predefined
as \textit{do you prefer the alternative $a_{j}$ to $a_{l}$?}.
When receiving the query,
the agent adopts RR mechanism
to report the perturbed answer (\Cref{alg:local-perturbation-rr}).
After collecting the answers from all agents,
the curator aggregates these data
and estimates the needed statistics
which can be further used by the \texttt{KwikSort} algorithm.
Finally,
the \texttt{KwikSort} algorithm is executed
to generate an aggregate ranking
(\Cref{alg:post-processing-rr}).

\begin{figure}[h]
\centering
\includegraphics[scale=0.5]{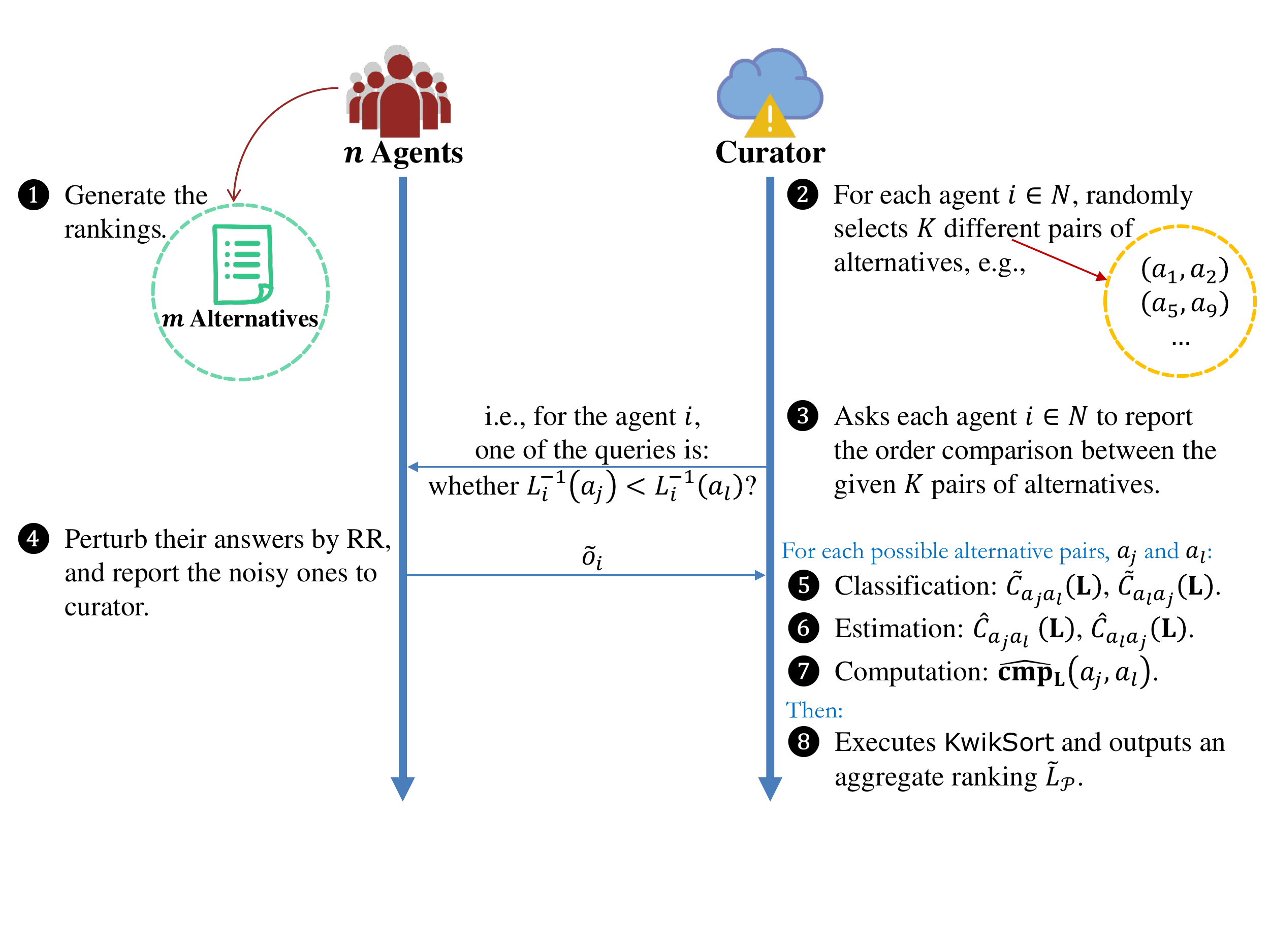}
\caption{Rationale for \texttt{LDP-KwikSort:RR} Solution}
\label{fig:protocol-rationale}
\end{figure}

\subsection{Local Perturbation}
\label{sec-local-perturbation}

As shown in~\Cref{alg:local-perturbation-rr},
when receiving the $K$ queries,
an agent $i \in N$ adopts
the standard RR mechanism
as the local randomizer $\mathcal{M}_R$,
and reports the curator with the perturbed answer
$\widetilde{\textbf{o}}_{i}$.
Here,
we adopt the transformation matrix $\textbf{M}_{rr}$
introduced in~\Cref{sec-local-differential-privacy}
to include the coin flipping probabilities $p_{\widetilde{o}_{ik},o_{ik}}$
of a true answer $o_{ik}$ being transformed
into the reported answer $\widetilde{o}_{ik}$.
Let $\textbf{M}_{rr}$ be the common knowledge
of the agents and the curator:
\[ \textbf{M}_{rr} =
\begin{pmatrix}
  p_{0,0} & p_{0,1} \\
  p_{1,0} & p_{1,1} \\
\end{pmatrix} =
\begin{pmatrix}
  p_{rr} & 1-p_{rr} \\
  1-p_{rr} & p_{rr} \\
\end{pmatrix}\]
where all the diagonal elements get assigned the value $p_{rr}$
while $1-p_{rr}$ for other elements,
and $p_{rr} = \frac{e^{\epsilon_{k}}}{e^{\epsilon_{k}} + 1}$,
$\epsilon_{k}=\epsilon/K$.

\begin{algorithm}[h]
\small
\caption{Local Perturbation in \texttt{LDP-KwikSort:RR}}
\label{alg:local-perturbation-rr}
\begin{algorithmic}[1]
\Require
        Agent $i$'s ranking profile $L_{i}$,
        $K$ queries,
        transformation matrix $\textbf{M}_{rr}$
\Ensure
        Agent $i$'s perturbed answer $\widetilde{\textbf{o}}_{i}$
\State  Receives $K$ queries (e.g., \emph{do you prefer alternative $a_{j}$ to $a_{l}$?}) from curator
\For    {each query $k \in [K]$}
\State  Randomly generates a number $g$ from $[0,1]$ \Comment{Local Randomizer $\mathcal{M}_{R}$}
\If     {$g \leq p_{rr}$}
\State  $\widetilde{o}_{ik} \leftarrow 1$
\Else
\State  $\widetilde{o}_{ik} \leftarrow 0$
\EndIf
\EndFor
\State  Agent $i$ reports the perturbed answer $\widetilde{\textbf{o}}_{i}=\{\widetilde{o}_{i1},...,\widetilde{o}_{iK}\}$ to curator
\end{algorithmic}
\end{algorithm}

\subsection{Post-processing}
\label{sec-post-processing}

As shown in~\Cref{alg:post-processing-rr},
once collected the answers from the agents,
\textbf{firstly} (Line $1-9$),
the curator classifies these noisy data.
For each possible pair of alternatives $a_{j}$ and $a_{l}$,
two counters $\widetilde{C}_{a_{j} a_{l}}(\textbf{L})$
and $\widetilde{C}_{a_{l} a_{j}}(\textbf{L})$ are updated
to indicate the numbers of the agent
who prefers alternative $a_{j}$ (or $a_{l}$)
to alternative $a_{l}$ (or $a_{j}$).
\textbf{Secondly} (Line $10-13$),
the curator obtains the estimations of
$C_{a_{j} a_{l}}(\textbf{L})$ and $C_{a_{l} a_{j}}(\textbf{L})$
by~\Cref{equ:naive-reconstruction},
where
$\overrightarrow{\widehat{X}}
= \left(\widehat{C}_{a_{j} a_{l}}(\textbf{L}), \widehat{C}_{a_{l} a_{j}}(\textbf{L})\right)$,
$\overrightarrow{Y}
= \left(\widetilde{C}_{a_{j} a_{l}}(\textbf{L}), \widetilde{C}_{a_{l} a_{j}}(\textbf{L})\right)^{\tau}$.
Then the results from the comparison function
$\widehat{\texttt{cmp}}_{\textbf{L}}(a_{j},a_{l})$ can be computed.
\textbf{Finally} (Line $14-15$),
the curator assembles all the values
of $\widehat{\texttt{cmp}}_{\textbf{L}}(a_{j},a_{l})$ into the estimated version
of aggregate pairwise comparison profile
$\widehat{\texttt{cmp}}(\textbf{L})$.
By taking $\widehat{\texttt{cmp}}(\textbf{L})$ as input,
the standard $\texttt{KwikSort}$ algorithm is executed
to generate an aggregate ranking $\widetilde{L}_{\mathcal{P}}$.

\begin{algorithm}[h]
\small
\caption{Post-processing in \texttt{LDP-KwikSort:RR}}
\label{alg:post-processing-rr}
\begin{algorithmic}[1]
\Require
        Agents' perturbed answers $\widetilde{\textbf{o}}=\{\widetilde{\textbf{o}}_{1},...,\widetilde{\textbf{o}}_{n}\}$,
        transformation matrix $\textbf{M}_{rr}$
\Ensure
        Aggregate ranking $\widetilde{L}_{\mathcal{P}}$
\For    {$\widetilde{\textbf{o}}_{i} \in \widetilde{\textbf{o}}$} \Comment{Classification}
\For    {$\widetilde{o}_{ik} \in \widetilde{\textbf{o}}_{i}$}
\If     {$\widetilde{o}_{ik}=1$}
\State  $\widetilde{C}_{a_{j} a_{l}}(\textbf{L}) + 1$
\Else
\State  $\widetilde{C}_{a_{l} a_{j}}(\textbf{L}) + 1$
\EndIf
\EndFor
\EndFor
\For    {each possible pair $(a_{j},a_{l})$} \Comment{Estimation and computation}
\State  Obtains two estimations $\widehat{C}_{a_{j} a_{l}}(\textbf{L})$ and $\widehat{C}_{a_{l} a_{j}}(\textbf{L})$ by ~\Cref{equ:naive-reconstruction}
\State  $\widehat{\texttt{cmp}}_{\textbf{L}}(a_{j},a_{l}) \leftarrow \widehat{C}_{a_{j} a_{l}}(\textbf{L}) - \widehat{C}_{a_{l} a_{j}}(\textbf{L})$
\EndFor
\State  Obtains $\widehat{\texttt{cmp}}(\textbf{L})$ and executes \texttt{KwikSort} algorithm \\
\Return Aggregate ranking $\widetilde{L}_{\mathcal{P}}$
\end{algorithmic}
\end{algorithm}

\subsection{Laplace Solution}

Since the agent $i$'s perturbed answer $\widetilde{\textbf{o}}_{i}$
contains numeric values,
we also propose the Laplace mechanism based solution \texttt{LDP-KwikSort:Lap}.

\paragraph{\textbf{Local perturbation.}}
As shown in~\Cref{alg:local-perturbation-lap},
when receiving the $K$ queries,
an agent $i \in N$ adopts
the Laplace mechanism
as the local randomizer $\mathcal{M}_R$,
and reports the curator with the perturbed answer
$\widetilde{\textbf{o}}_{i}$.
Here,
$f_{o}(\cdot)$ is a query function
which takes agent's ranking preference as input
and outputs $1$ or $0$ to indicate
whether this agent prefers $a_{j}$ to $a_{l}$ or not,
for default index $j<l$.
Then the local sensitivity $\Delta_{l} f_{o}$ is $1$,
which reflects the maximum change of its outputs.

\begin{algorithm}[h]
\small
\caption{Local Perturbation in \texttt{LDP-KwikSort:Lap}}
\label{alg:local-perturbation-lap}
\begin{algorithmic}[1]
\Require
        Agent $i$'s ranking profile $L_{i}$,
        $K$ queries,
        privacy budget $\epsilon$
\Ensure
        Agent $i$'s perturbed answer $\widetilde{\textbf{o}}_{i}$
\State  Receives $K$ queries (e.g., \emph{do you prefer alternative $a_{j}$ to $a_{l}$?}) from curator
\For    {each query $k \in [K]$}
\State  $\widetilde{o}_{ik} \leftarrow o_{ik} + X_{ldp}$
where
$X_{ldp} \sim \textsf{Lap}(\frac{\Delta_{l} f_{o}}{\epsilon_{k}})$
\Comment{Local Randomizer $\mathcal{M}_{R}$}
\EndFor
\State  Agent $i$ reports the perturbed answer $\widetilde{\textbf{o}}_{i}=\{\widetilde{o}_{i1},...,\widetilde{o}_{iK}\}$ to curator
\end{algorithmic}
\end{algorithm}

\paragraph{\textbf{Post-processing.}}
As shown in~\Cref{alg:post-processing-lap},
once collected the answers from the agents,
\textbf{firstly} (Line $1-9$),
the curator also classifies these noisy data.
For each possible pair of alternatives $a_{j}$ and $a_{l}$,
two counters $\widehat{C}_{a_{j} a_{l}}(\textbf{L})$
and $\widehat{C}_{a_{l} a_{j}}(\textbf{L})$ are updated.
Note that these counters are the estimations
of $C_{a_{j} a_{l}}(\textbf{L})$ and $C_{a_{l} a_{j}}(\textbf{L})$,
and the updates are based on judging whether $\widetilde{o}_{ik} \geq 0.5$,
which is different from the approach in~\Cref{alg:post-processing-rr}.
\textbf{Secondly} (Line $10-12$),
the curator obtains $\widehat{\texttt{cmp}}_{\textbf{L}}(a_{j},a_{l})$
by directly subtracting $\widehat{C}_{a_{l} a_{j}}(\textbf{L})$
from $\widehat{C}_{a_{j} a_{l}}(\textbf{L})$.
\textbf{Finally} (Line $13-14$),
the curator assembles all the values
of $\widehat{\texttt{cmp}}_{\textbf{L}}(a_{j},a_{l})$ into the estimated version
of aggregate pairwise comparison profile $\widehat{\texttt{cmp}}(\textbf{L})$.
By taking $\widehat{\texttt{cmp}}(\textbf{L})$ as input,
the standard $\texttt{KwikSort}$ algorithm is executed
to generate an aggregate ranking $\widetilde{L}_{\mathcal{P}}$.

\begin{algorithm}[h]
\small
\caption{Post-processing in \texttt{LDP-KwikSort:Lap}}
\label{alg:post-processing-lap}
\begin{algorithmic}[1]
\Require
        Agents' perturbed answers $\widetilde{\textbf{o}}=\{\widetilde{\textbf{o}}_{1},...,\widetilde{\textbf{o}}_{n}\}$
\Ensure
        Aggregate ranking $\widetilde{L}_{\mathcal{P}}$
\For    {$\widetilde{\textbf{o}}_{i} \in \widetilde{\textbf{o}}$} \Comment{Classification}
\For    {$\widetilde{o}_{ik} \in \widetilde{\textbf{o}}_{i}$}
\If     {$\widetilde{o}_{ik} \geq 0.5$}
\State  $\widehat{C}_{a_{j} a_{l}}(\textbf{L}) + 1$
\Else
\State  $\widehat{C}_{a_{l} a_{j}}(\textbf{L}) + 1$
\EndIf
\EndFor
\EndFor
\For    {each possible pair $(a_{j},a_{l})$} \Comment{Computation}
\State  $\widehat{\texttt{cmp}}_{\textbf{L}}(a_{j},a_{l}) \leftarrow \widehat{C}_{a_{j} a_{l}}(\textbf{L}) - \widehat{C}_{a_{l} a_{j}}(\textbf{L})$.
\EndFor
\State  Obtains $\widehat{\texttt{cmp}}(\textbf{L})$ and executes \texttt{KwikSort} algorithm \\
\Return Aggregate ranking $\widetilde{L}_{\mathcal{P}}$
\end{algorithmic}
\end{algorithm}

\section{Theoretical Analysis}
\label{sec-ldp-ra-theoretical-analysis}

In this section,
we provide the privacy and utility guarantees,
as well as the computational complexity
of the proposed \texttt{LDP-KwikSort} protocol.

\subsection{Privacy Guarantee}

We have the following theorem on the privacy guarantee
of the \texttt{LDP-KwikSort} protocol.

\begin{theorem}\label{thr:privacy-guarantee}
\texttt{LDP-KwikSort:RR} satisfies
$\epsilon$-LDP
and \texttt{LDP-KwikSort:Lap} satisfies $\epsilon$-LiDP.
\end{theorem}
\begin{proof}
We respectively analyze two solutions of \texttt{LDP-KwikSort} protocol as follows:
\begin{enumerate}
\item[a.]
In the local perturbation algorithm,
since the solution \texttt{LDP-KwikSort:RR} adopts
the transformation matrix based local randomizer $\mathcal{M}_{R}$
to perturb each answer from the agent,
we have
\begin{equation*}
\frac{Pr[\mathcal{M}_{R}(o_{ik} = u) = u]}{Pr[\mathcal{M}_{R}(o_{ik} = v) = u]} =
\frac{2 p_{rr}}{1-p_{rr}} = e^{\epsilon_{k}}.
\end{equation*}
Based on~\Cref{def:DP},
we can conclude that the local randomizer $\mathcal{M}_{R}$
in solution \texttt{LDP-KwikSort:RR} satisfies $\epsilon_{k}$-DP.

\item[b.]
In the local perturbation algorithm,
since the solution \texttt{LDP-KwikSort:Lap} adopts
the local sensitivity based Laplace mechanism
to design the local randomizer $\mathcal{M}_{R}$
for perturbing each answer from the agent,
based on~\Cref{def:iDP},
we can conclude that the local randomizer $\mathcal{M}_{R}$
in solution \texttt{LDP-KwikSort:Lap}
satisfies $\epsilon_{k}$-iDP.
\end{enumerate}
Consequently,
according to~\Cref{def:LDP},
we can conclude that
\texttt{LDP-KwikSort:RR} satisfies $\epsilon$-LDP and
\texttt{LDP-KwikSort:Lap} satisfies $\epsilon$-LiDP,
where $\epsilon=K \cdot \epsilon_{k}$.
\end{proof}

\subsection{Utility Guarantee}
\label{sec-ldp-ra-utility-guarantee}

\subsubsection{Analysis of \texttt{LDP-KwikSort:RR}}

As described in~\Cref{alg:post-processing-rr},
the post-processing procedure on the side of the curator
mainly contains two parts:
the estimation of $\texttt{cmp}(\textbf{L})$ from the noisy data,
and the execution of the \texttt{KwikSort} algorithm.
It is known that \texttt{KwikSort} can
output an $11/7$-optimal aggregate ranking
with a given aggregate pairwise comparison profile
$\texttt{cmp}(\textbf{L})$.
Therefore,
it is important to measure the error
of estimating $\texttt{cmp}(\textbf{L})$ from the first part,
which indicates the gap between
the aggregate ranking $\widetilde{L}_{\mathcal{P}}$
from \texttt{KwikSort} protocol and
the $11/7$-optimal aggregate ranking ${L}_{KS}$.
When investigating the internal disagreements of
the original aggregate pairwise comparison profile $\texttt{cmp}(\textbf{L})$
and its estimated version $\widehat{\texttt{cmp}}(\textbf{L})$,
we use~\Cref{equ:naive-reconstruction} to obtain
\begin{equation*}
\textbf{M}_{rr} \overrightarrow{\widehat{X}} = \overrightarrow{Y},
\end{equation*}
That is
\[
\begin{pmatrix}
  p_{rr} & 1-p_{rr} \\
  1-p_{rr} & p_{rr}
\end{pmatrix}
\begin{pmatrix}
  \widehat{C}_{a_{j} a_{l}}(\textbf{L}) \\
  \widehat{C}_{a_{l} a_{j}}(\textbf{L})
\end{pmatrix}
=
\begin{pmatrix}
  \widetilde{C}_{a_{j} a_{l}}(\textbf{L}) \\
  \widetilde{C}_{a_{l} a_{j}}(\textbf{L})
\end{pmatrix}.
\]
Expend the equation,
we have
\begin{dmath}
\widehat{\texttt{cmp}}_{\textbf{L}}(a_{j},a_{l})
= \widehat{C}_{a_{j} a_{l}}(\textbf{L}) - \widehat{C}_{a_{l} a_{j}}(\textbf{L})
=\frac{\widetilde{C}_{a_{j} a_{l}}(\textbf{L}) - \widetilde{C}_{a_{l} a_{j}}(\textbf{L})}{p_{rr}-(1-p_{rr})}
=\frac{\widetilde{\texttt{cmp}}_{\textbf{L}}(a_{j},a_{l})}{2p_{rr}-1}.
\end{dmath}
Since the denominator $(2p_{rr}-1)>0$ when $\epsilon_{k}>0$,
we can conclude that
the sign of $\widehat{\texttt{cmp}}_{\textbf{L}}(a_{j},a_{l})$
is the same as that of $\widetilde{\texttt{cmp}}_{\textbf{L}}(a_{j},a_{l})$.
As the \texttt{KwikSort} algorithm
determines the relative order of alternatives $(a_{j},a_{l})$
only by checking the sign of $\texttt{cmp}_{\textbf{L}}(a_{j},a_{l})$,
we turn to measure the probability that
the sign of $\widetilde{\texttt{cmp}}_{\textbf{L}}(a_{j},a_{l})$
is different with that of $\texttt{cmp}_{\textbf{L}}(a_{j},a_{l})$.
That is to say,
the achieved utility of \texttt{LDP-KwikSort}
relies on the condition that
for the ground truth
$L_{(\textbf{L})}^{-1}(a_{j}) < L_{(\textbf{L})}^{-1}(a_{l})$
whether the curator can directly obtain it
by just observing the noisy data from the agents.
Besides,
when $\widetilde{\texttt{cmp}}_{\textbf{L}}(a_{j},a_{l})$
and $\texttt{cmp}_{\textbf{L}}(a_{j},a_{l})$ have the same sign,
the difference between their absolute values
will not impact on the utility of solution \texttt{LDP-KwikSort:RR}.

Next,
we present an error bound and the associated proof,
which are based on the assumption that
the agents' ranking preference profiles
are generated by the \textsf{Mallows} model.
The \textsf{Mallows} model~\cite{Mallows57}
is a probabilistic model for ranking generation
in which the probability of generating a ranking is based on
the dispersion parameter $\theta \in [0,1]$ and the distance
between this ranking and a ground truth ranking.
Specifically,
$\theta=1-\frac{q_{M}}{p_{M}}$
where $p_{M} \in [\frac{1}{2},1]$ is the probability
for generating the relative order of alternatives
$(a_{j},a_{l})$ which is consistent with the ground truth
and $q_{M}=1-p_{M}$ is the opposite probability.

\begin{theorem}\label{thr:utility-guarantee-rr}
If all the ranking preference profiles are generated by
the \textsf{Mallows} model with the dispersion parameter $\theta$
and a ground truth ranking,
the estimation of $\texttt{cmp}(\textbf{L})$
by \texttt{LDP-KwikSort:RR} produces the error $< 6\mu$
with probability at least $1-2^{-6\mu}$,
where
$\mu = 2 \binom{m}{2}
\exp\left(-\frac{\epsilon^2 K}{(\epsilon+2K)^2}\cdot\frac{{\theta^*}^2 n}{m (m-1)}\right)$
and $\theta^*=\frac{\theta}{2-\theta}$.
\end{theorem}

\begin{proof}
In our scenario,
for a certain pairwise comparison query such that
`for alternative pair $(a_{j},a_{l})$,
whether $L_{i}^{-1}(a_{j}) < L_{i}^{-1}(a_{l})$ or not',
there will be $n^* = n \frac{K}{\binom{m}{2}}$ agents involved,
and each agent reports his/her true answer (resp. false answer)
with probability $p_{rr}=\frac{e^{\epsilon_{k}}}{e^{\epsilon_{k}} + 1}$
(resp. $q_{rr} = 1-p_{rr} =\frac{1}{e^{\epsilon_{k}} + 1}$).
When considering the assumption regarding the \textsf{Mallows} model,
an agent finally reports the ground truth answer
(i.e.,
$L_{i}^{-1}(a_{j}) < L_{i}^{-1}(a_{l})$ for all $j<l$)
with probability $\verb"p" = p_{M} \cdot p_{rr} + q_{M} \cdot q_{rr}$
and reports the opposite answer with probability
$\verb"q" = p_{M} \cdot q_{rr} + q_{M} \cdot p_{rr}$.

Since the above procedure could be seen as a Bernoulli trial,
here we adopt the variation of Hoeffding's inequality~\cite{HoeffdingInequality}
for Bernoulli trial as follows:

\begin{lemma}\label{lem:hoeffdinginequality}
For a Bernoulli trial in which $n$ times of experiment
have been conducted and each experiment outputs
outcome $A$ with probability $\verb"p"$
and outcome $B$ with probability $\verb"q"=1-\verb"p"$,
we have the following probability inequality for some $\delta>0$:
\begin{equation*}
Pr[(\verb"p"-\delta)n \leq H(n) \leq (\verb"p"+\delta)n] \geq 1-2\exp(-2{\delta}^2n),
\end{equation*}
where $H(n)$ is the number of outcome $A$ in $n$ experiments.
\end{lemma}

We then have the following probability inequalities
by~\Cref{lem:hoeffdinginequality}:
\begin{equation*}
Pr[(\verb"p"-\delta)n^* \leq
\widetilde{C}_{a_{j} a_{l}}(\textbf{L}) \leq (\verb"p"+\delta)n^*] \geq 1-2\exp(-2{\delta}^2n^*),
\end{equation*}
\begin{equation*}
Pr[(\verb"q"-\delta)n^* \leq
\widetilde{C}_{a_{l} a_{j}}(\textbf{L}) \leq (\verb"q"+\delta)n^*] \geq 1-2\exp(-2{\delta}^2n^*),
\end{equation*}
and hence
\begin{equation*}
Pr[(\verb"p"-\verb"q"-2\delta)n^* \leq
\widetilde{C}_{a_{j} a_{l}}(\textbf{L}) - \widetilde{C}_{a_{l} a_{j}}(\textbf{L})
\leq (\verb"p"-\verb"q"+2\delta)n^*] \geq 1-2\exp(-2{\delta}^2n^*).
\end{equation*}

If we let $\verb"p"-\verb"q"=2\delta$
and $\theta^*=\frac{\theta}{2-\theta}$,
the following probability inequalities are obtained:
\begin{equation*}
Pr[0 \leq \widetilde{C}_{a_{j} a_{l}}(\textbf{L}) -
\widetilde{C}_{a_{l} a_{j}}(\textbf{L}) \leq 4\delta n^*] \geq 1-2\exp(-2{\delta}^2n^*)
\end{equation*}
and
\begin{dmath}\label{equ:ldp-kwiksort-rr-error}
Pr[\widetilde{\texttt{cmp}}_{\textbf{L}}(a_{j},a_{l}) =
\widetilde{C}_{a_{j} a_{l}}(\textbf{L}) -
\widetilde{C}_{a_{l} a_{j}}(\textbf{L})<0] < 2\exp(-2{\delta}^2n^*)
= 2\exp\left(-2(\frac{p_{M} \cdot p_{rr} +
q_{M} \cdot q_{rr} - p_{M} \cdot q_{rr} - q_{M} \cdot p_{rr}}{2})^2n^*\right)
= 2\exp\left(- \frac{({p_{rr}-q_{rr}})^2 {\theta^*}^2 n^*}{2}\right)
\approx 2\exp\left(-\frac{\epsilon^2 K}{(\epsilon+2K)^2}\cdot\frac{{\theta^*}^2 n}{m (m-1)}\right)
= P_{\iota}.
\end{dmath}
This formula reflects the error bound
of generating data and reporting answers by $n$ agents
for a certain pair of alternatives $(a_{j},a_{l})$.
Here,
we adopt the inequality $e^x \geq x+1$
for simplifying the result in the last derivation.

After obtaining the error bound
for a certain pairwise comparison query,
we now apply the following Chernoff bound
for all the possible queries:
\begin{lemma}
Let $X_1,...,X_n$ be independent Poisson trials such that
$Pr(X_i=1)=p_i$.
Let $X=\sum_{i=1}^{n} X_i$ and $\mu=\mathbf{E}[X]$.
Then for $R \geq 6\mu$,
\begin{equation*}
Pr(X \geq R) \leq 2^{-R}.
\end{equation*}
\end{lemma}

Therefore,
we finish the proof with obtaining
\begin{equation*}
\mu = \mathbf{E}[X] = \mathbf{E}[\sum_{\iota=1}^{\binom{m}{2}} X_{\iota}]
= \sum_{\iota=1}^{\binom{m}{2}}\mathbf{E}[X_{\iota}]
= \sum_{\iota=1}^{\binom{m}{2}} P_{\iota}
\leq 2\binom{m}{2} \exp
\left(-\frac{\epsilon^2 K}{(\epsilon+2K)^2}\cdot\frac{{\theta^*}^2 n}{m (m-1)}\right).
\end{equation*}
\end{proof}

Based on~\Cref{thr:utility-guarantee-rr},
we can find that the estimation error of \texttt{LDP-KwikSort:RR}
depends on the expectation $\mu$,
and a smaller $\mu$ indicates a smaller error.
Then we observe that
1) when the privacy budget $\epsilon$,
the number of agents $n$ and alternatives $m$,
and the dispersion parameter $\theta$ are fixed,
$\mu$ is influenced by the function $g(K)=\frac{\epsilon^2 K}{(\epsilon+2K)^2}$,
and it gets the approximate minimum value
when the number of queries $K=\frac{\epsilon}{2}$;
2) when more agents are involved
or given a large privacy budget for each agent,
i.e.,
$n\rightarrow \infty$ or $\epsilon\rightarrow \infty$,
$\mu$ will be reduced to zero;
3) with increasing the number of alternatives $m$,
$\mu$ is also increased;
4) when the dispersion parameter $\theta$
is relatively large,
i.e.,
the generated rankings are closer to the ground truth ranking,
we obtain a relatively small $\mu$.
In~\Cref{sec-ldp-ra-empirical-analysis},
we conduct a series of experiments
to verify these theoretical observations.

\subsubsection{Analysis of \texttt{LDP-KwikSort:Lap}}

As described in~\Cref{alg:post-processing-lap},
the post-processing procedure on the side of the curator
also contains two parts:
the estimation of $\texttt{cmp}(\textbf{L})$ from the noisy data,
and the execution of the \texttt{KwikSort} algorithm.
Its most difference with~\Cref{alg:post-processing-rr}
lies in using the threshold checking about $\widetilde{o}_{ik}$
instead of using~\Cref{equ:naive-reconstruction}
to obtain $\widehat{C}_{a_{j} a_{l}}(\textbf{L})$
and $\widehat{C}_{a_{l} a_{j}}(\textbf{L})$.
Then we have the expressions regarding $\widetilde{o}_{ik}$:
\begin{equation*}
\widetilde{o}_{ik} =
\begin{cases}
0, & \text{if $o_{ik} + \textsf{Lap}(\frac{\Delta_{l} f_{o}}{\epsilon_{k}}) < 0.5$},\\
1, & \text{if $o_{ik} + \textsf{Lap}(\frac{\Delta_{l} f_{o}}{\epsilon_{k}}) \geq 0.5$}.
\end{cases}
\end{equation*}
According to the cumulative distribution function
of Laplacian random variables,
we have the probabilities
of a true answer $o_{i}$ being transformed
into the reported answer $\widetilde{o}_{i}$ as
\begin{equation*}
p_{\widetilde{o}_{ik},o_{ik}} =
\begin{cases}
p_{0,0} = F(0.5) = 1-\frac{1}{2}\exp(-\frac{\epsilon_{k}}{2}),\\
p_{0,1} = 1-F(0.5) = \frac{1}{2}\exp(-\frac{\epsilon_{k}}{2}),\\
p_{1,0} = F(-0.5) = \frac{1}{2}\exp(-\frac{\epsilon_{k}}{2}),\\
p_{1,1} = 1-F(-0.5) = 1-\frac{1}{2}\exp(-\frac{\epsilon_{k}}{2}),
\end{cases}
\end{equation*}
and we further have the following transformation matrix:
\[ \textbf{M}_{lap} =
\begin{pmatrix}
  p_{0,0} & p_{0,1} \\
  p_{1,0} & p_{1,1} \\
\end{pmatrix} =
\begin{pmatrix}
  p_{lap} & 1-p_{lap} \\
  1-p_{lap} & p_{lap} \\
\end{pmatrix}, \]
where all the diagonal elements are assigned with the value $p_{lap}$
while $1-p_{lap}$ for other elements,
and $p_{lap} = 1-\frac{1}{2}\exp(-\frac{\epsilon_{k}}{2})$,
$\epsilon_{k}=\epsilon/K$.
Then we follow the analysis technique of \texttt{LDP-KwikSort:RR}
and replace $p_{rr}$ with $p_{lap}$ in~\Cref{equ:ldp-kwiksort-rr-error}
to obtain the following conclusion:

\begin{theorem}\label{thr:utility-guarantee-lap}
If all the ranking preference profiles are generated by
the \textsf{Mallows} model with the dispersion parameter $\theta$
and a ground truth ranking,
the estimation of $\texttt{cmp}(\textbf{L})$
by \texttt{LDP-KwikSort:Lap} produces the error $< 6\mu$
with probability at least $1-2^{-6\mu}$,
where
$\mu = 2 \binom{m}{2} \exp
\left(-(1-e^{-\frac{\epsilon}{2K}})^{2} K\cdot\frac{{\theta^*}^2 n}{m (m-1)}\right)$
and $\theta^*=\frac{\theta}{2-\theta}$.
\end{theorem}

Based on~\Cref{thr:utility-guarantee-lap},
it is found that the estimation error
of \texttt{LDP-KwikSort:Lap}
also depends on the expectation $\mu$,
and a smaller $\mu$ indicates a smaller error.
We observe that when the privacy budget $\epsilon$,
the number of agents $n$ and alternatives $m$,
and the dispersion parameter $\theta$ are fixed,
$\mu$ is influenced by the function
$g(K)=(1-e^{-\frac{\epsilon}{2K}})^{2} K$,
and it gets the approximate minimum value
when the number of queries $K<\frac{\epsilon}{2}$.
Other observations are the same as \texttt{LDP-KwikSort:RR},
and in the experiments
we will verify these theoretical observations.

\subsection{Computational Complexity}

As shown in~\Cref{tab:ldp-ra-complexity},
we analyze the computational complexity of the proposed protocol
with $K$ queries in terms of the following aspects:

\begin{table*}[h] \centering
\small
\caption{Computational complexity of \texttt{LDP-KwikSort} protocol}
\label{tab:ldp-ra-complexity}
  \begin{tabular}{p{2.5cm}<{\centering} p{2cm}<{\centering} p{4.5cm}
  <{\centering} p{4.5cm}<{\centering}}
    \hline
    \multicolumn{2}{c}{\diagbox[width=15em,trim=l]{Complexity}{Algorithm}} &
    \tabincell{c}{\textbf{Local perturbation}} & \textbf{Post-processing} \\
    \hline
    \multirow{2}{*}{Time}
    & Agent & $O(K)$ & \diagbox[dir=SW]{}{}  \\ \cline{2-2}
    & Curator & $O(nK)$ & $O(nK)+O(m^2)+O(m\log m)$  \\ \hline
    \multirow{2}{*}{Space}
    & Agent & $O(\log K)$ & \diagbox[dir=SW]{}{}  \\ \cline{2-2}
    & Curator & $O(n\log K)$ & $O(n\log K)$  \\ \hline
    \multirow{2}{*}{Communication}
    & Agent & $O(\log K)$ & $O(\log K)$  \\ \cline{2-2}
    & Curator & $O(n\log K)$ & $O(n\log K)$  \\ \hline
  \end{tabular}
\end{table*}

\paragraph{\textbf{Running time.}}
For each agent,
in the local perturbation algorithm
of solution \texttt{LDP-KwikSort:RR},
the basic operations are generating random number $K$ times
and number comparison $K$ times,
thus the complexity is $O(K)$.
For that of \texttt{LDP-KwikSort:Lap},
since the basic operations are generating random number $K$ times
and addition $K$ times,
the complexity is also $O(K)$.

On the curator's side,
in the local perturbation algorithms of two solutions,
it is required to generate $K$ queries to each agent,
thus the complexity is $O(nK)$.
For the post-processing algorithm,
the complexities of the classification phase
and the estimation \& computation phase are $O(nK)$
and $O(\binom{m}{2})\approx O(m^2)$,
respectively.
The execution of \texttt{KwikSort} algorithm
consumes $O(m\log m)$.
Thus,
the total time complexity of the curator's operation
is $O(2nK)+O(m^2)+O(m\log m)$. \\

\paragraph{\textbf{Processing memory.}}
For each agent,
the main processing memory lies in the $K$ queries,
which consumes $O(\log K)$ bits.
For the curator,
as it maintains
$n$ times the sums of at most $2K$ bits for each agent,
the processing memory is $O(2n\log K)$ bits.

\paragraph{\textbf{Communication cost.}}
The number of queries directly impacts the communication cost,
which includes $O(2\log K)$ for each agent
and $O(2n\log K)$ for the curator.

\section{Empirical Analysis}
\label{sec-ldp-ra-empirical-analysis}

In this section,
we present the performance evaluation of
the \texttt{LDP-KwikSort} protocol
and the competitors on both real and synthetic datasets.
\Cref{sec-experiment-settings}
introduces the experiment settings,
and~\Cref{sec-ldp-ra-results-of-K}-\ref{sec-ldp-ra-results-of-theta}
investigate how will the parameters
(number of queries $K$,
privacy budget $\epsilon$,
number of agent $n$ and alternative $m$,
and dispersion parameter $\theta$)
impact on performance.
\Cref{sec-ldp-ra-time-cost}
compares the time costs of involved solutions,
and~\Cref{sec-ldp-ra-discussion}
provides a summarized discussion.

\subsection{Experiment Settings}
\label{sec-experiment-settings}

\subsubsection{Competitors}

\paragraph{\textbf{KwikSort.}}
We adopt the approximate rank aggregation algorithm \texttt{KwikSort}
described in~\cite{AilonCN08JACM} as the non-private solution.
Since each agent responds to each query
with the true answer to the curator,
and the latter releases the aggregate ranking
without adding any noise,
the performances of \texttt{KwikSort} can be seen as
the empirical error lower bound in our experiments.

\paragraph{\textbf{DP-KwikSort.}}
We adopt the differentially private algorithm
\texttt{DP-KwikSort}~\cite{HayEM17SDM}
as the solution under the central model of DP.
When executing this algorithm,
each agent responds to each query
with the true answer to the curator,
and the latter adopts the Laplace mechanism
to introduce noises during the rank aggregation.
Specifically,
the algorithm adds noises into the results of the comparison function:
$\widetilde{\texttt{cmp}}_{\textbf{L}}(a_{j},a_{l})
= \texttt{cmp}_{\textbf{L}}(a_{j},a_{l}) + X_{dp}$
where $X_{dp} \sim \textsf{Lap}(1/\epsilon')$
and $\epsilon' = \epsilon / ((m-1)\log m)$.
Besides,
the number of queries in \texttt{DP-KwikSort}
is defaulted as $K=\binom{m}{2}$.

\subsubsection{Datasets and Configuration}

The experiments are conducted
on synthetic datasets and three real-world datasets
(\textsf{TurkDots},
\textsf{TurkPuzzle} and \textsf{SUSHI}).
Among them,
datasets \textsf{TurkDots} and \textsf{TurkPuzzle}~\cite{MaoPC13AAAI}
were collected from the crowdsourcing marketplace
\emph{Amazon Mechanical Turk},
which respectively contains $n=795$ and $n=793$
agents' full ranking preference profiles
over $m=4$ alternatives.
Dataset \textsf{SUSHI}~\cite{Kamishima03KDD} contains the full rankings
of $m=10$ types of sushi from $n=5000$ agents in a questionnaire survey.
The synthetic datasets were generated by the \textsf{Mallows} model with R package
\texttt{PerMallows} $1.13$~\cite{PerMallows16}.

The involved protocols and algorithms are implemented
in Python $2.7$ based on the package \texttt{pwlistorder} $0.1$
\cite{pwlistorder17},
and executed on an Intel Core i$5-3210$M $2.50$GHz machine with $6$GB memory.
In each experiment,
the protocols and algorithms were tested $30$ times,
and their mean score of the adopted utility metrics was reported.

\subsubsection{Utility Metrics}

\paragraph{\textbf{Error rate.}}
In experiments we coin the measurement \emph{error rate}
to reflect how accurate the estimated version
of aggregate pairwise comparison profile $\widehat{\texttt{cmp}}(\textbf{L})$
agrees with the ground truth $\texttt{cmp}(\textbf{L})$.
Specifically,
for all possible alternative pairs $a_{j}$ and $a_{l}$ where $j<l$,
\begin{equation}
\begin{aligned}
Error Rate =
\frac{|\{ \texttt{cmp}_{\textbf{L}}(a_{j},a_{l}) > 0 \
and \ \widehat{\texttt{cmp}}_{\textbf{L}}(a_{j},a_{l}) < 0 \}|}{\binom{m}{2}}\\ +
\frac{|\{ \texttt{cmp}_{\textbf{L}}(a_{j},a_{l}) < 0 \
and \ \widehat{\texttt{cmp}}_{\textbf{L}}(a_{j},a_{l}) > 0 \}|}{\binom{m}{2}}.
\end{aligned}
\end{equation} \\

\paragraph{\textbf{Average Kendall tau distance.}}
As mentioned before,
we measure the achieved accuracy
of the aggregate ranking $\widetilde{L}_{\mathcal{P}}$
by adopting the average \emph{Kendall tau} distance
$\overline{\textbf{K}} (\widetilde{L}_{\mathcal{P}},\textbf{L}) = \frac{1}{n} \sum_{i \in N}
\textbf{K} (\widetilde{L}_{\mathcal{P}},L_{i})$.
Furthermore,
for the convenience of comparison with different $m$,
we normalize the values by $m(m-1)/2$.

\subsection{The Impact of Query Amount}
\label{sec-ldp-ra-results-of-K}

The number of queries determines the amount
of provided information
about each agent's private ranking preference profile.
Intuitively,
a relatively large $K$ helps the curator to generate
an approximate optimal aggregate ranking.
However,
for \texttt{LDP-KwikSort} protocol,
since the overall privacy budget $\epsilon$
of each agent will be split into $K$ parts
for reporting each query,
a large number of queries will lead to more noises per answer,
which may make the collected information chaotic.
According to the theoretical analysis
in~\Cref{sec-ldp-ra-utility-guarantee},
we observe that two solutions of \texttt{LDP-KwikSort}
can get the approximate minimum error of estimation
around $K=\frac{\epsilon}{2}$.
To verify this conclusion,
we ran \texttt{LDP-KwikSort:RR} and \texttt{LDP-KwikSort:Lap}
on three real-world datasets and three synthetic datasets,
and set the privacy budget $\epsilon \in \{1.0,2.0,4.0,10.0\}$
as well as varying the number of queries
$K \in \{1,2,3,4,m+1,\binom{m}{2}\}$,
for observing how the parameter $K$ impacts on
the performance of \texttt{LDP-KwikSort}
under the utility metrics such as the error rate
and the average Kendall tau distance.
The results are shown in~\Cref{fig:varying-k-rr} and~\Cref{fig:varying-k-lap}.

The results demonstrate that these two solutions
of \texttt{LDP-KwikSort} protocol can achieve
the approximate minimum error of estimation
around $K=\frac{\epsilon}{2}$.
For instance,
when $\epsilon=1.0,2.0$,
with  the increasing number of queries,
we observe that each solution gets a larger error
and when $K=1$ the error is the global minimum.
As for $\epsilon=4.0,10.0$,
each solution achieves the minimum error
when $K=\frac{\epsilon}{2}$.
We also observe that
this property is more obvious
on the synthetic datasets,
which is due to the theoretical conclusion
being made on the assumption of the \textsf{Mallows} model.
Thus,
in the following experiments,
we relate the choice of $K$ to the privacy budget $\epsilon$,
and select the appropriate $K$
by comparing the values of error functions $g(K)$
at the integers which are less or large than $\frac{\epsilon}{2}$.
This strategy will help us to tune the solutions
to their optimal performance.

\begin{figure*}[p]
\centering
\subfigtopskip=2pt
\subfigbottomskip=1pt
\subfigcapskip=-10pt
\includegraphics[width = 0.25\linewidth]{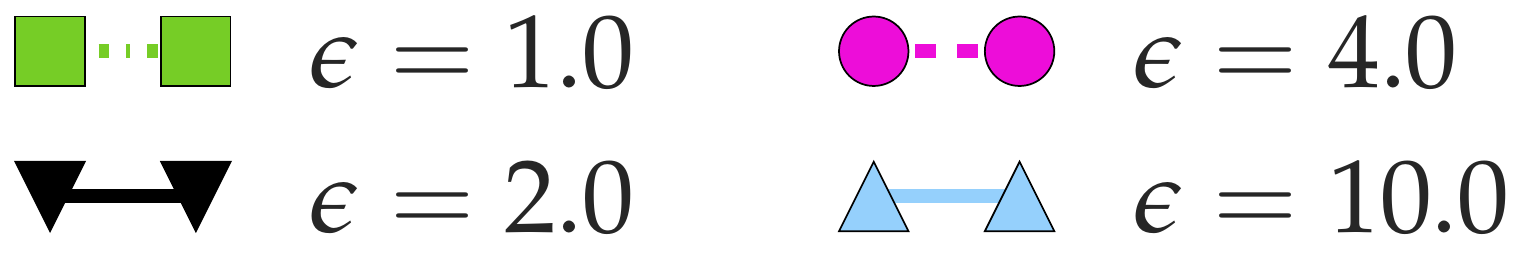}\\
\vspace{0.2cm}
\subfigure[]{\label{subfigure:fig:varying-k-rr:a}
    \includegraphics[width = 0.31\linewidth]{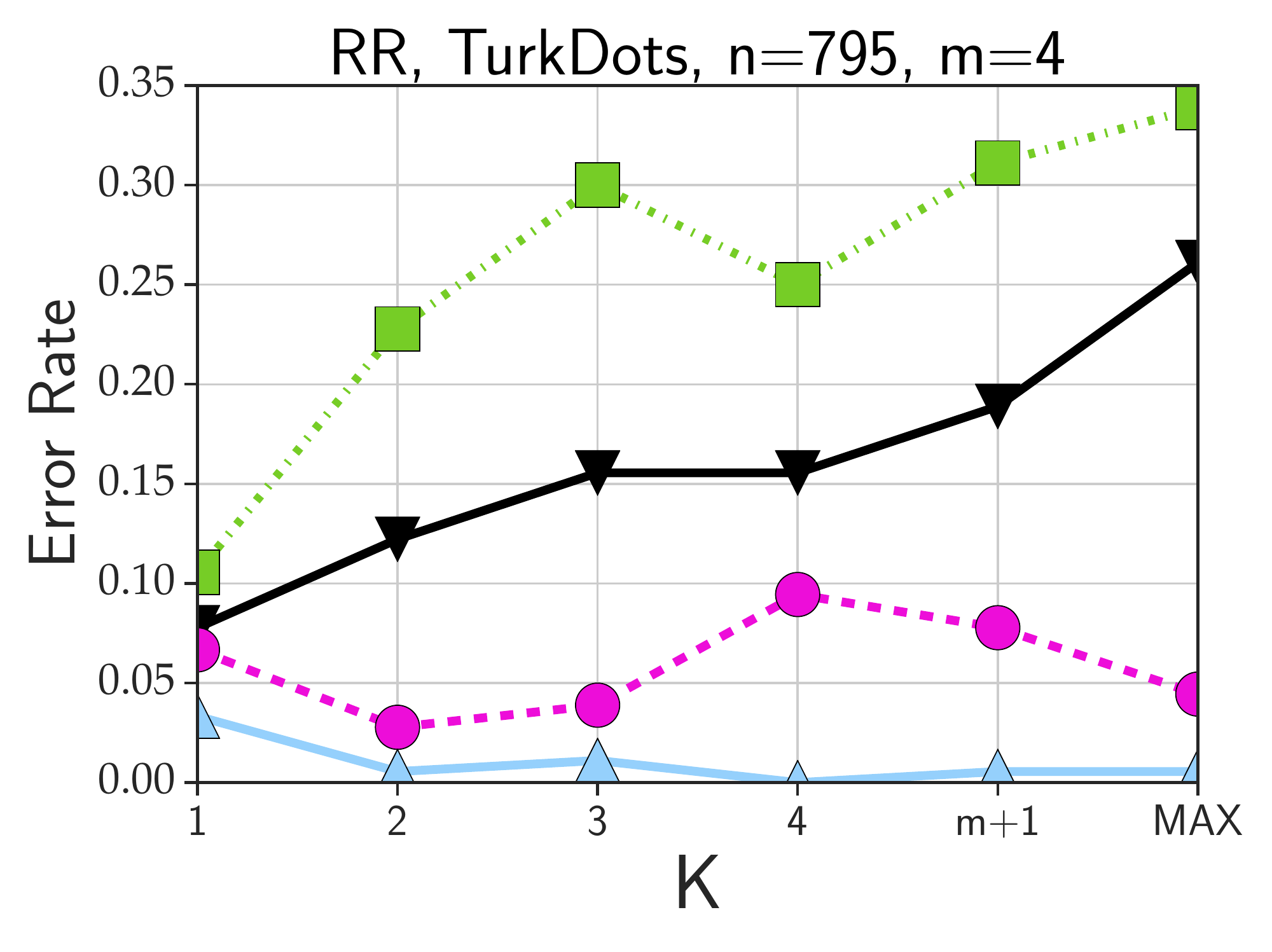}}\
\subfigure[]{\label{subfigure:fig:varying-k-rr:b}
    \includegraphics[width = 0.31\linewidth]{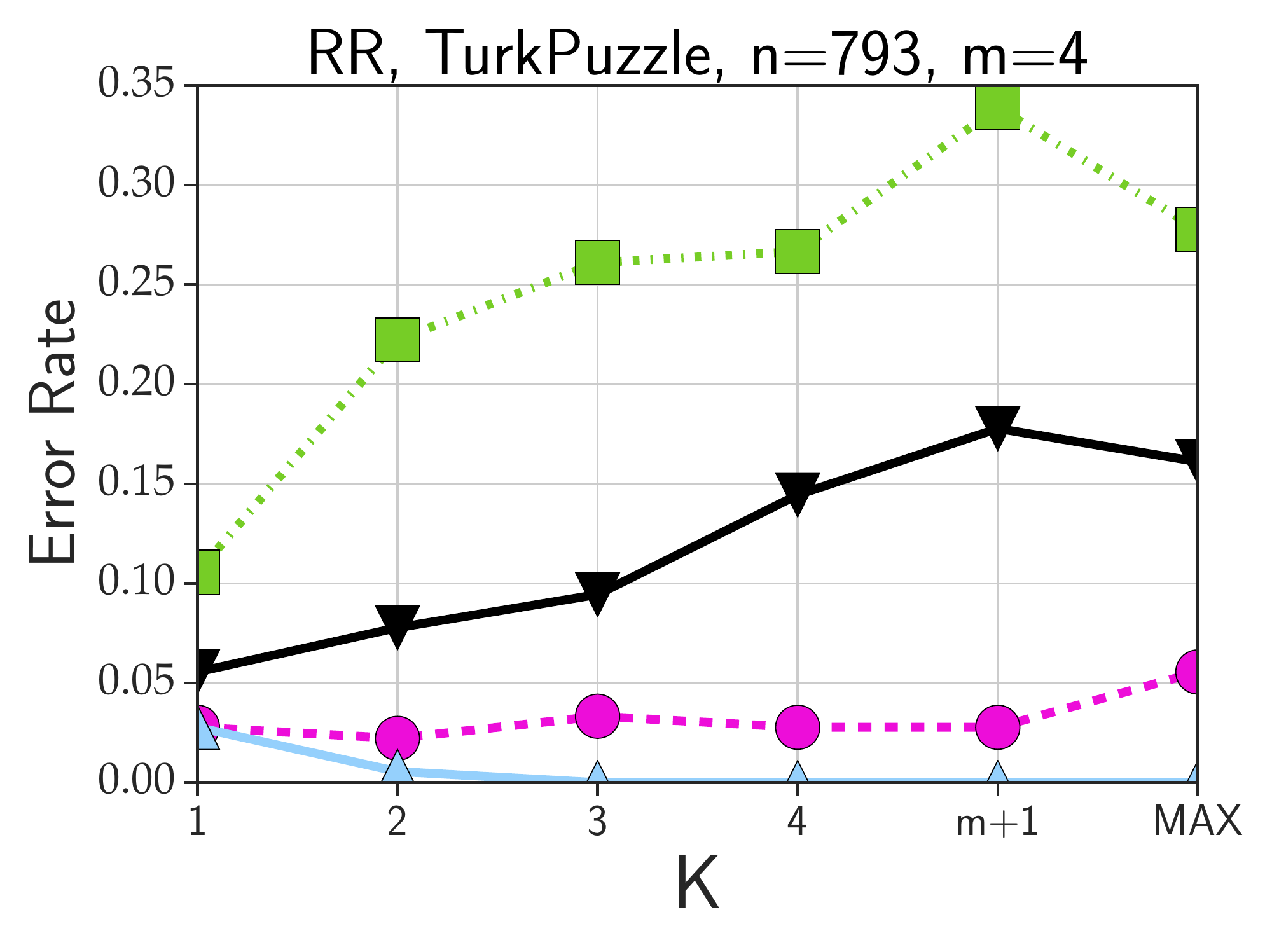}}\
\subfigure[]{\label{subfigure:fig:varying-k-rr:c}
    \includegraphics[width = 0.31\linewidth]{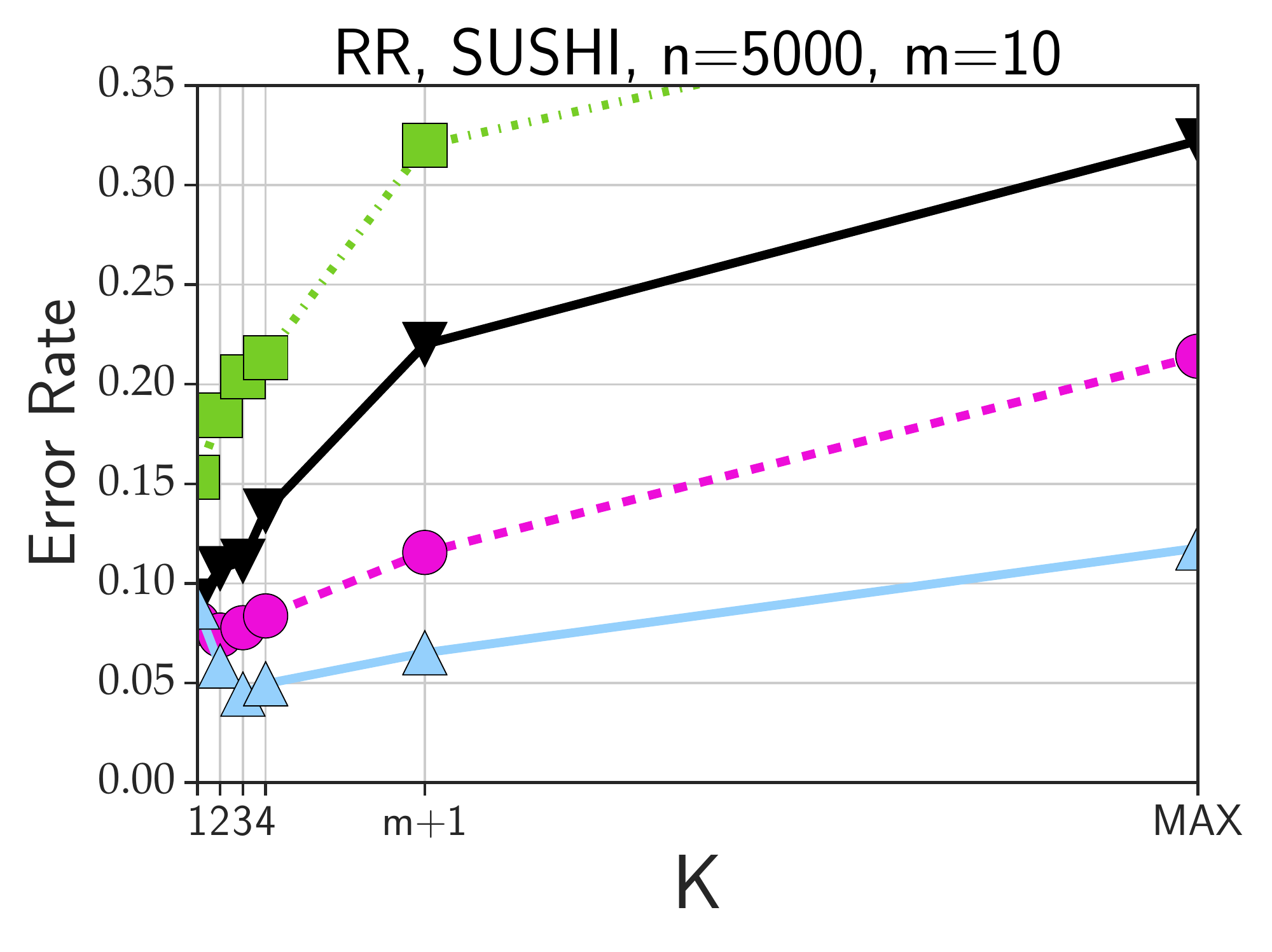}}\
\subfigure[]{\label{subfigure:fig:varying-k-rr:d}
    \includegraphics[width = 0.31\linewidth]{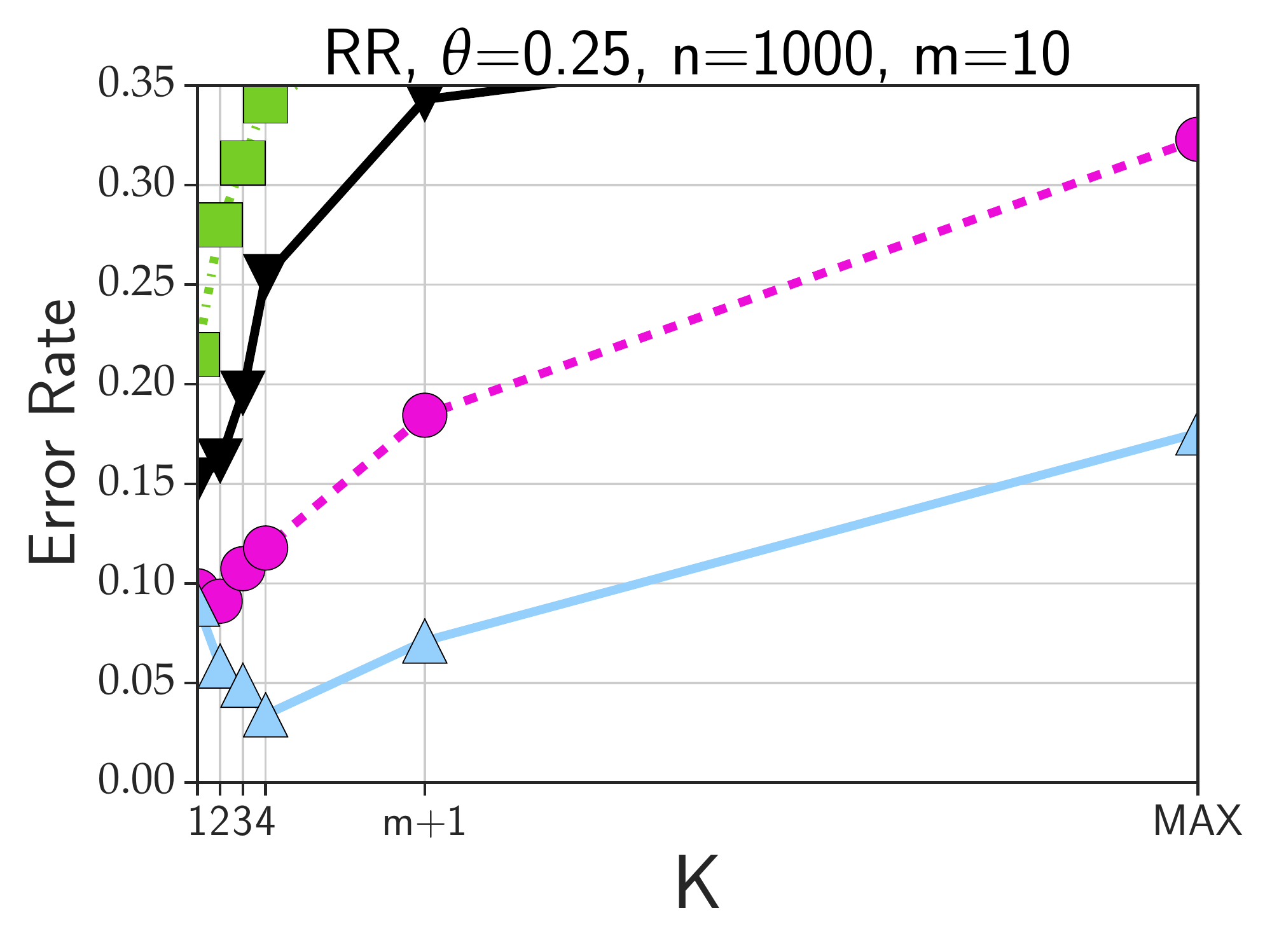}}\
\subfigure[]{\label{subfigure:fig:varying-k-rr:e}
    \includegraphics[width = 0.31\linewidth]{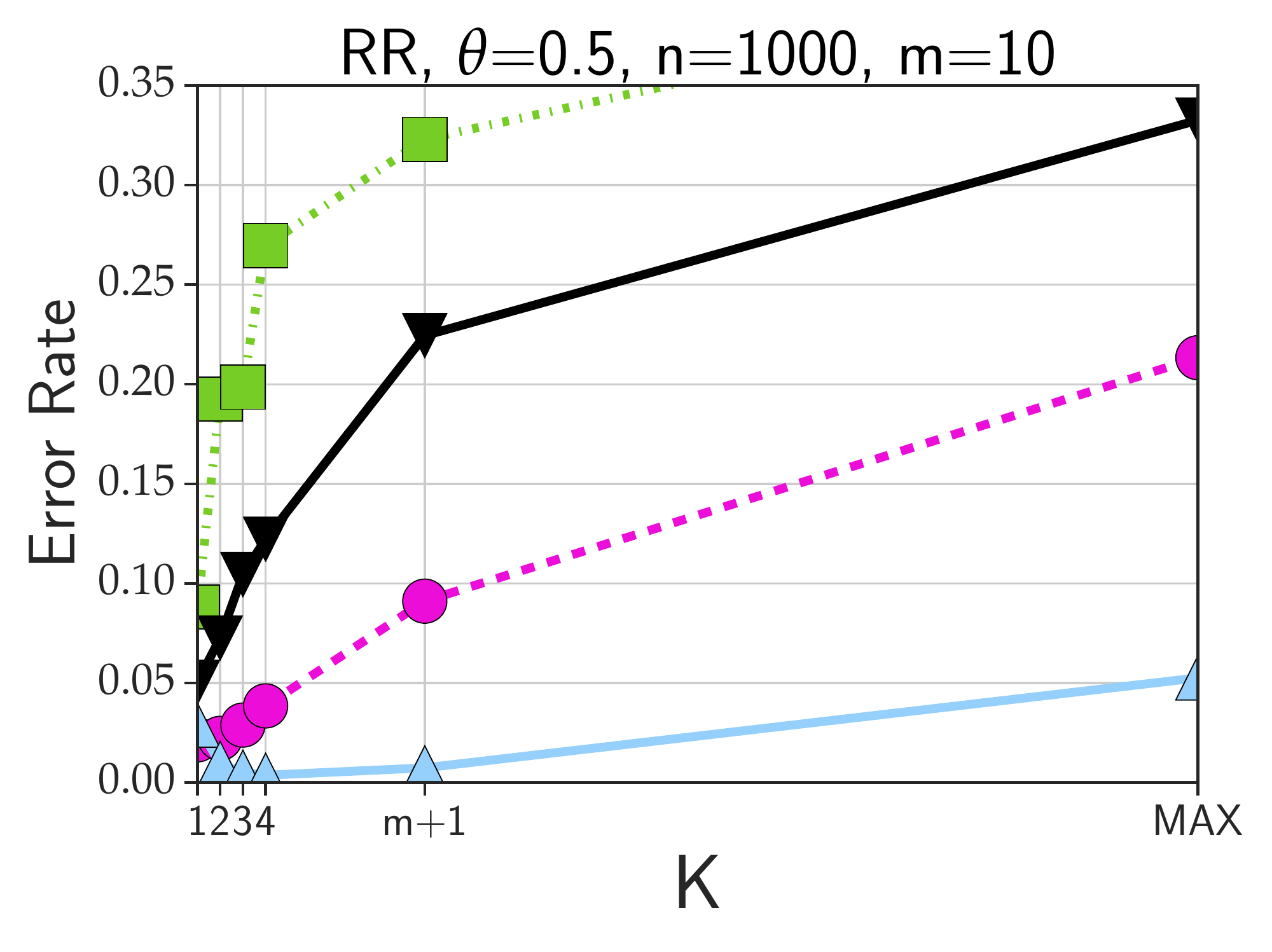}}\
\subfigure[]{\label{subfigure:fig:varying-k-rr:f}
    \includegraphics[width = 0.31\linewidth]{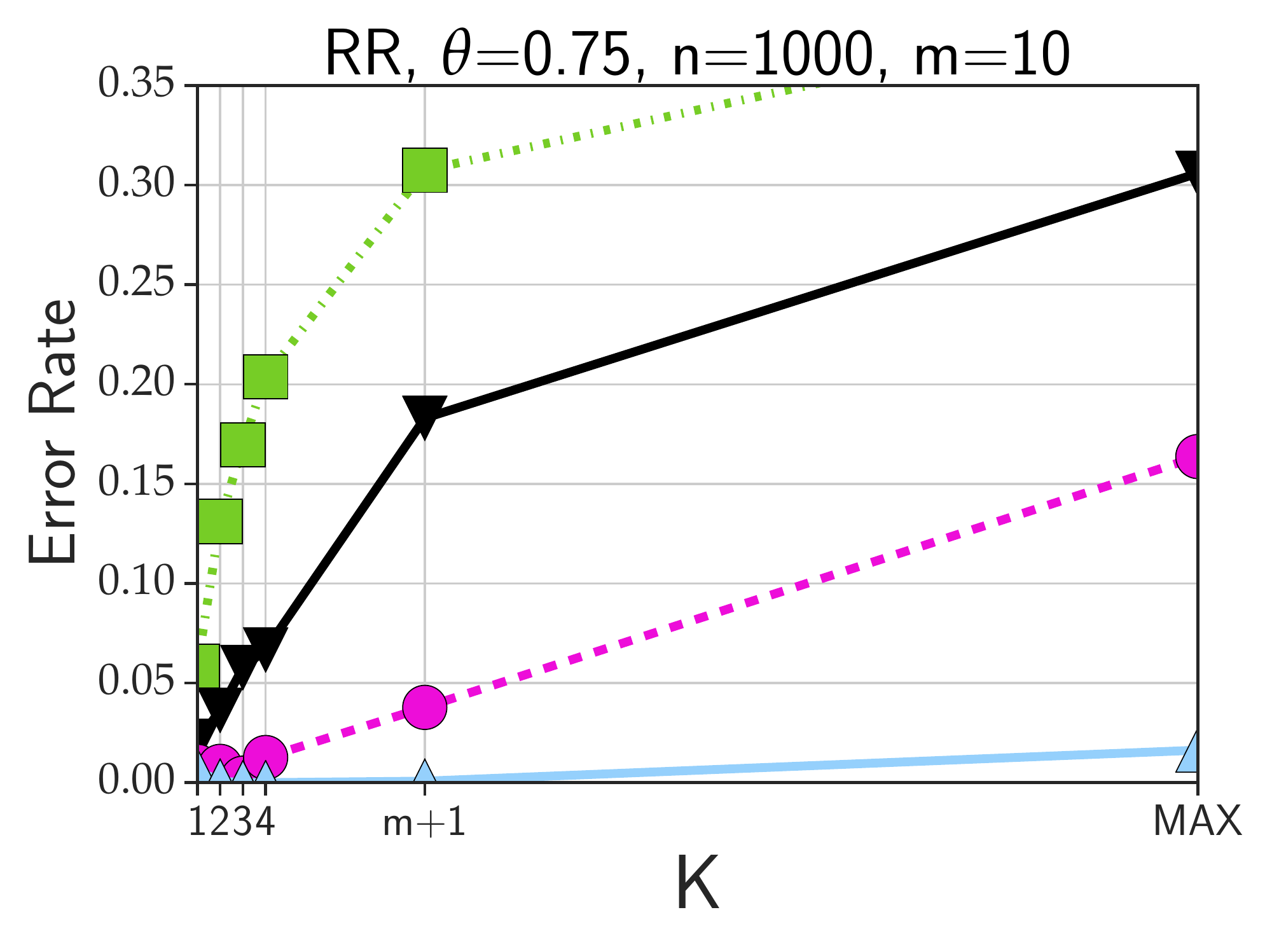}}\
\subfigure[]{\label{subfigure:fig:varying-k-rr:g}
    \includegraphics[width = 0.31\linewidth]{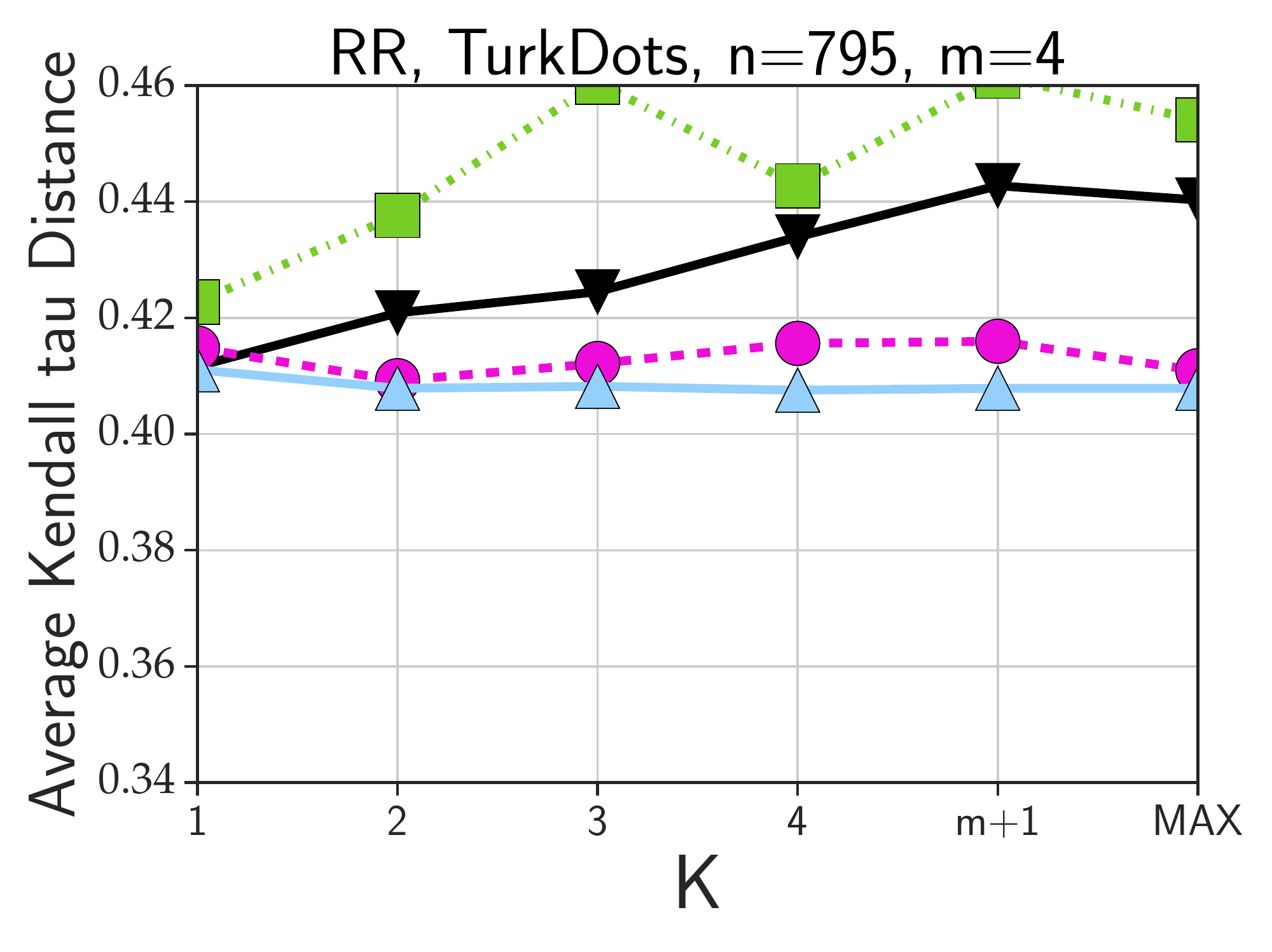}}\
\subfigure[]{\label{subfigure:fig:varying-k-rr:h}
    \includegraphics[width = 0.31\linewidth]{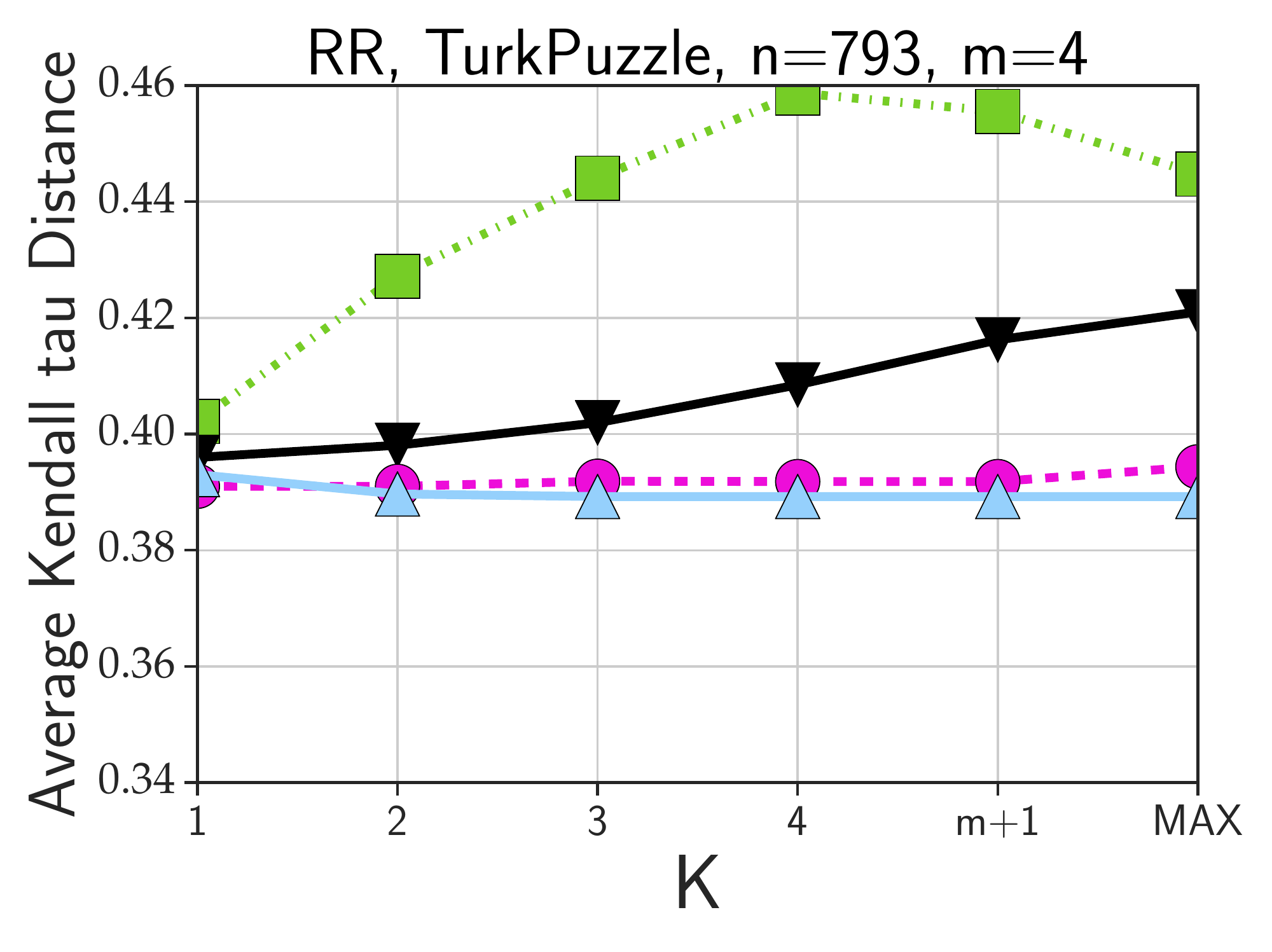}}\
\subfigure[]{\label{subfigure:fig:varying-k-rr:i}
    \includegraphics[width = 0.31\linewidth]{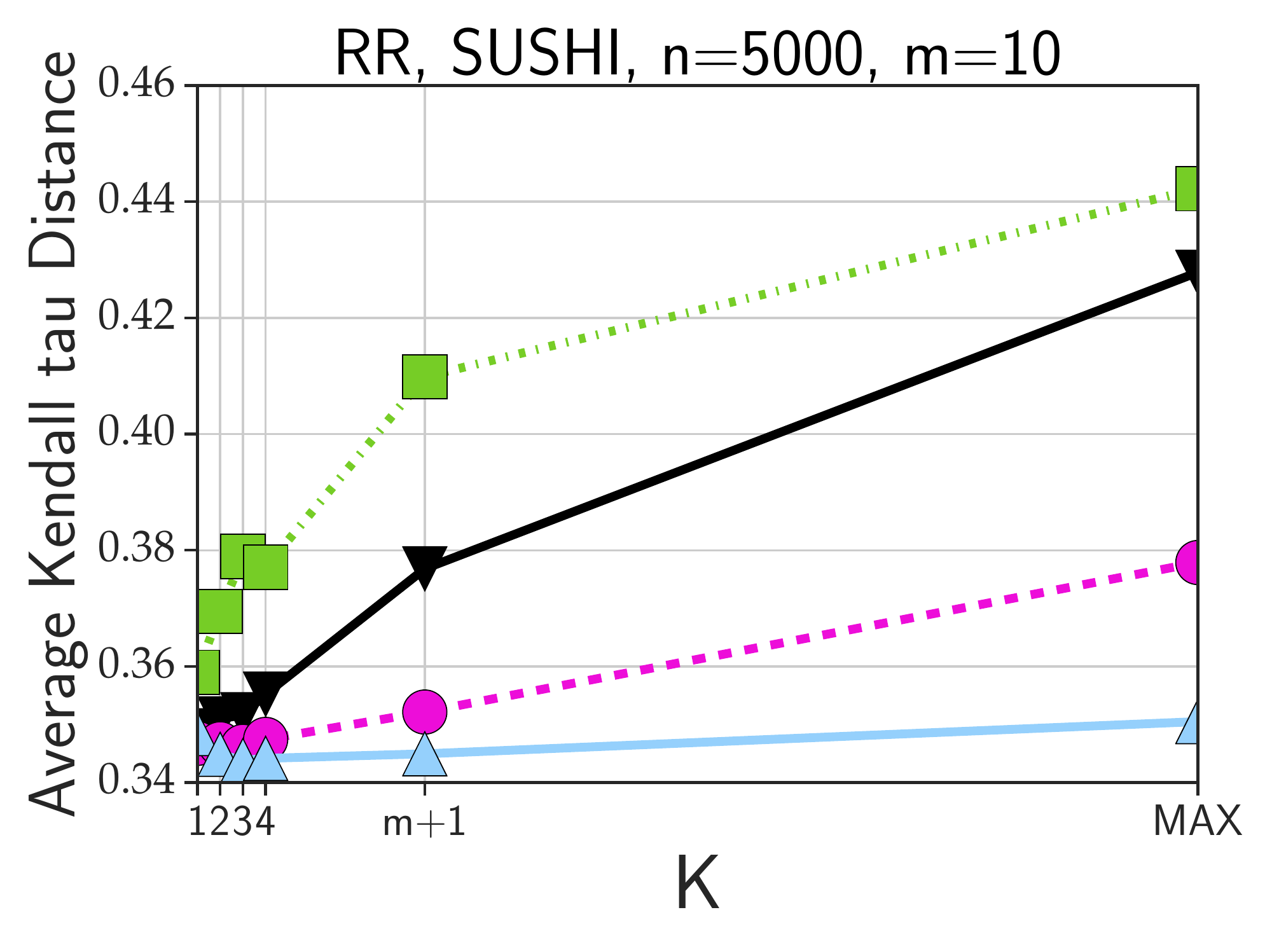}}\
\subfigure[]{\label{subfigure:fig:varying-k-rr:j}
    \includegraphics[width = 0.31\linewidth]{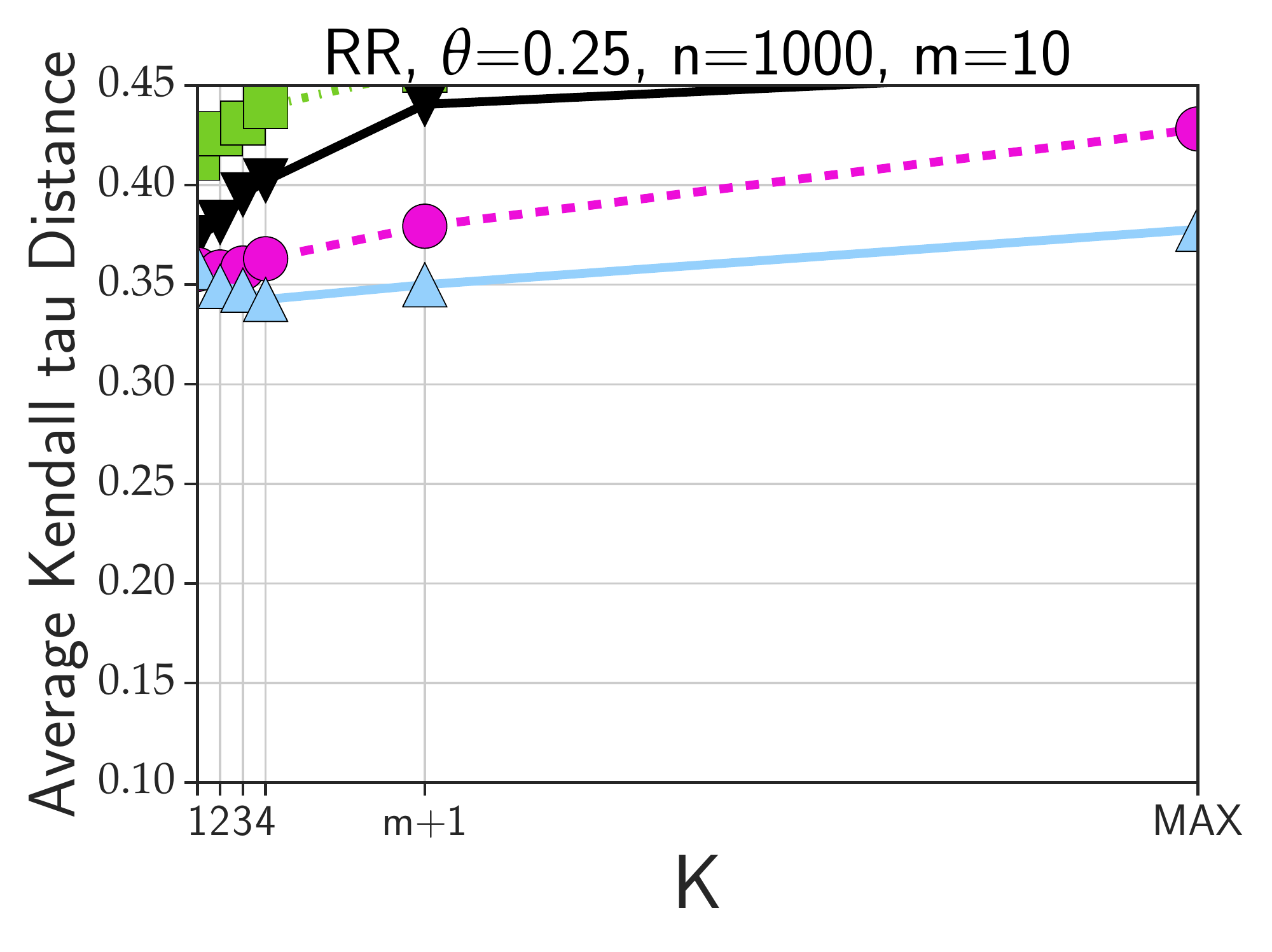}}\
\subfigure[]{\label{subfigure:fig:varying-k-rr:k}
    \includegraphics[width = 0.31\linewidth]{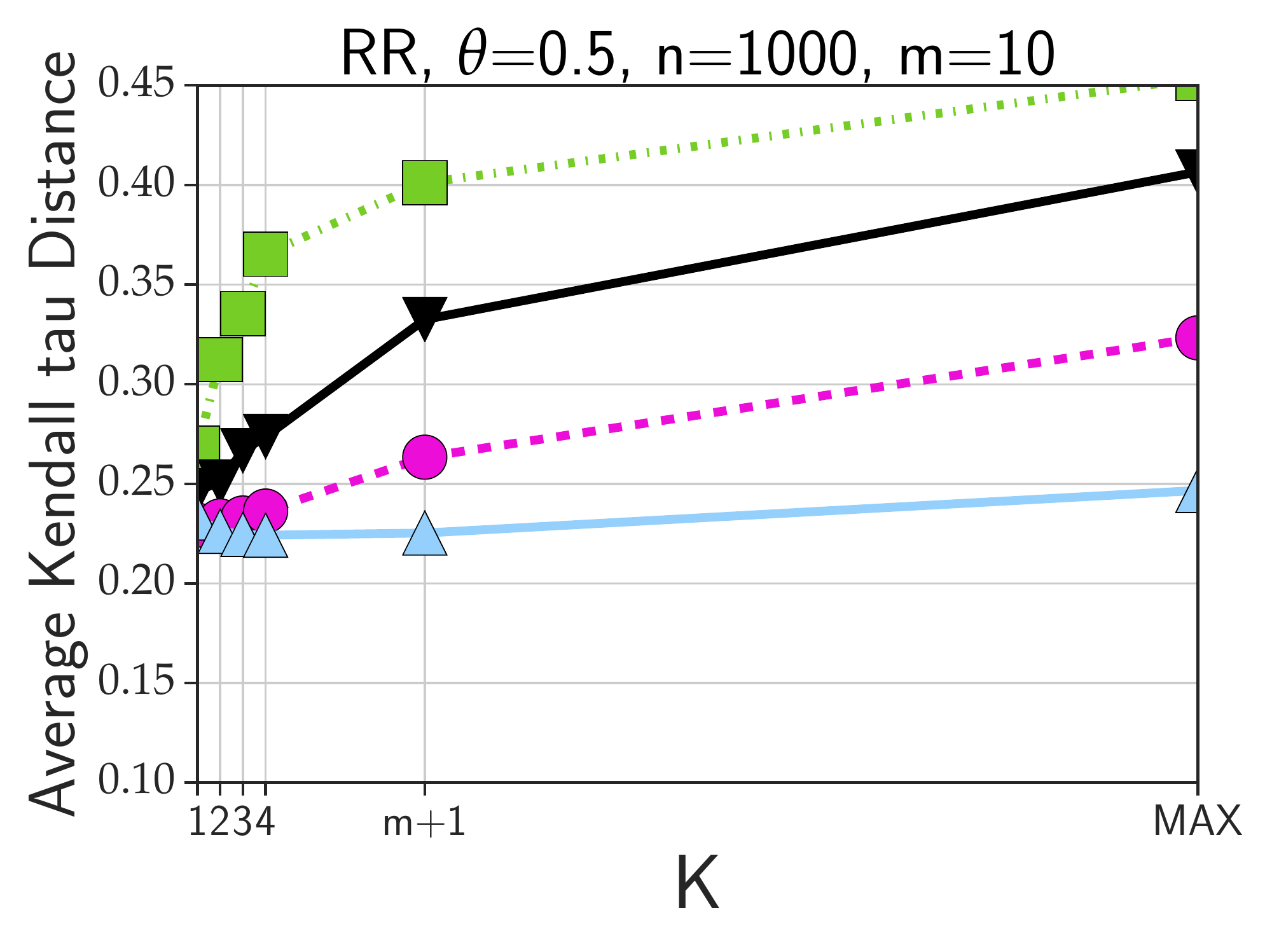}}\
\subfigure[]{\label{subfigure:fig:varying-k-rr:l}
    \includegraphics[width = 0.31\linewidth]{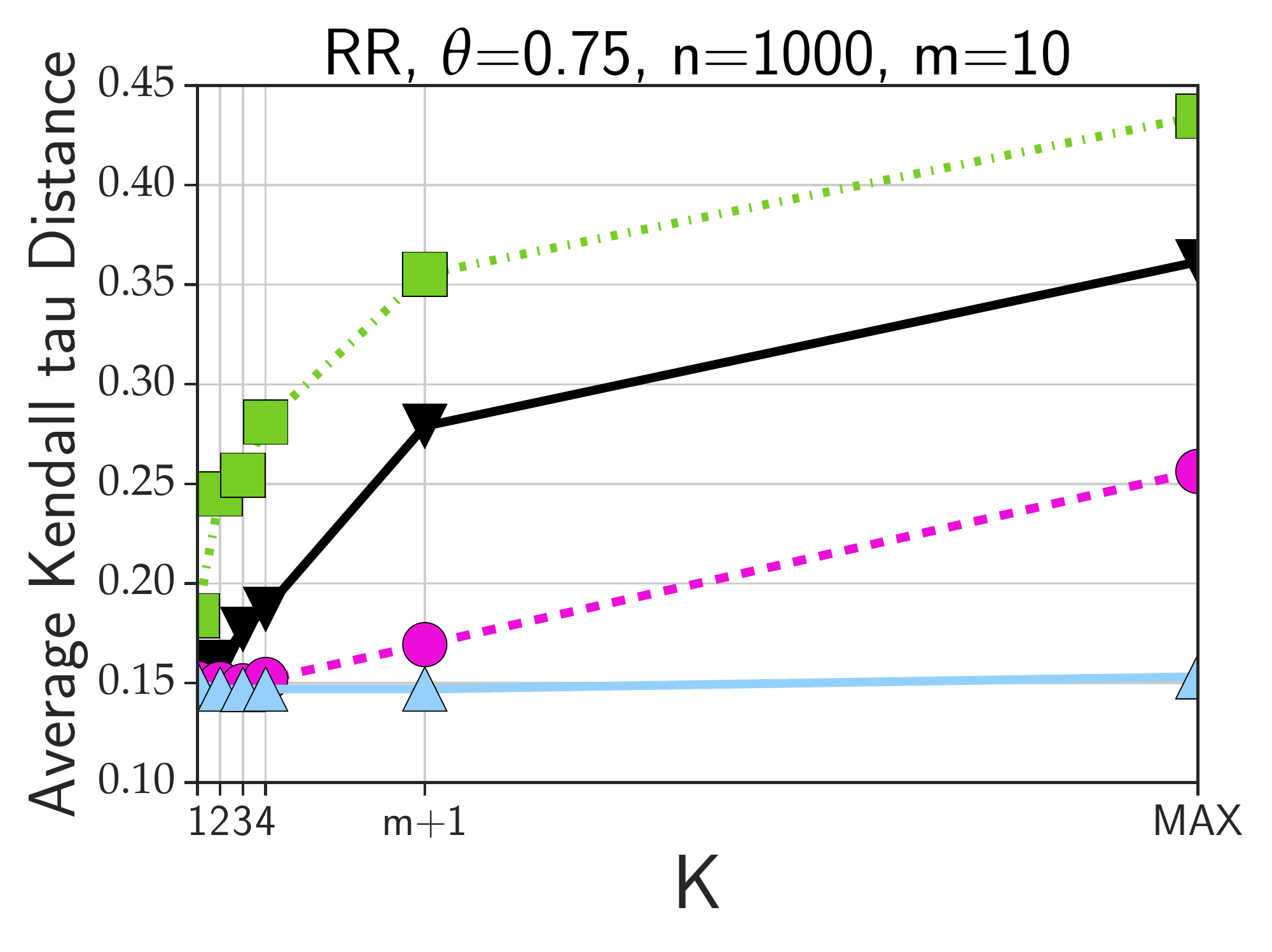}}\
\caption{Performance of \texttt{LDP-KwikSort:RR} in terms of error rate and average Kendall tau distance on the real-world and synthetic datasets,
across varying the number of queries $K$.}
\label{fig:varying-k-rr}
\end{figure*}

\begin{figure*}[p]
\centering
\subfigtopskip=2pt
\subfigbottomskip=1pt
\subfigcapskip=-10pt
\includegraphics[width = 0.25\linewidth]{./fig_legend1.pdf}\\
\vspace{0.2cm}
\subfigure[]{\label{subfigure:fig:varying-k-lap:a}
    \includegraphics[width = 0.31\linewidth]{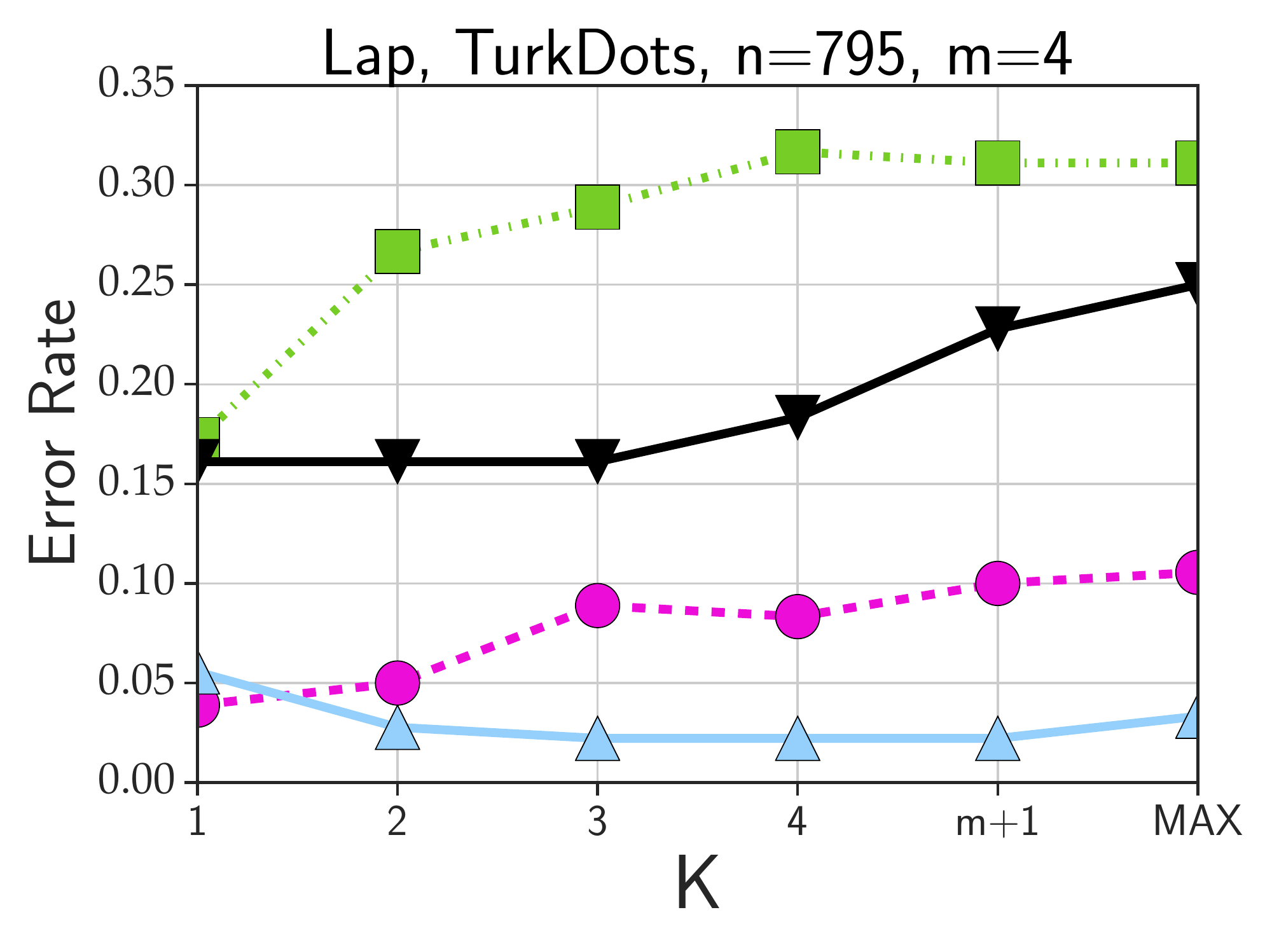}}\
\subfigure[]{\label{subfigure:fig:varying-k-lap:b}
    \includegraphics[width = 0.31\linewidth]{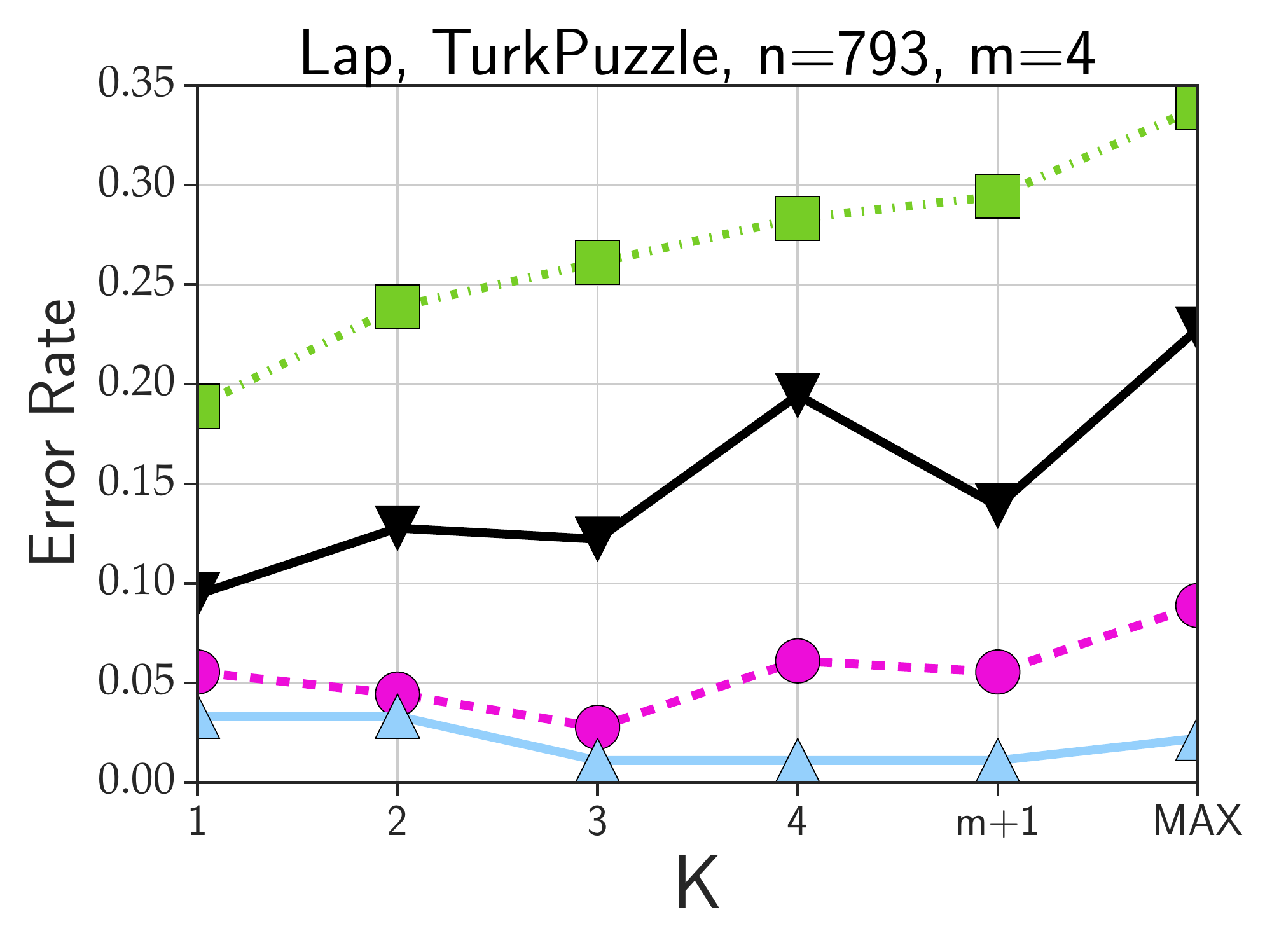}}\
\subfigure[]{\label{subfigure:fig:varying-k-lap:c}
    \includegraphics[width = 0.31\linewidth]{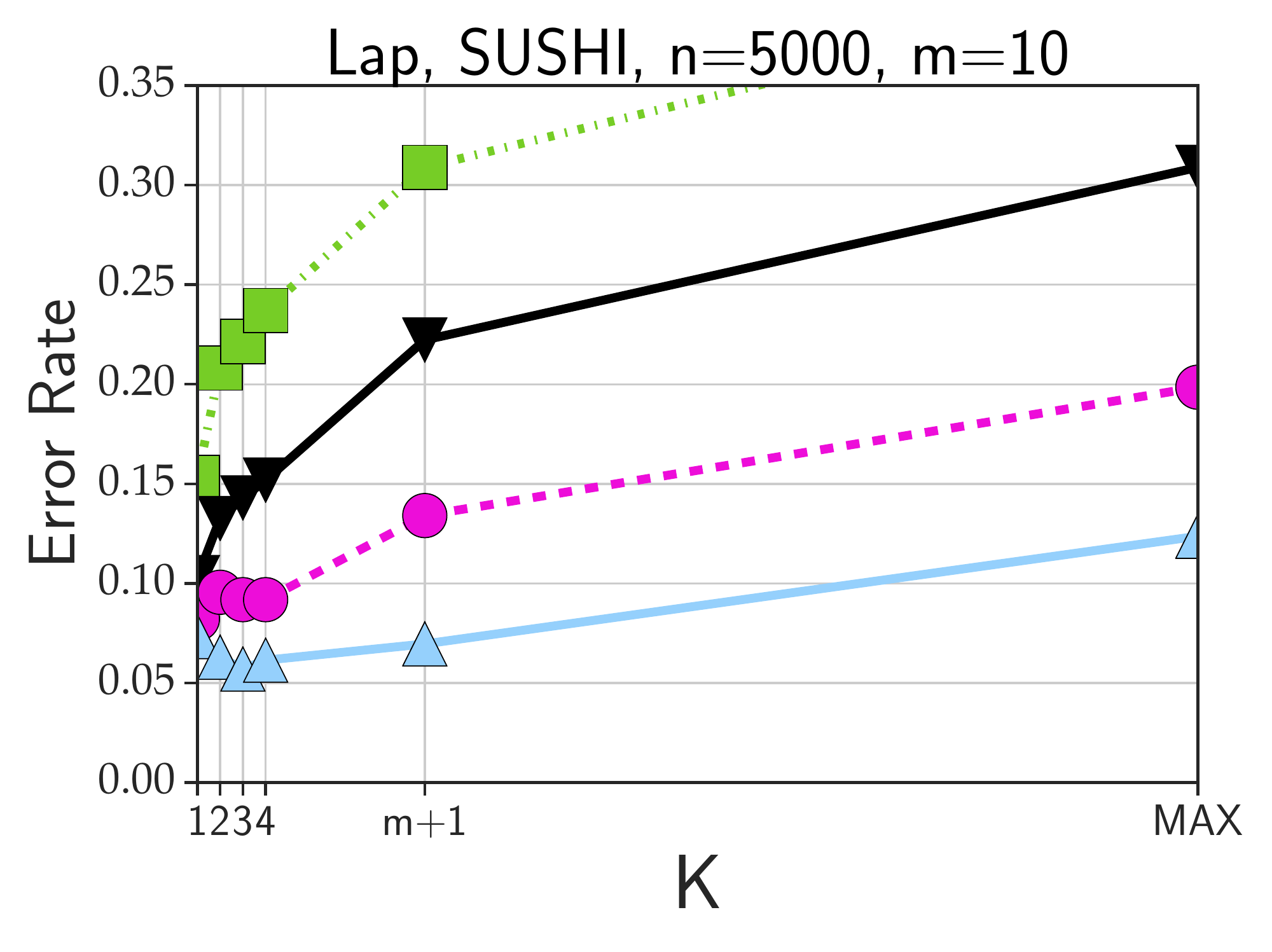}}\
\subfigure[]{\label{subfigure:fig:varying-k-lap:d}
    \includegraphics[width = 0.31\linewidth]{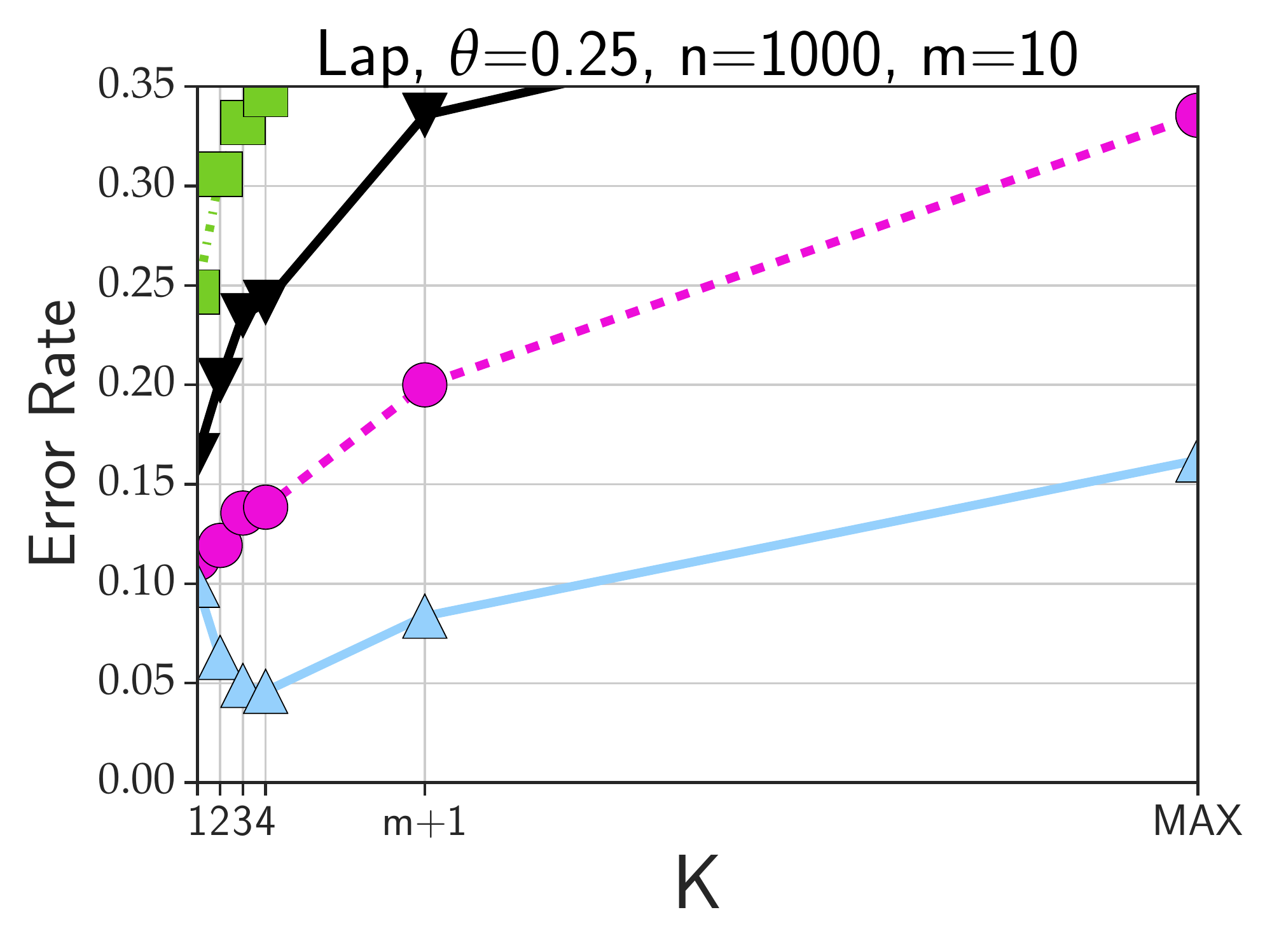}}\
\subfigure[]{\label{subfigure:fig:varying-k-lap:e}
    \includegraphics[width = 0.31\linewidth]{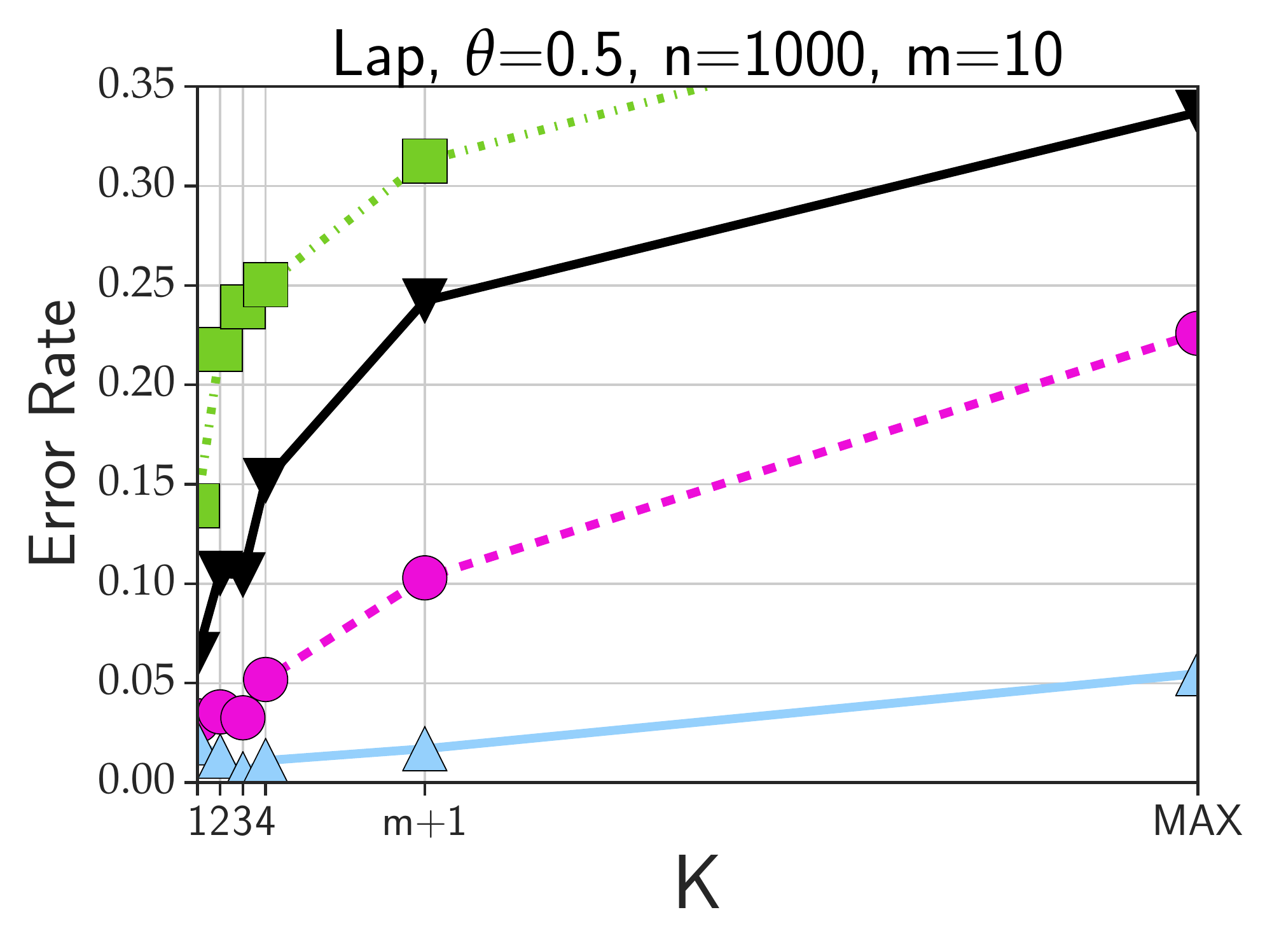}}\
\subfigure[]{\label{subfigure:fig:varying-k-lap:f}
    \includegraphics[width = 0.31\linewidth]{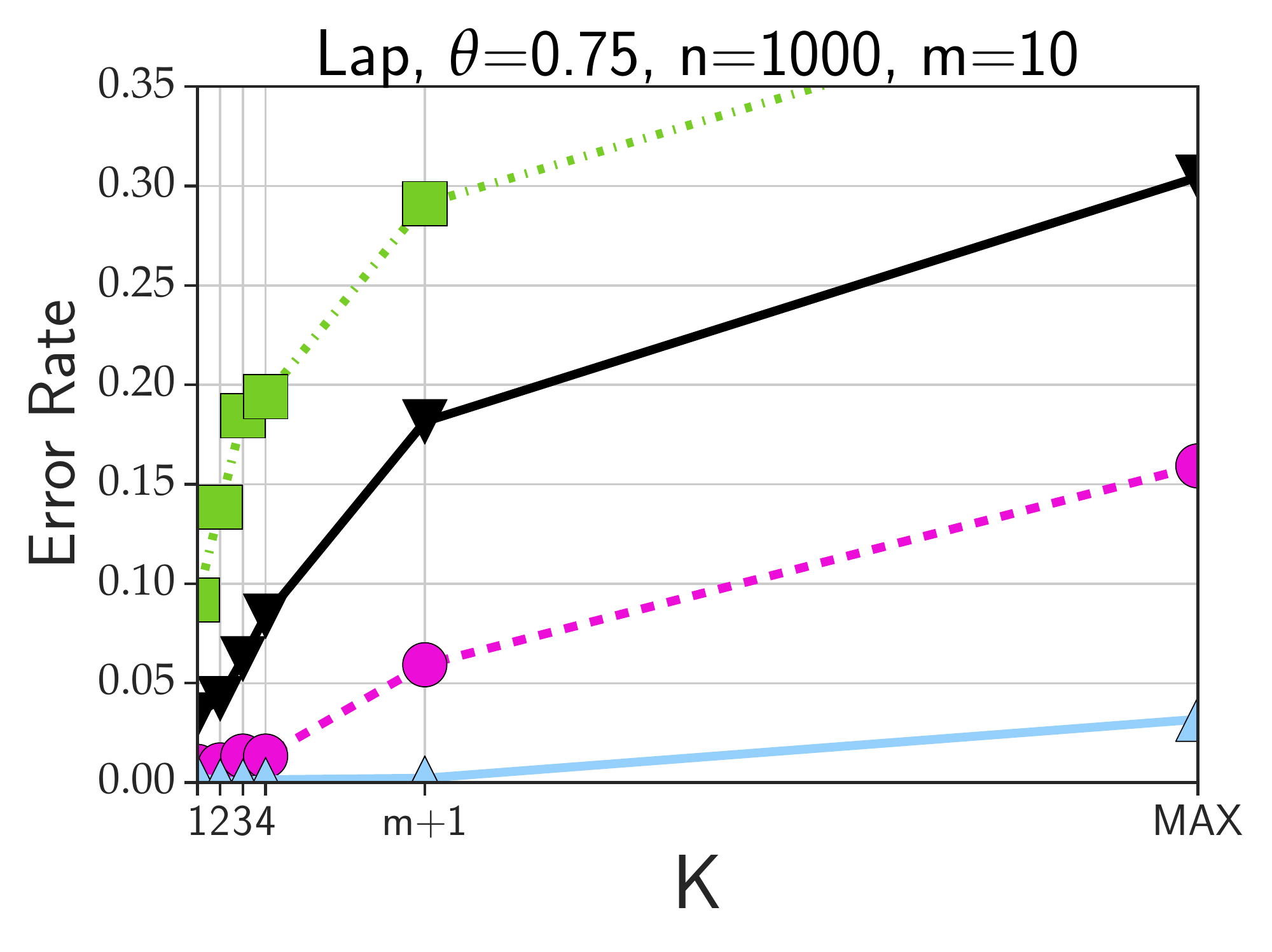}}\
\subfigure[]{\label{subfigure:fig:varying-k-lap:g}
    \includegraphics[width = 0.31\linewidth]{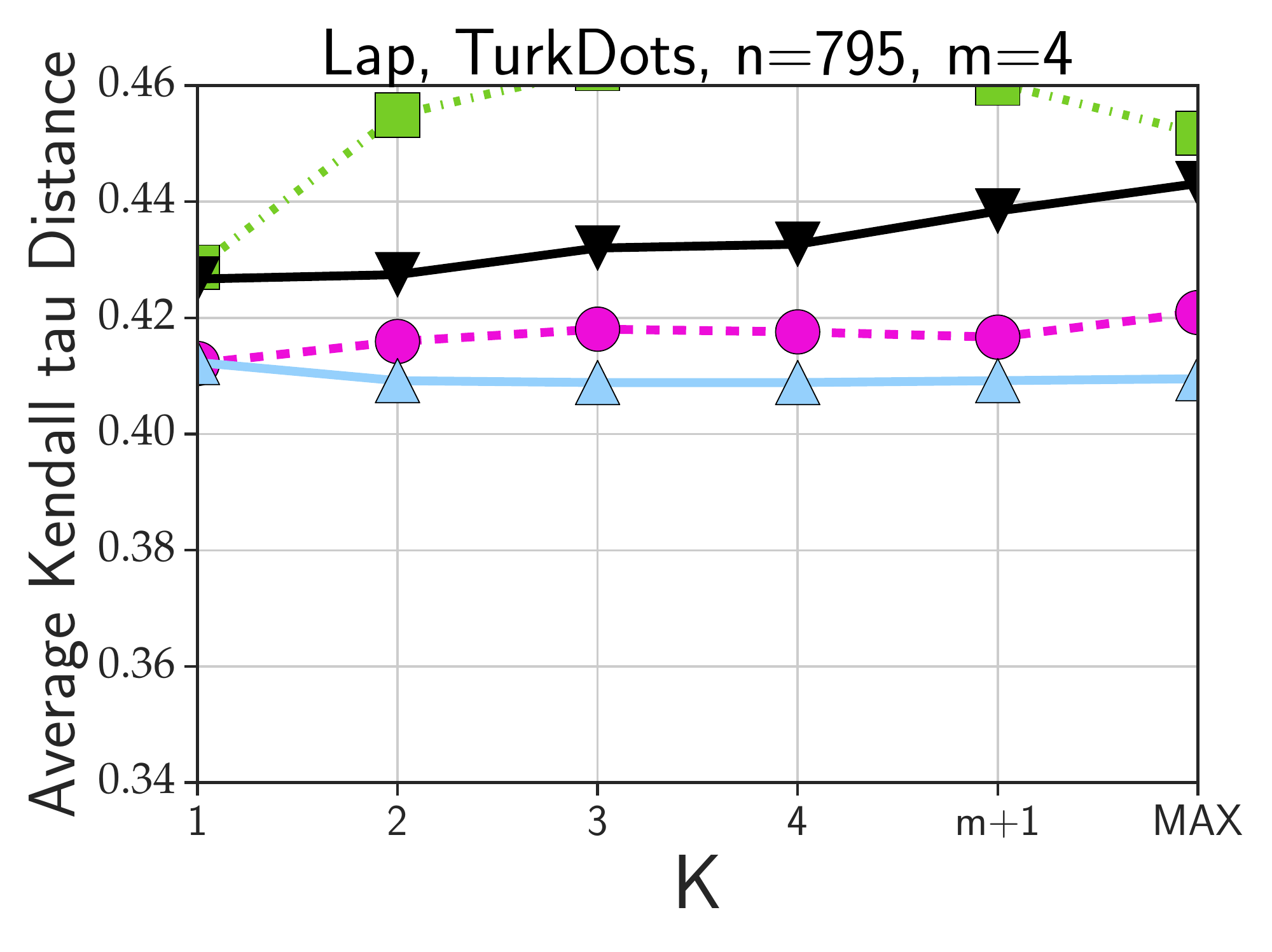}}\
\subfigure[]{\label{subfigure:fig:varying-k-lap:h}
    \includegraphics[width = 0.31\linewidth]{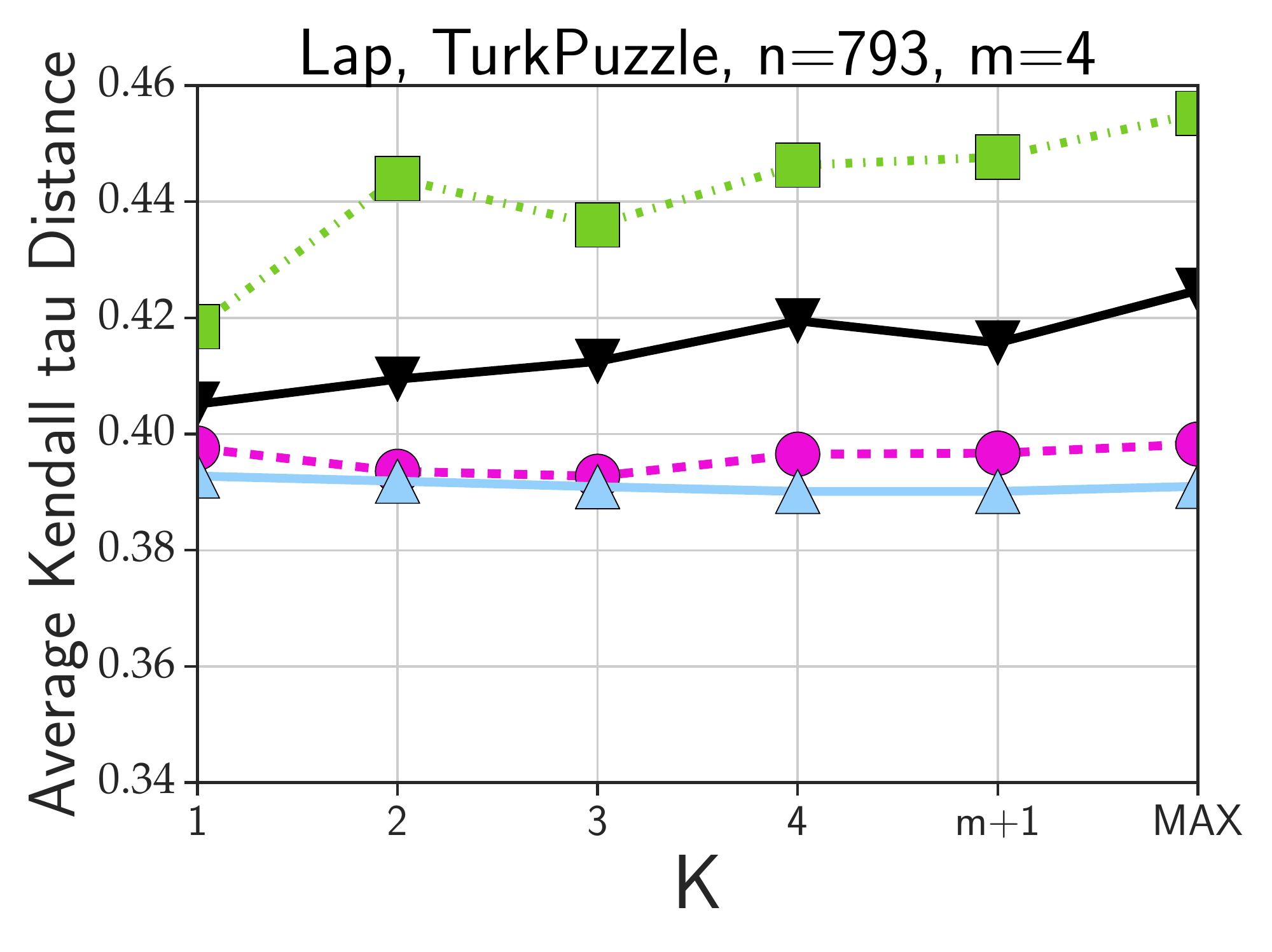}}\
\subfigure[]{\label{subfigure:fig:varying-k-lap:i}
    \includegraphics[width = 0.31\linewidth]{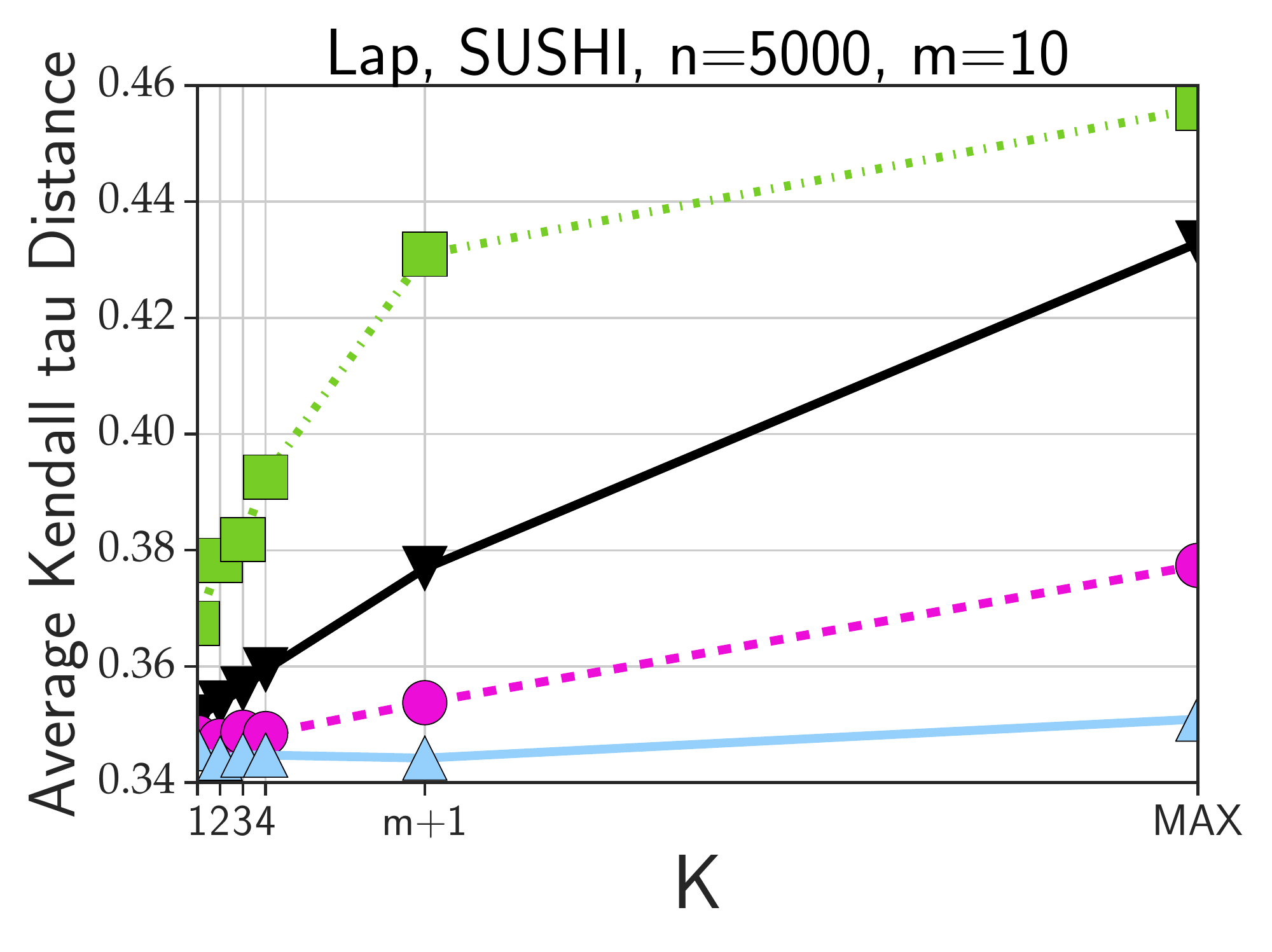}}\
\subfigure[]{\label{subfigure:fig:varying-k-lap:j}
    \includegraphics[width = 0.31\linewidth]{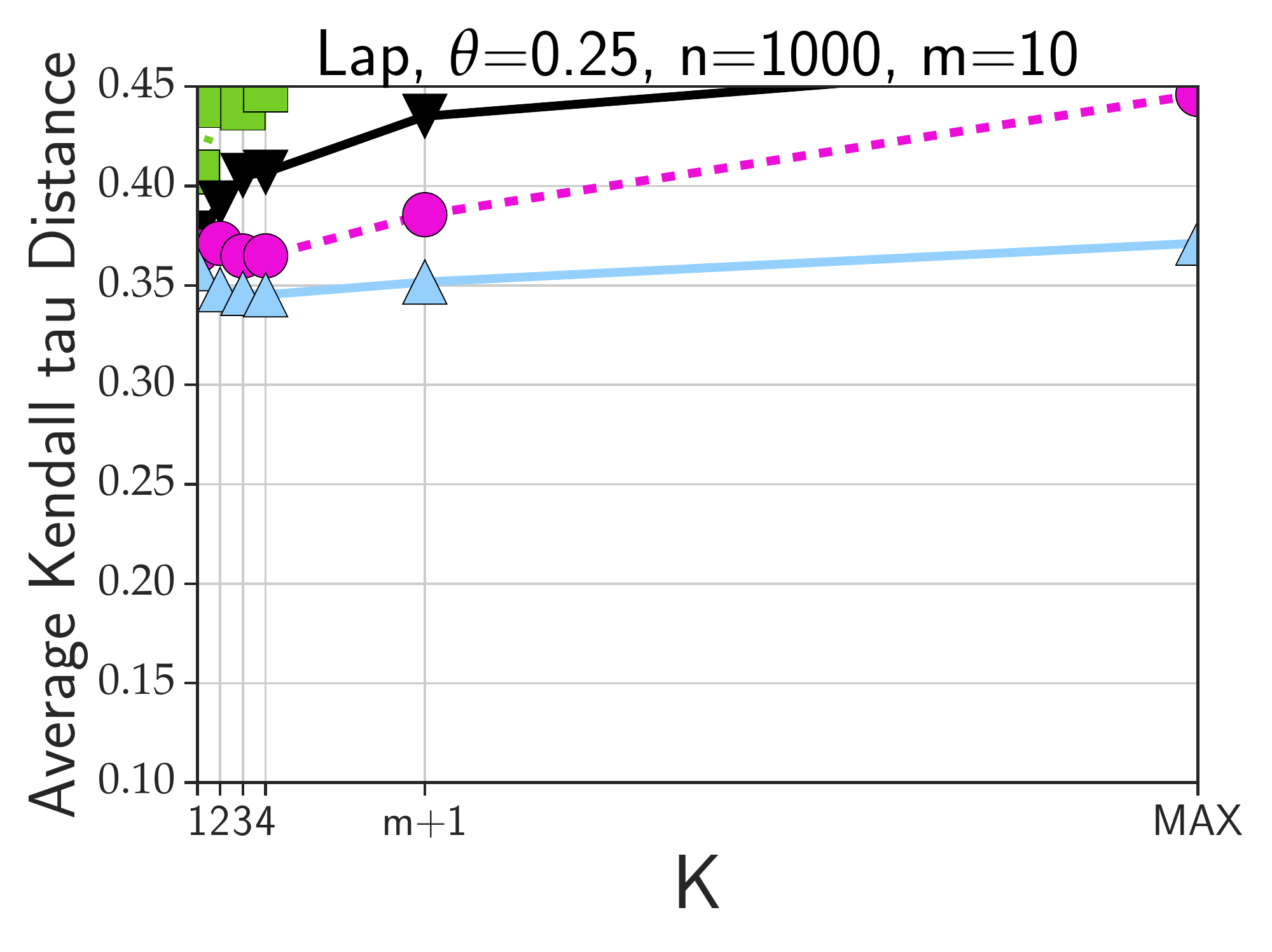}}\
\subfigure[]{\label{subfigure:fig:varying-k-lap:k}
    \includegraphics[width = 0.31\linewidth]{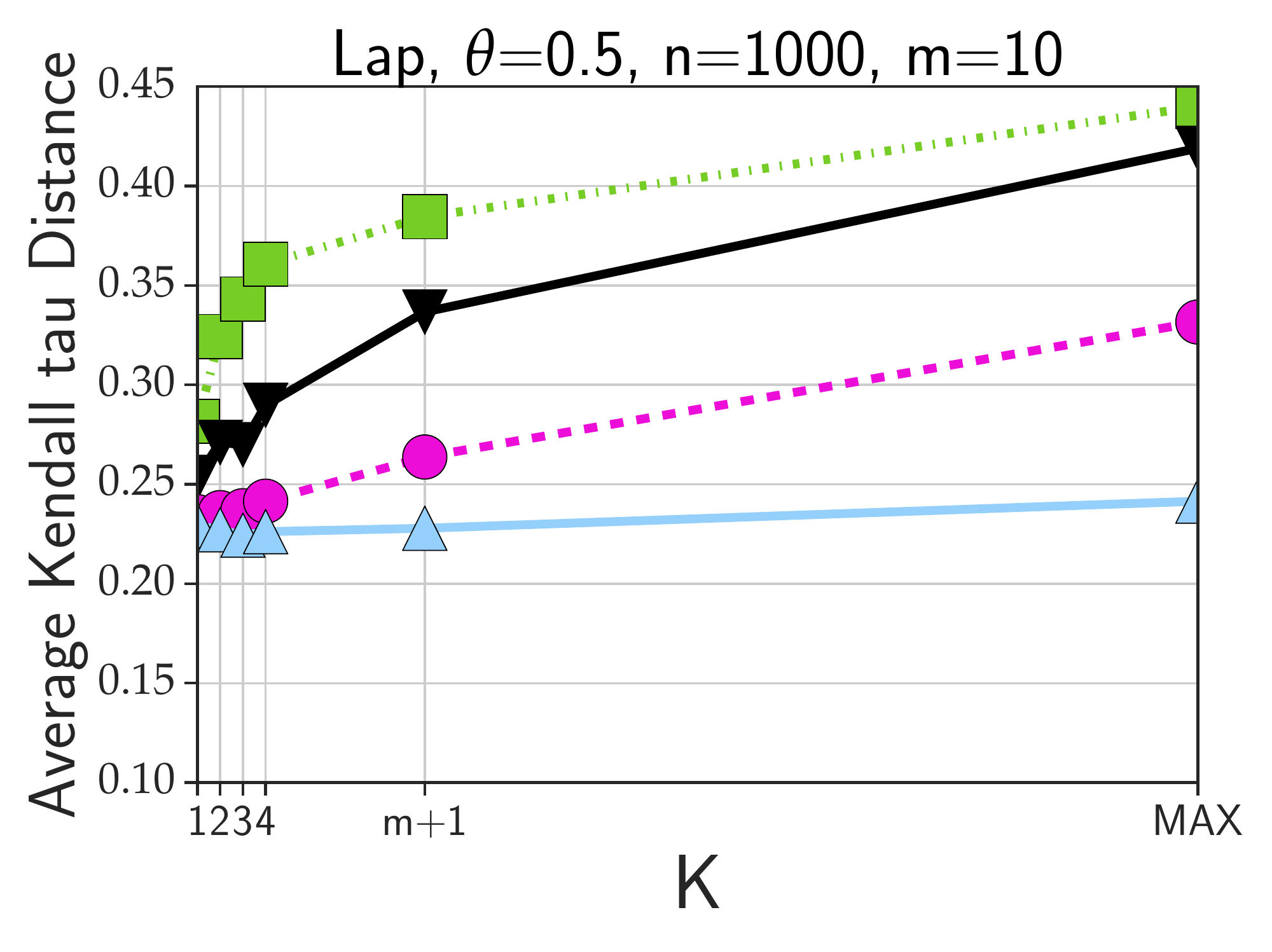}}\
\subfigure[]{\label{subfigure:fig:varying-k-lap:l}
    \includegraphics[width = 0.31\linewidth]{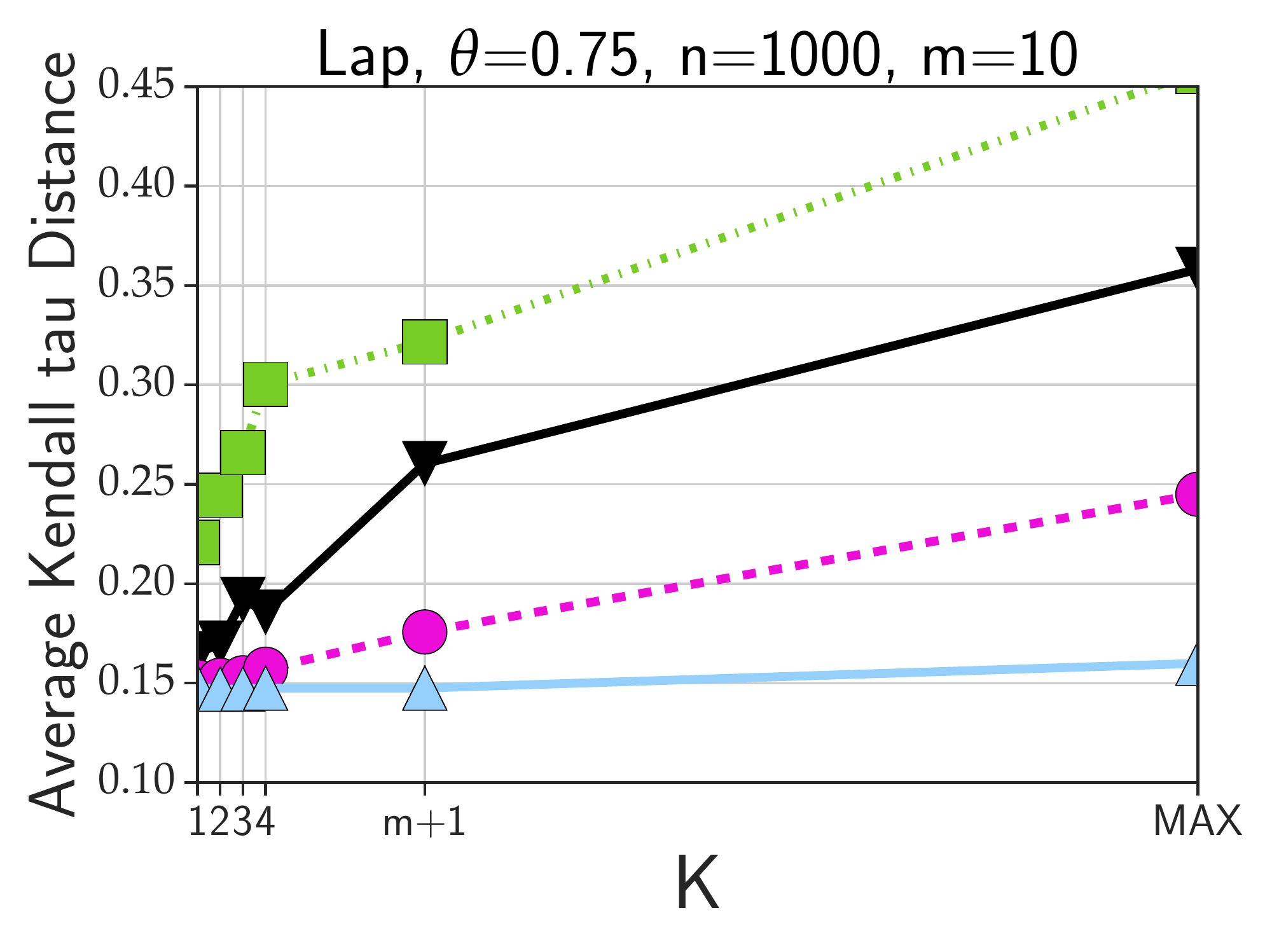}}\
\caption{Performance of \texttt{LDP-KwikSort:Lap} in terms of error rate and average Kendall tau distance on the real-world and synthetic datasets,
across varying the number of queries $K$.}
\label{fig:varying-k-lap}
\end{figure*}

\subsection{The Impact of Privacy Budget}
\label{sec-ldp-ra-results-of-epsilon}

Next,
we consider varying the privacy budget of each agent
in the range $\{0.1,...,5.0\}$,
and observe how this parameter impacts on
the performance of \texttt{LDP-KwikSort:RR}
and \texttt{LDP-KwikSort:Lap}
under the utility metrics such as the error rate
and the average Kendall tau distance.
For the \texttt{DP-KwikSort} algorithm,
as it is based on the local model of DP
and the consumption of the privacy budget
is originate from the curator,
we also vary $\epsilon \in \{0.1,...,5.0\}$.
In this experiment,
we ran these two solutions of \texttt{LDP-KwikSort}
and its competitors
on three real-world datasets and three synthetic datasets.
For the \textsf{TurkDots},
\textsf{TurkPuzzle},
and \textsf{SUSHI} datasets,
the numbers of agents and alternatives are fixed by default.
For the synthetic datasets,
we set their dispersion parameter $\theta=0.25, 0.5, 0.75$,
as well as $n=5000$ and $m=15$.
The results are shown in~\Cref{fig:avgkt-varying-epsilon}.

Firstly,
the non-private \texttt{KwikSort}
determines the overall lower error bound
of average Kendall tau distance,
and \texttt{DP-KwikSort} can approach it with a small privacy budget,
say $\epsilon=0.1$.
Secondly,
with increasing the privacy budget,
both of \texttt{LDP-KwikSort:RR} and \texttt{LDP-KwikSort:Lap}
get lower errors,
and the former outperforms the latter.
Besides,
the performance changing trends of solutions
under two metrics are consistent,
which is due to the error rate reflecting
the accuracy of intermediate results
while the average Kendall tau distance reflecting
the utility of the aggregate ranking.

\begin{figure*}[p]
\centering
\subfigtopskip=2pt
\subfigbottomskip=1pt
\subfigcapskip=-10pt
\includegraphics[width = 0.36\linewidth]{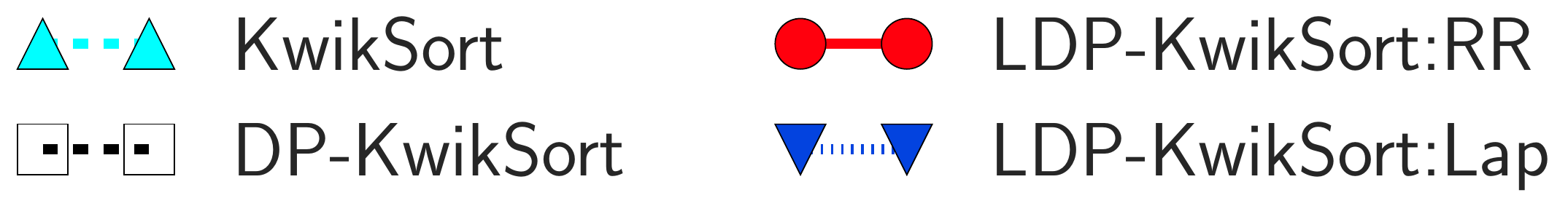}\\
\vspace{0.2cm}
\subfigure[]{\label{subfigure:fig:avgkt-varying-epsilon:a}
    \includegraphics[width = 0.31\linewidth]{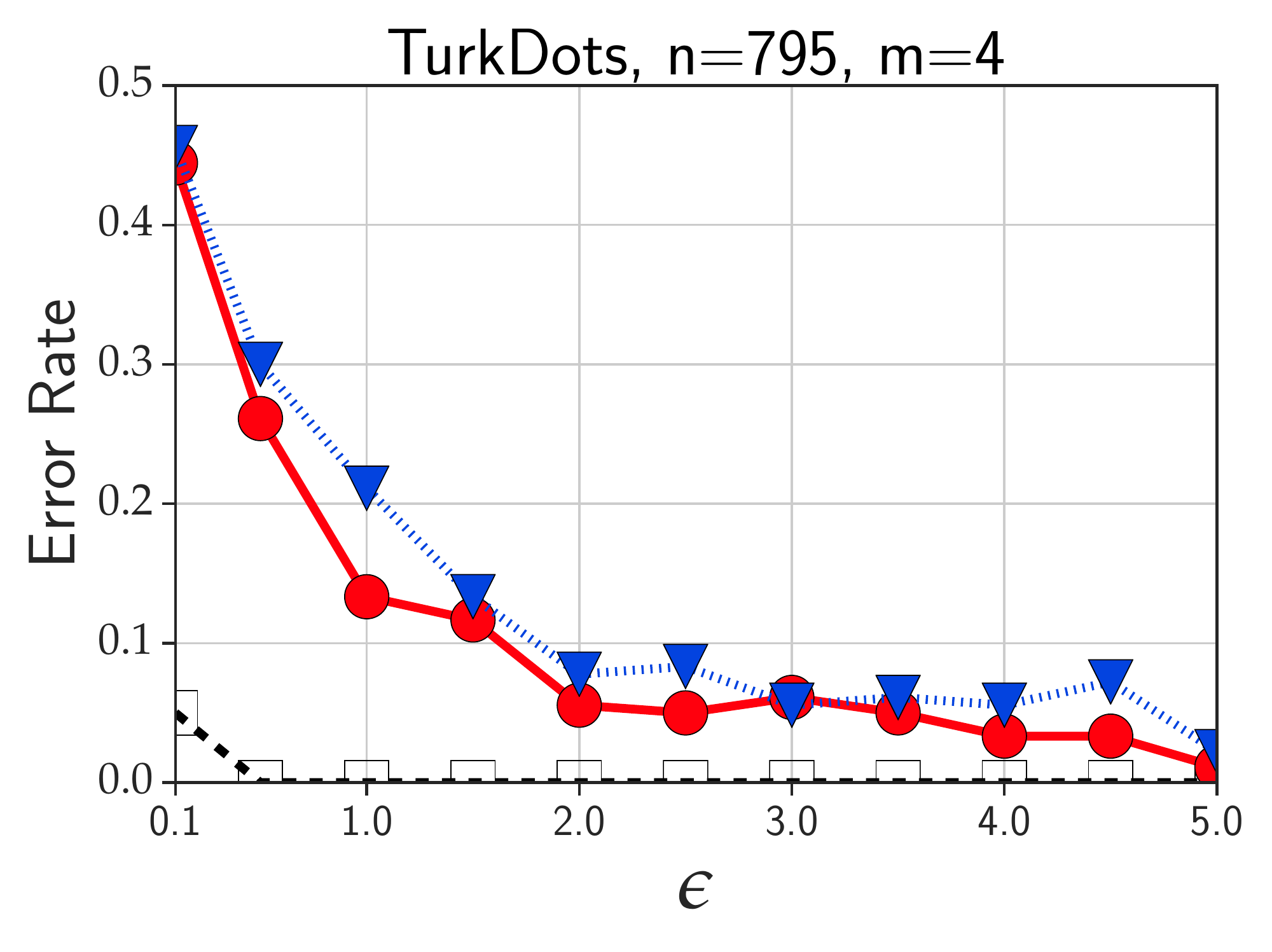}}\
\subfigure[]{\label{subfigure:fig:avgkt-varying-epsilon:b}
    \includegraphics[width = 0.31\linewidth]{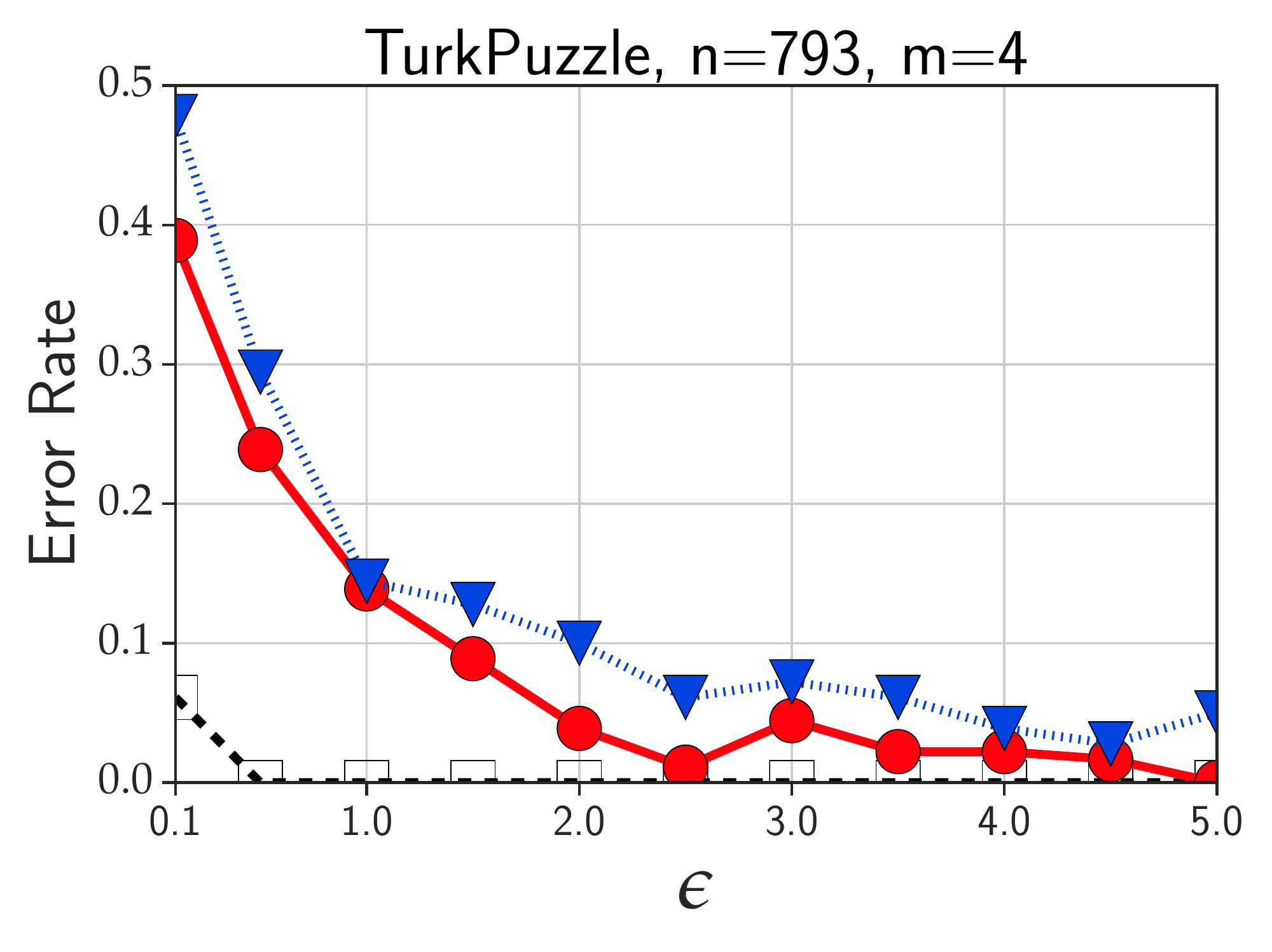}}\
\subfigure[]{\label{subfigure:fig:avgkt-varying-epsilon:c}
    \includegraphics[width = 0.31\linewidth]{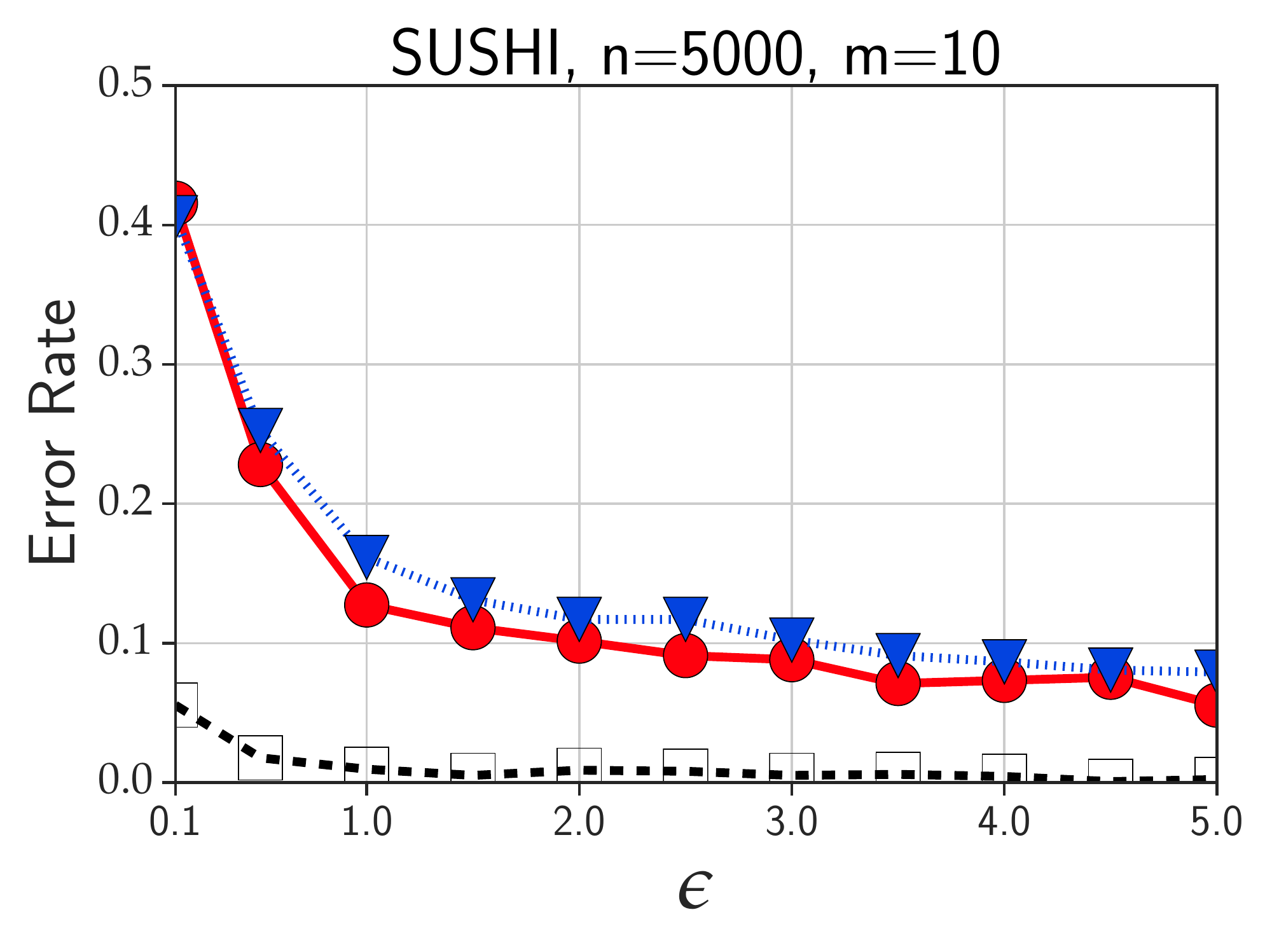}}\
\subfigure[]{\label{subfigure:fig:avgkt-varying-epsilon:d}
    \includegraphics[width = 0.31\linewidth]{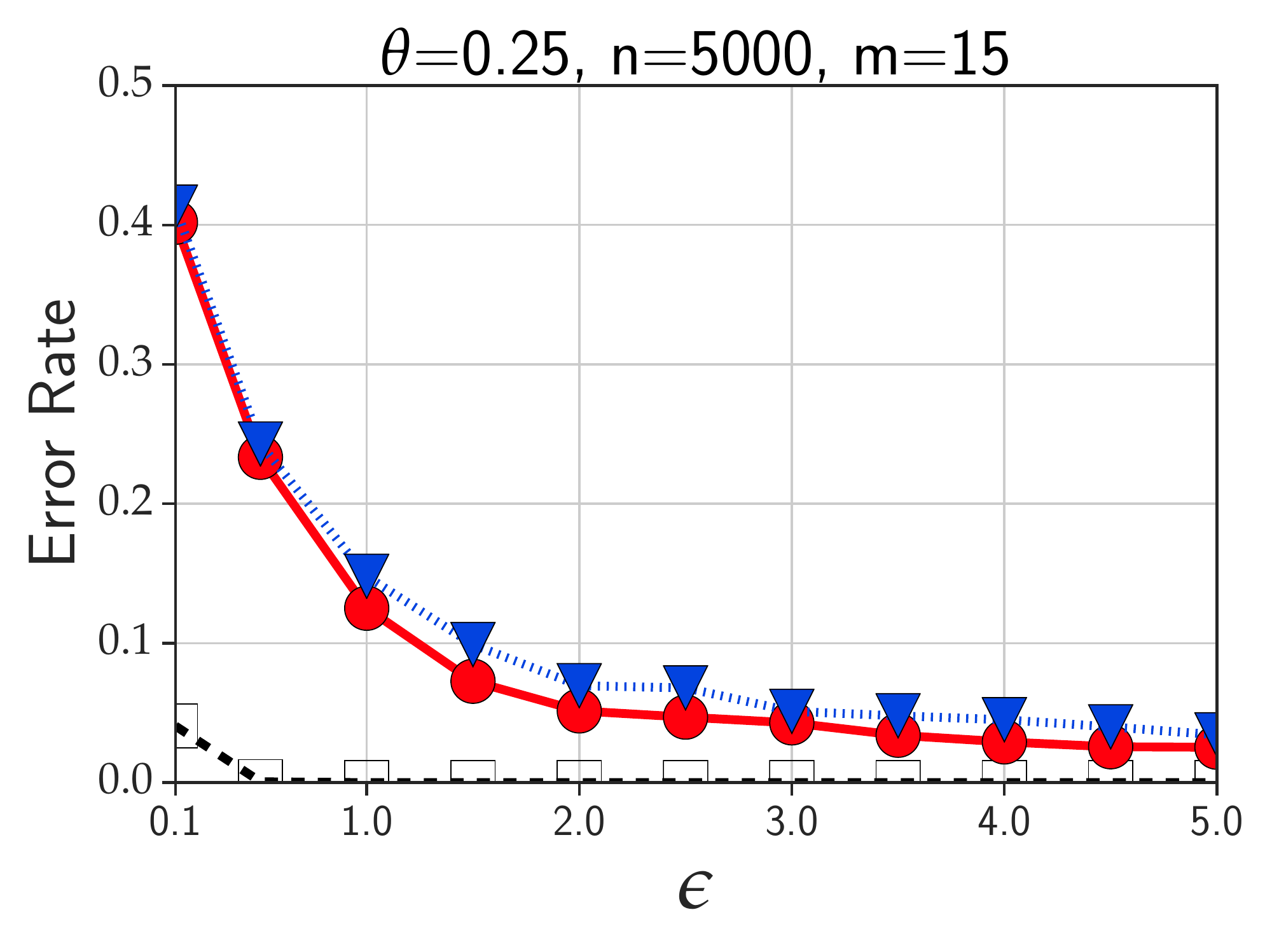}}\
\subfigure[]{\label{subfigure:fig:avgkt-varying-epsilon:e}
    \includegraphics[width = 0.31\linewidth]{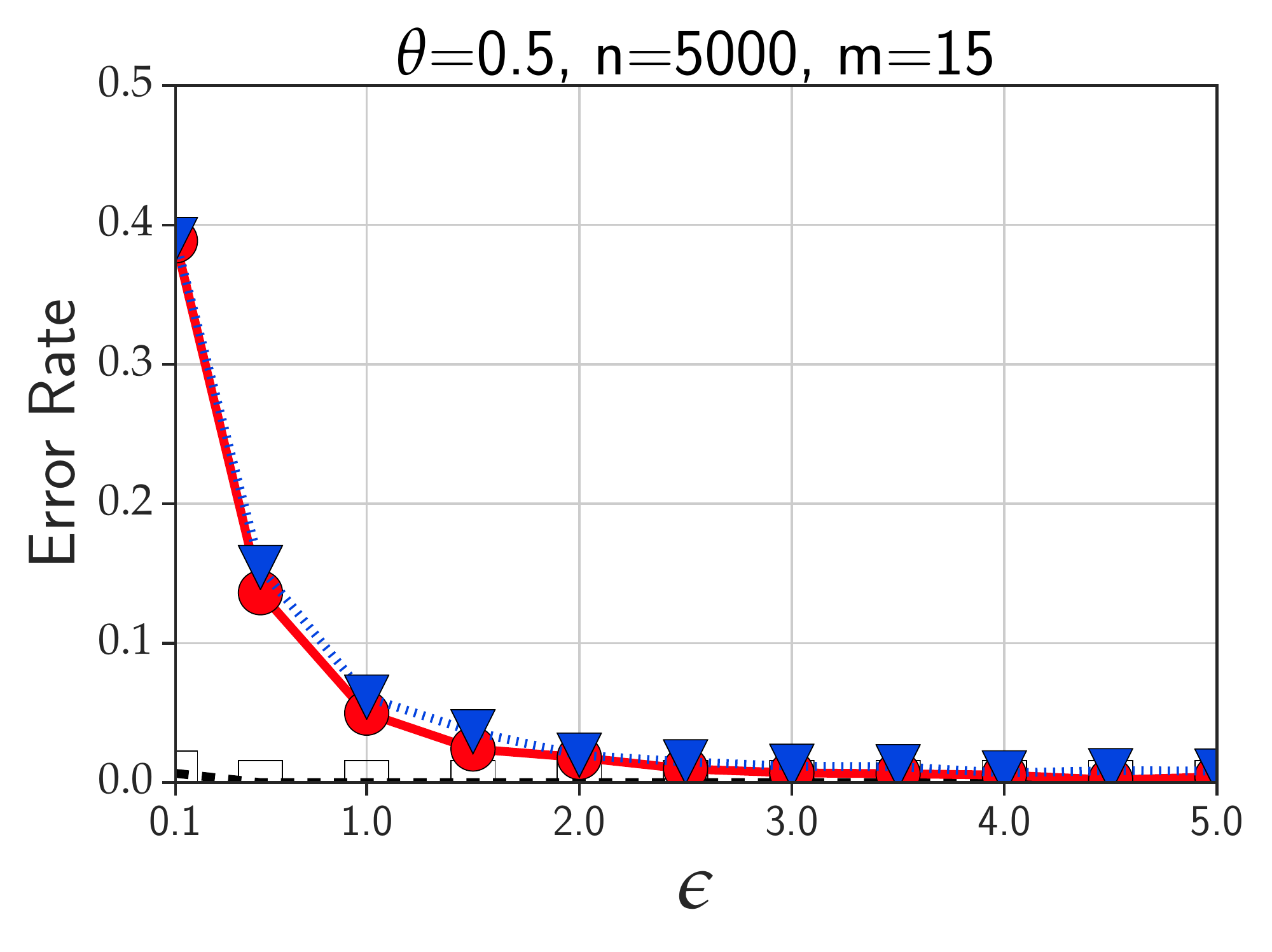}}\
\subfigure[]{\label{subfigure:fig:avgkt-varying-epsilon:f}
    \includegraphics[width = 0.31\linewidth]{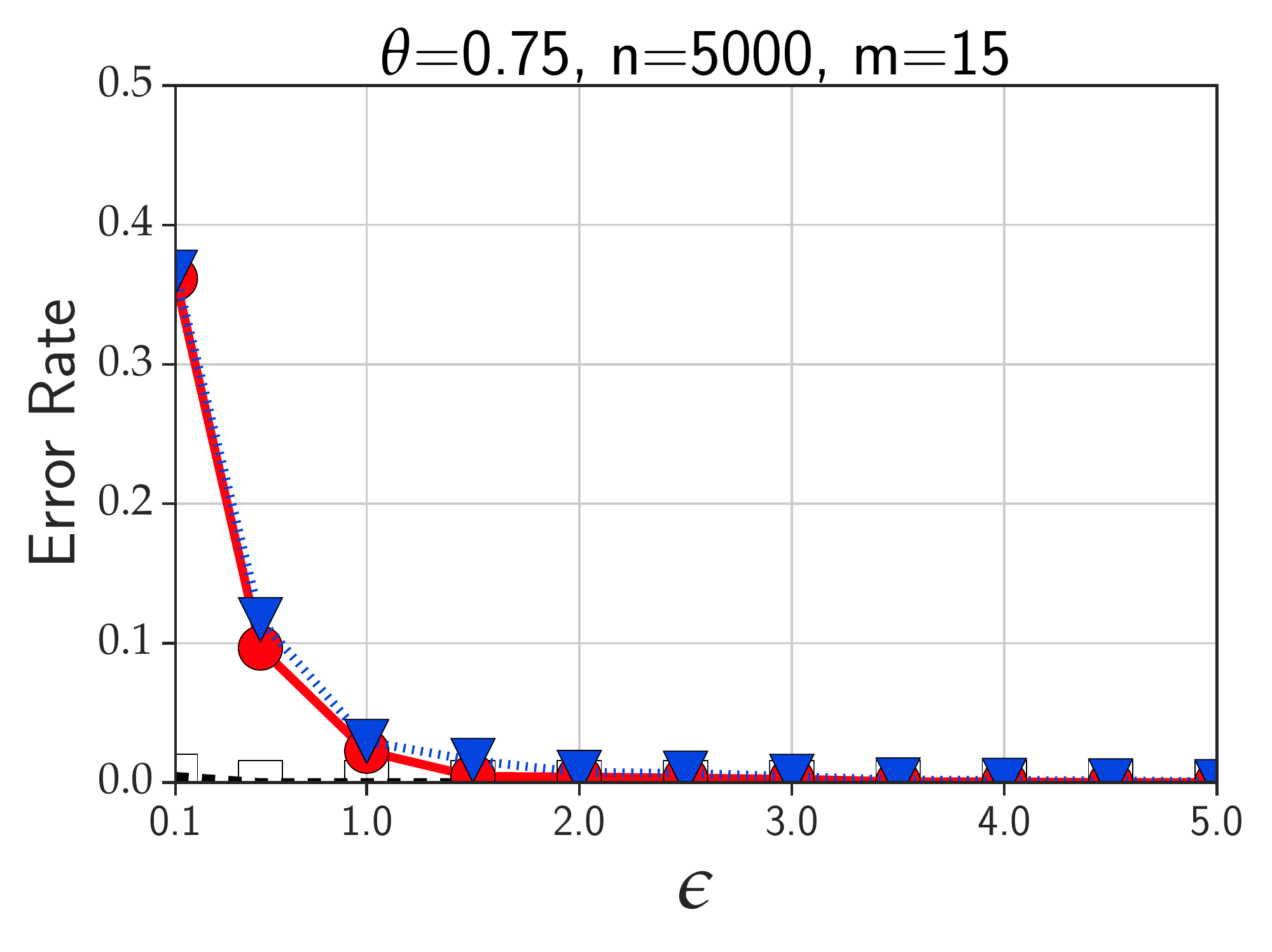}}\
\subfigure[]{\label{subfigure:fig:avgkt-varying-epsilon:g}
    \includegraphics[width = 0.31\linewidth]{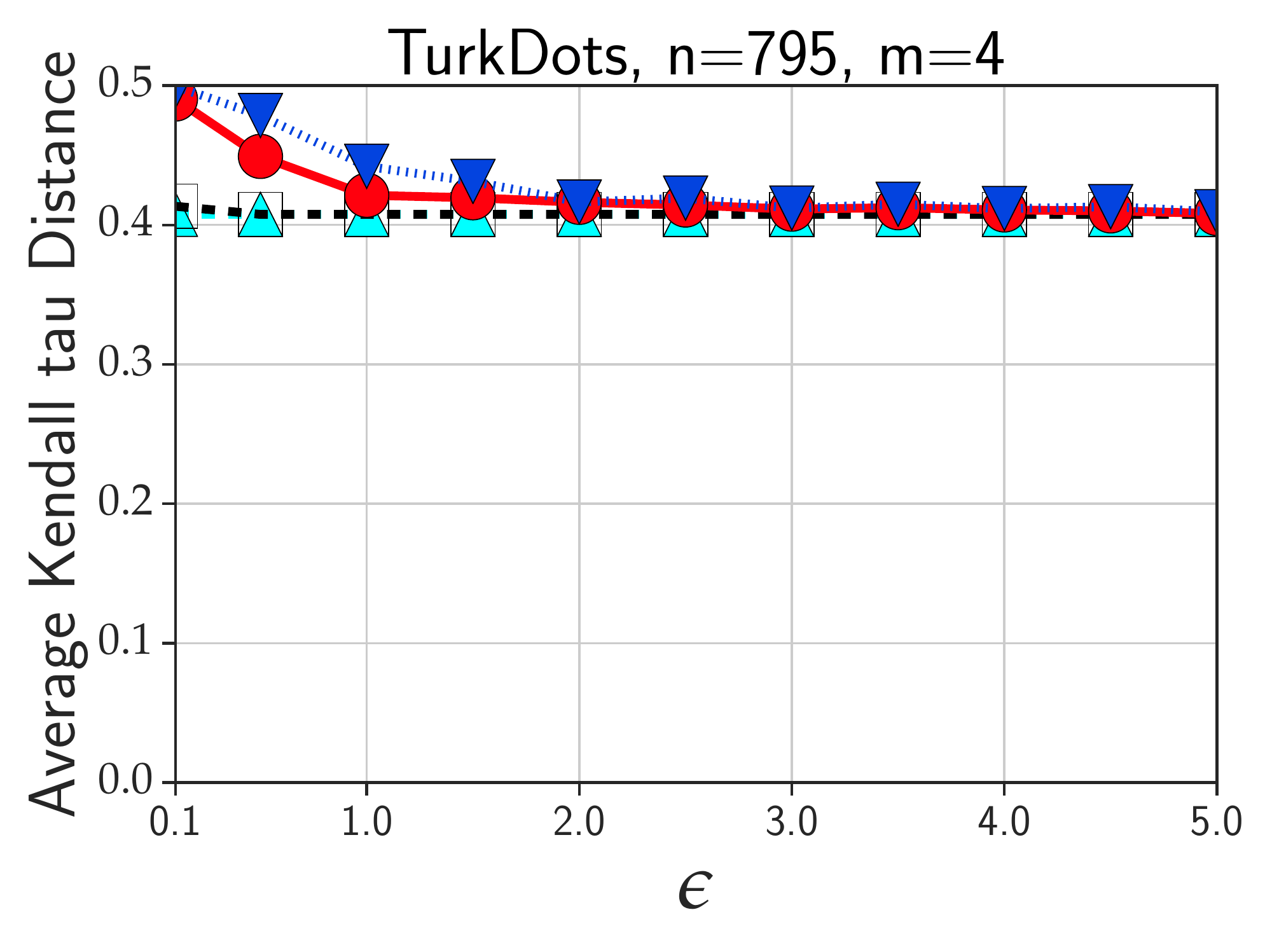}}\
\subfigure[]{\label{subfigure:fig:avgkt-varying-epsilon:h}
    \includegraphics[width = 0.31\linewidth]{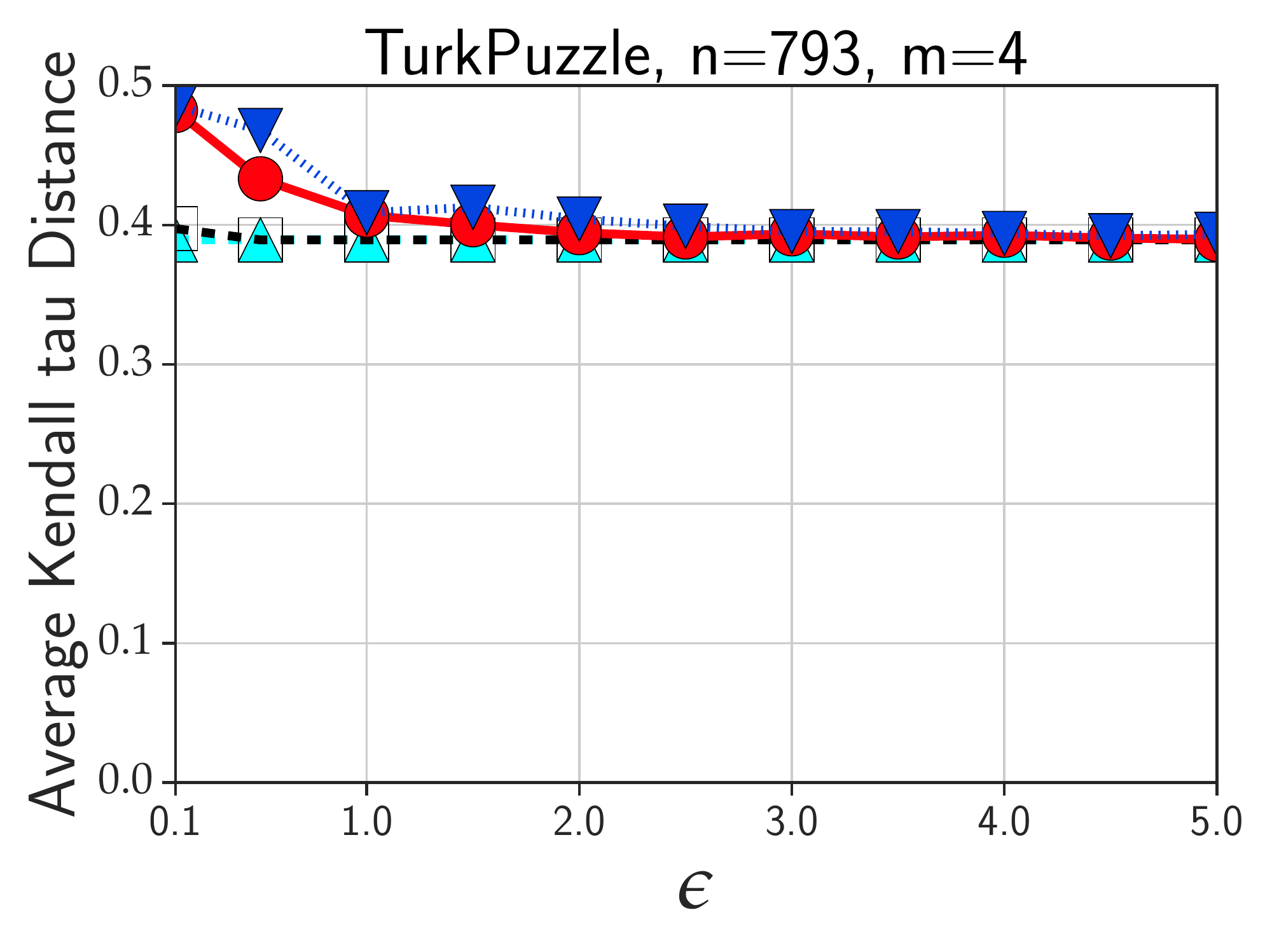}}\
\subfigure[]{\label{subfigure:fig:avgkt-varying-epsilon:i}
    \includegraphics[width = 0.31\linewidth]{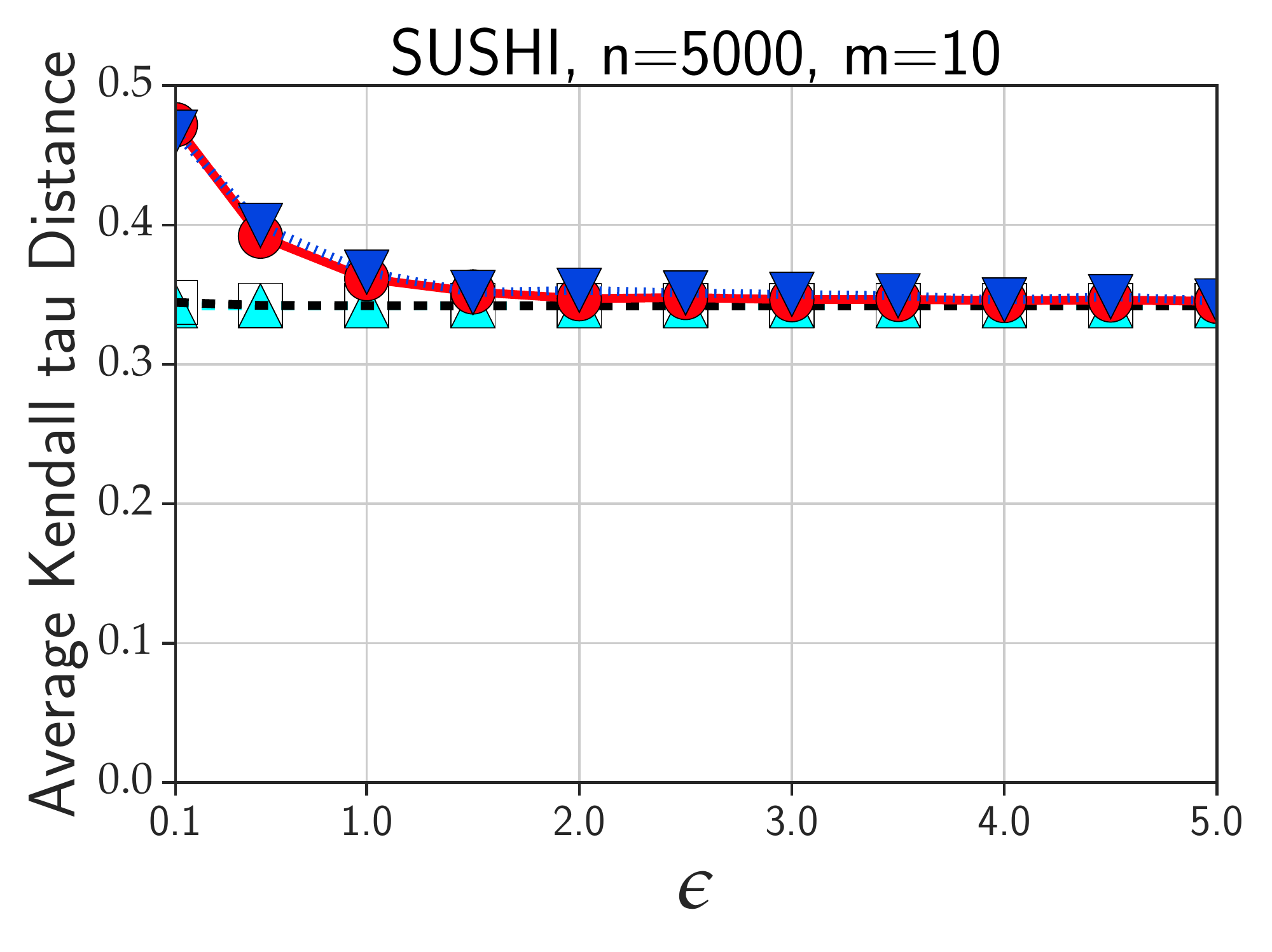}}\
\subfigure[]{\label{subfigure:fig:avgkt-varying-epsilon:j}
    \includegraphics[width = 0.31\linewidth]{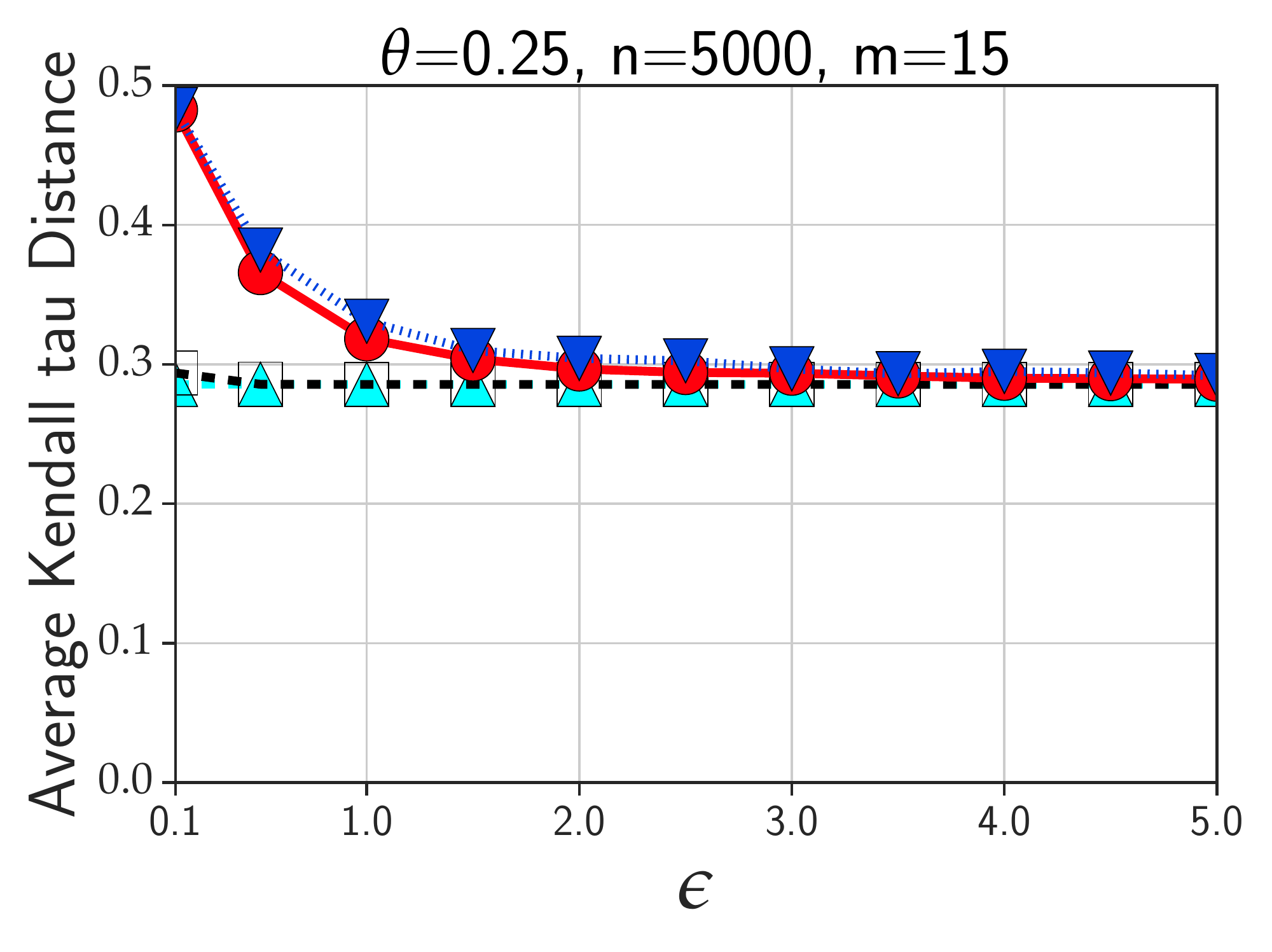}}\
\subfigure[]{\label{subfigure:fig:avgkt-varying-epsilon:k}
    \includegraphics[width = 0.31\linewidth]{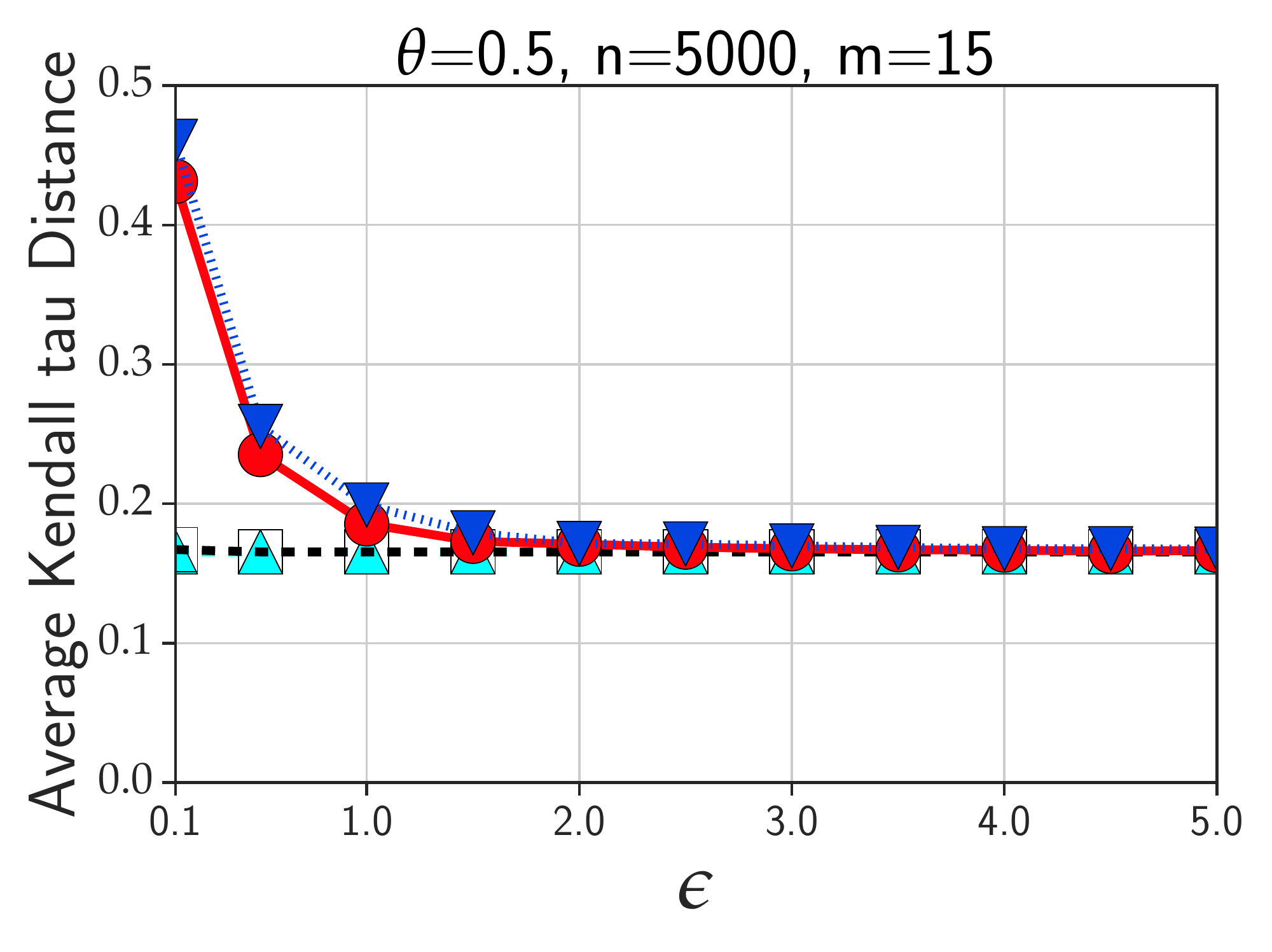}}\
\subfigure[]{\label{subfigure:fig:avgkt-varying-epsilon:l}
    \includegraphics[width = 0.31\linewidth]{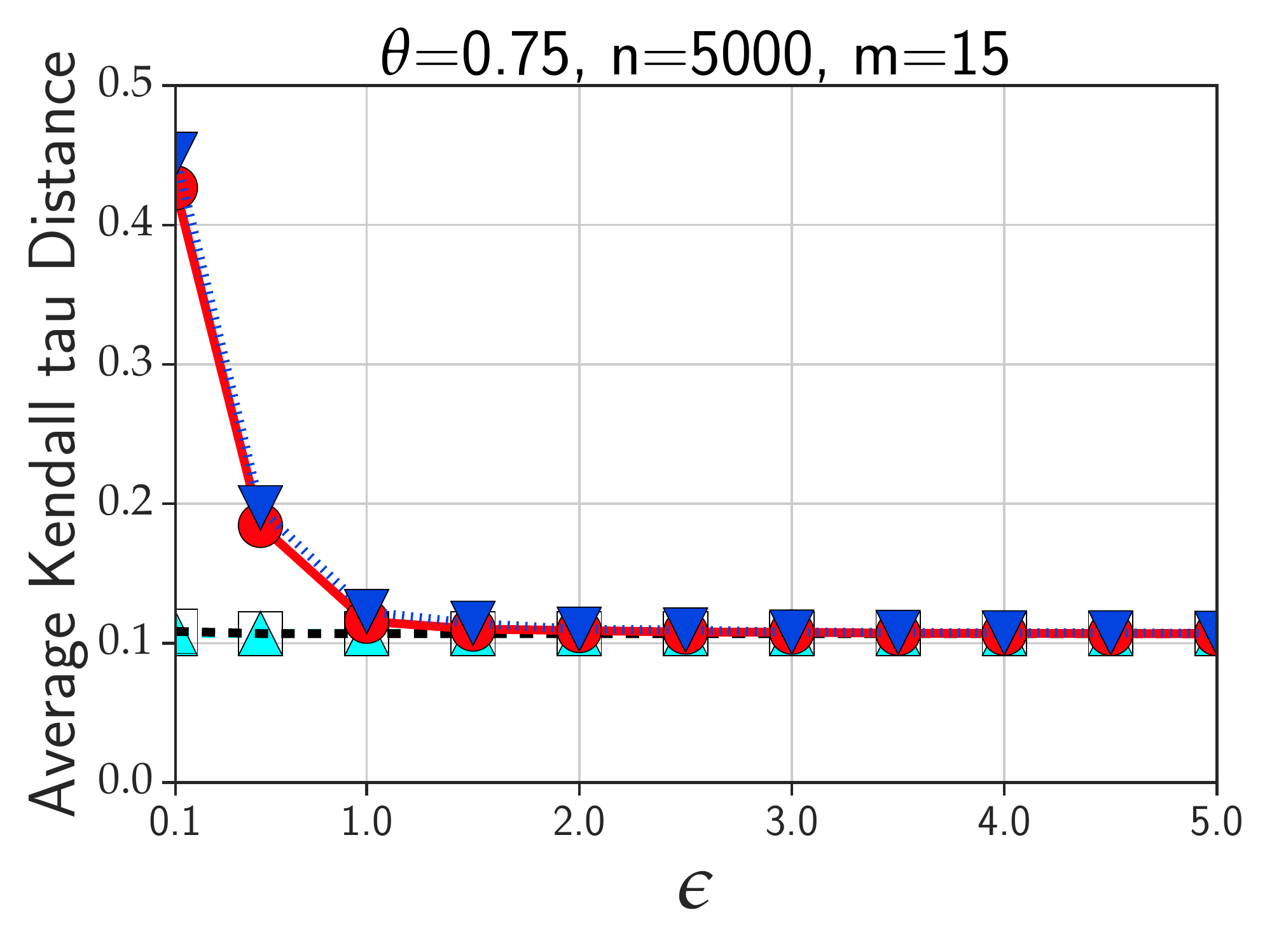}}\
\caption{Comparison of protocols in terms of error rate and average Kendall tau distance
on the real-world and synthetic datasets,
across varying the privacy budget $\epsilon \in \{0.1,...5.0\}$.}
\label{fig:avgkt-varying-epsilon}
\end{figure*}

\subsection{The Impact of Agent and Alternative Amount}
\label{sec-ldp-ra-results-of-n-m}

Intuitively,
more alternatives that need to be ranked
challenge the protocols for generating an optimal ranking.
In order to maintain the results with an acceptable utility,
the increase of the number of agents may help the curator
to collect more information for rank aggregation.
To investigate how the numbers of agents $n$
and alternatives $m$ impact on the performance,
we ran these two solutions of \texttt{LDP-KwikSort}
and its competitors
on three real-world datasets and nine synthetic datasets,
and consider to vary $n \in \{10,...,10000\}$
and $m \in \{4,10,15,30,45\}$,
for observing the performance under
the average Kendall tau distance.
The results are shown in~\Cref{fig:avgkt-varying-agents}.

\begin{figure*}[p]
\centering
\subfigtopskip=2pt
\subfigbottomskip=1pt
\subfigcapskip=-10pt
\includegraphics[width = 0.36\linewidth]{./fig_legend2.pdf}\\
\vspace{0.2cm}
\subfigure[]{\label{subfigure:fig:avgkt-varying-agents:a}
    \includegraphics[width = 0.3\linewidth]{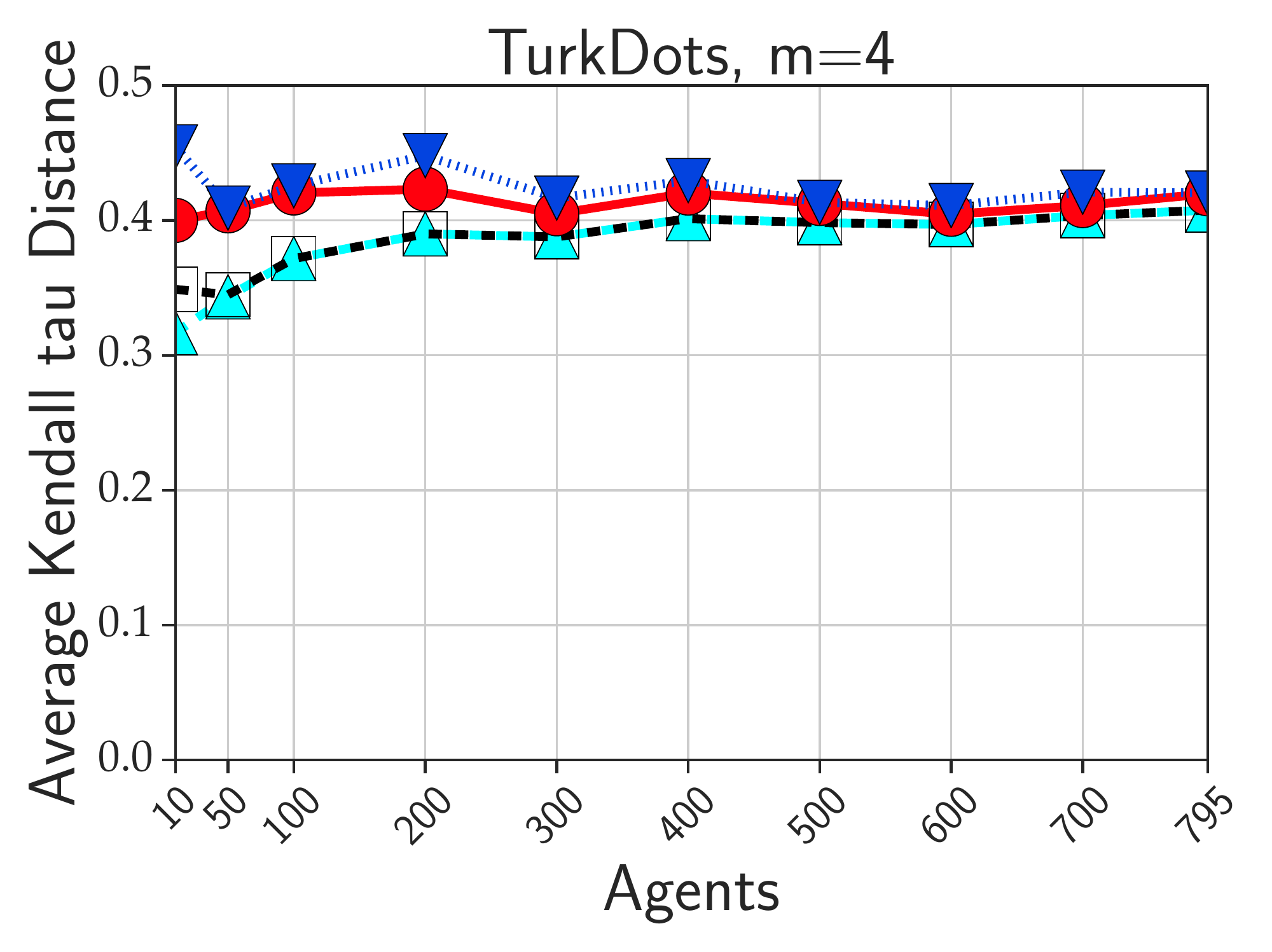}}\
\subfigure[]{\label{subfigure:fig:avgkt-varying-agents:b}
    \includegraphics[width = 0.3\linewidth]{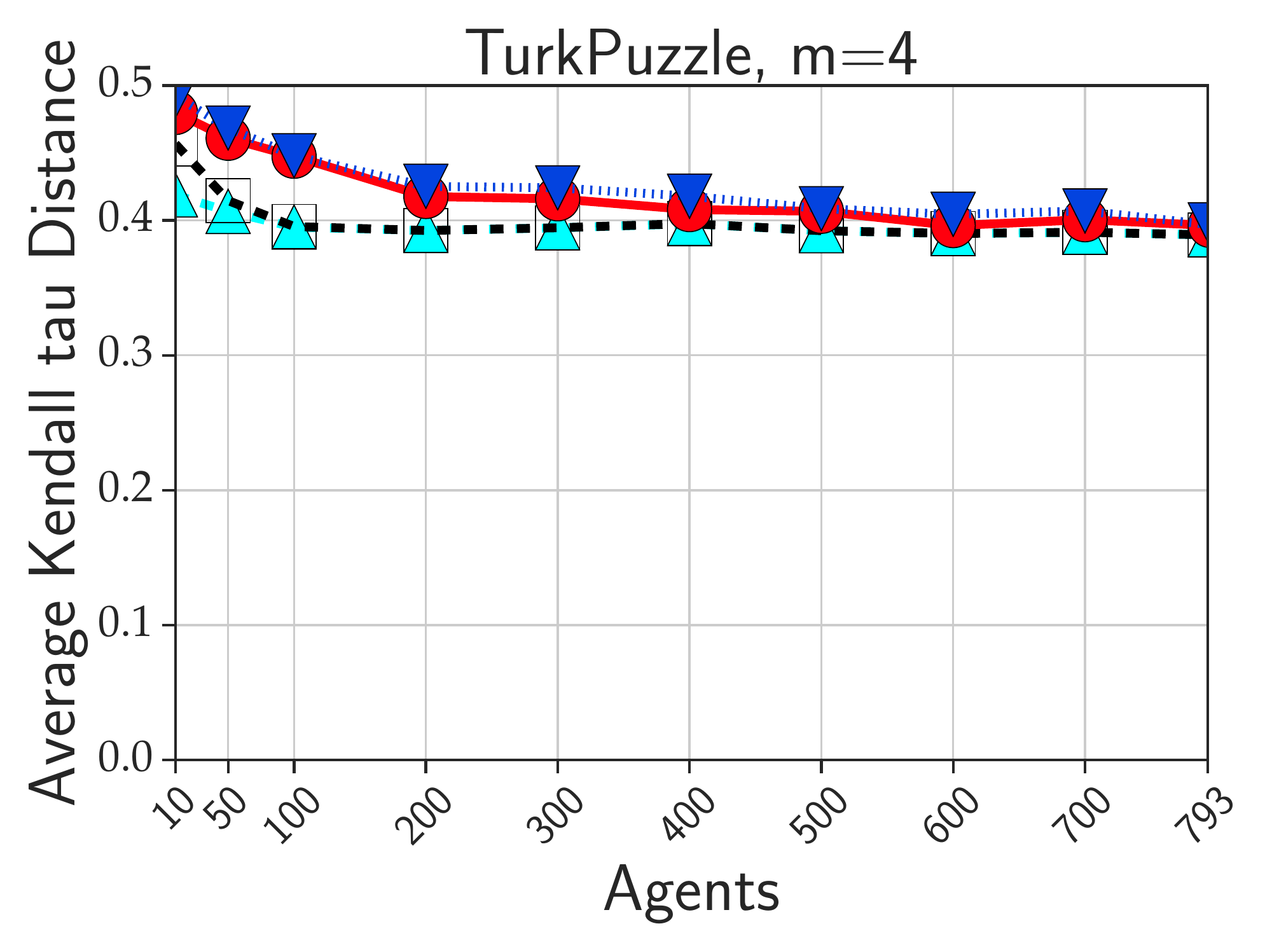}}\
\subfigure[]{\label{subfigure:fig:avgkt-varying-agents:c}
    \includegraphics[width = 0.3\linewidth]{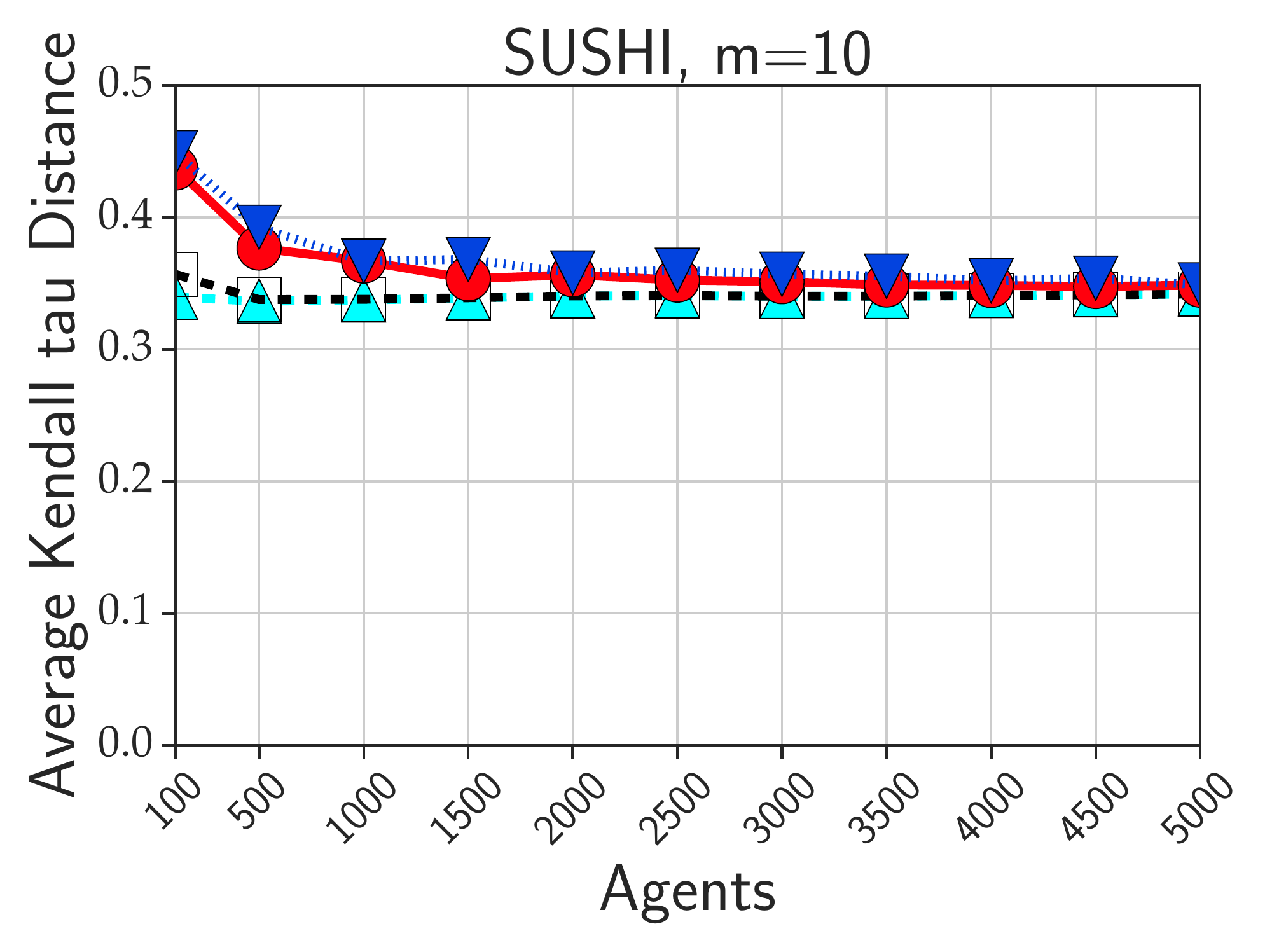}}\
\subfigure[]{\label{subfigure:fig:avgkt-varying-agents:d}
    \includegraphics[width = 0.3\linewidth]{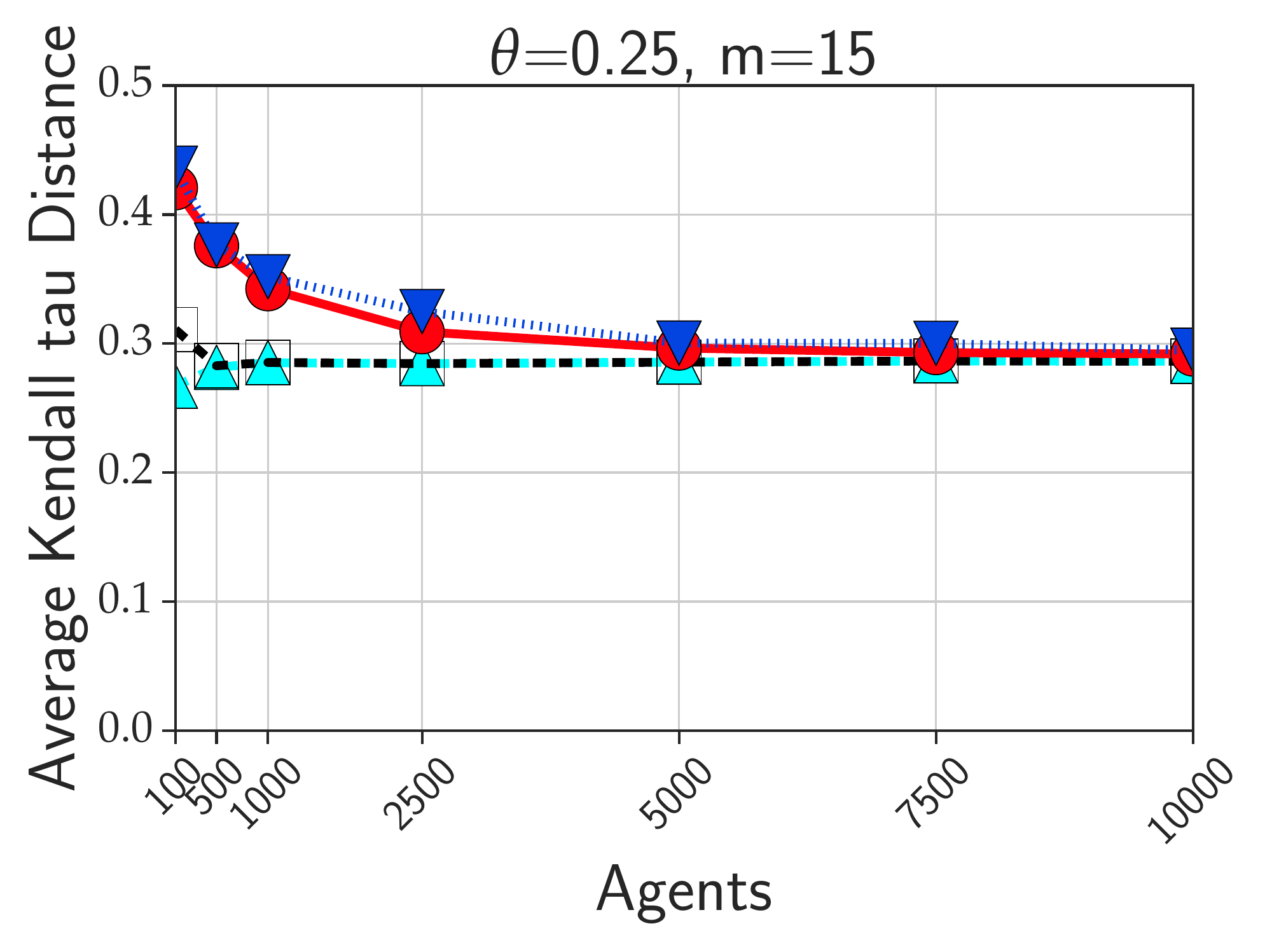}}\
\subfigure[]{\label{subfigure:fig:avgkt-varying-agents:e}
    \includegraphics[width = 0.3\linewidth]{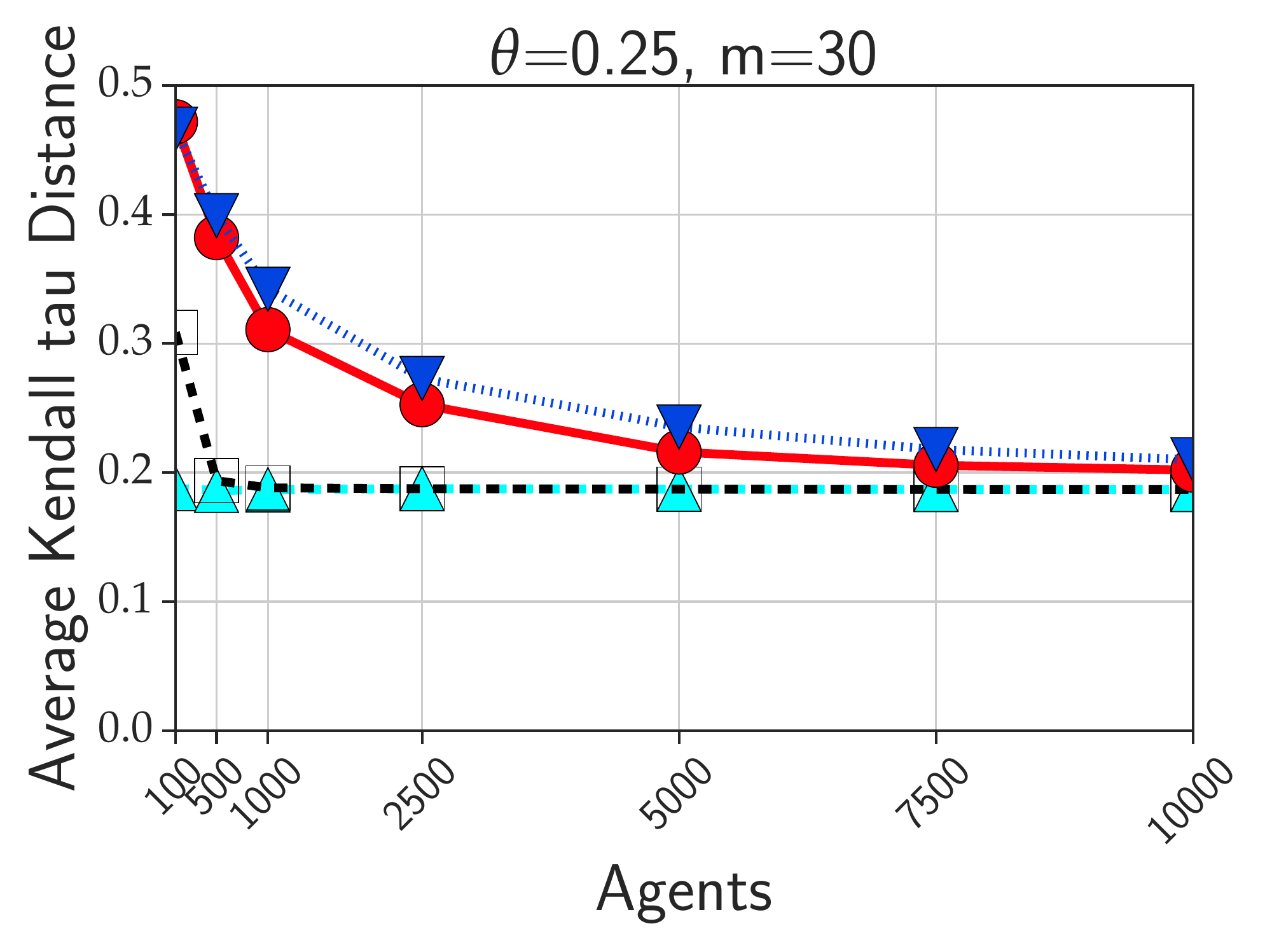}}\
\subfigure[]{\label{subfigure:fig:avgkt-varying-agents:f}
    \includegraphics[width = 0.3\linewidth]{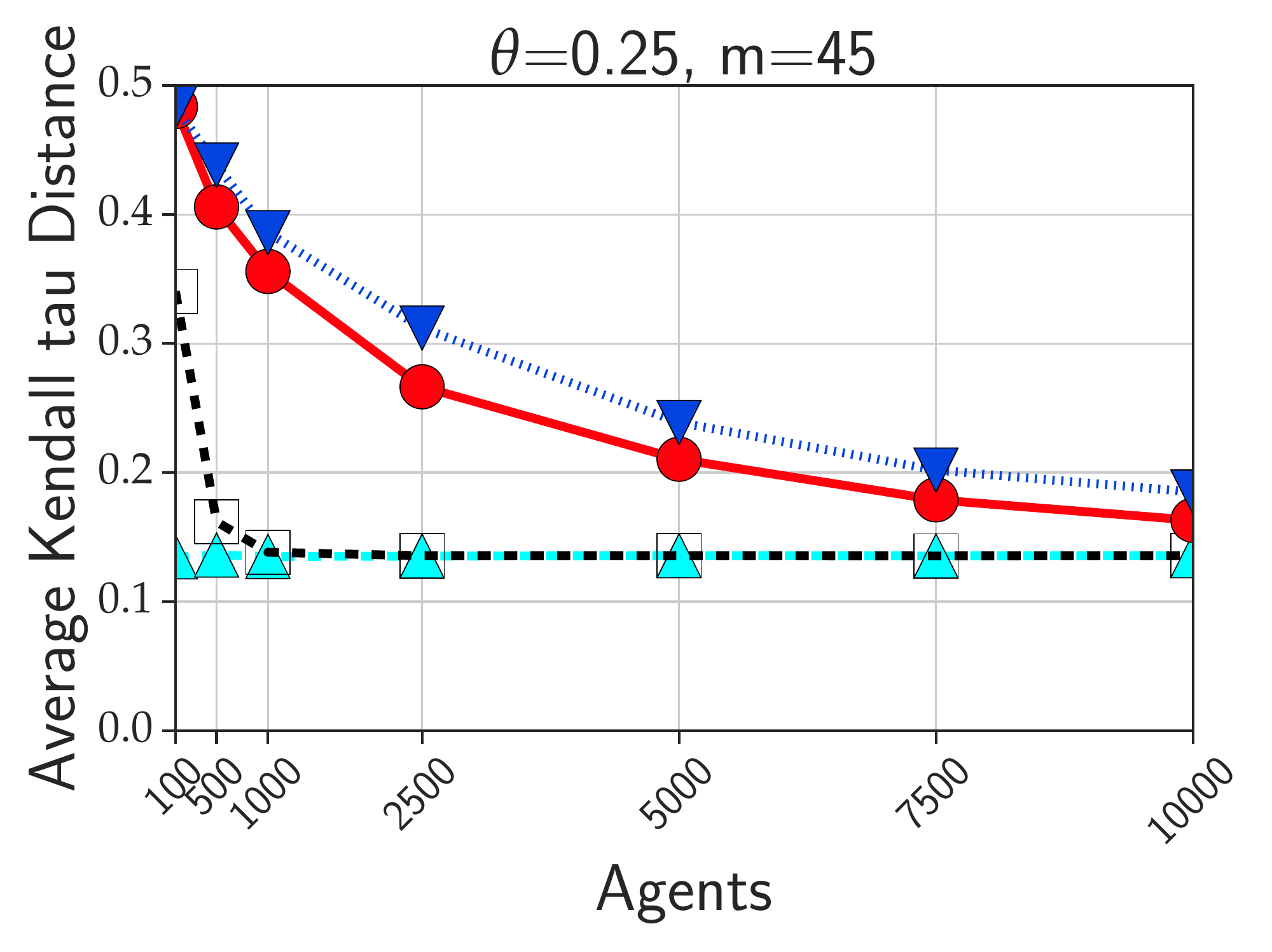}}\
\subfigure[]{\label{subfigure:fig:avgkt-varying-agents:g}
    \includegraphics[width = 0.3\linewidth]{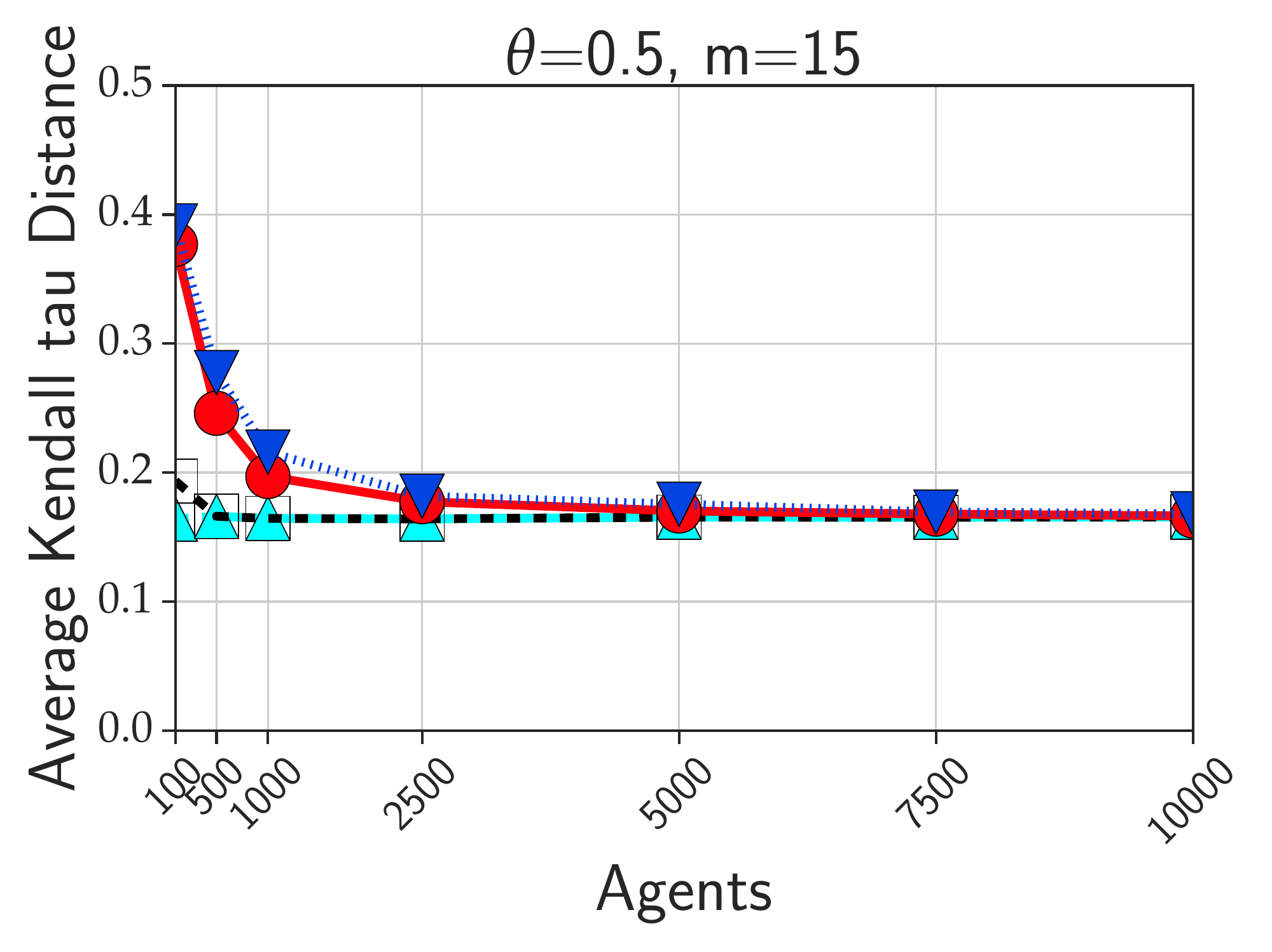}}\
\subfigure[]{\label{subfigure:fig:avgkt-varying-agents:h}
    \includegraphics[width = 0.3\linewidth]{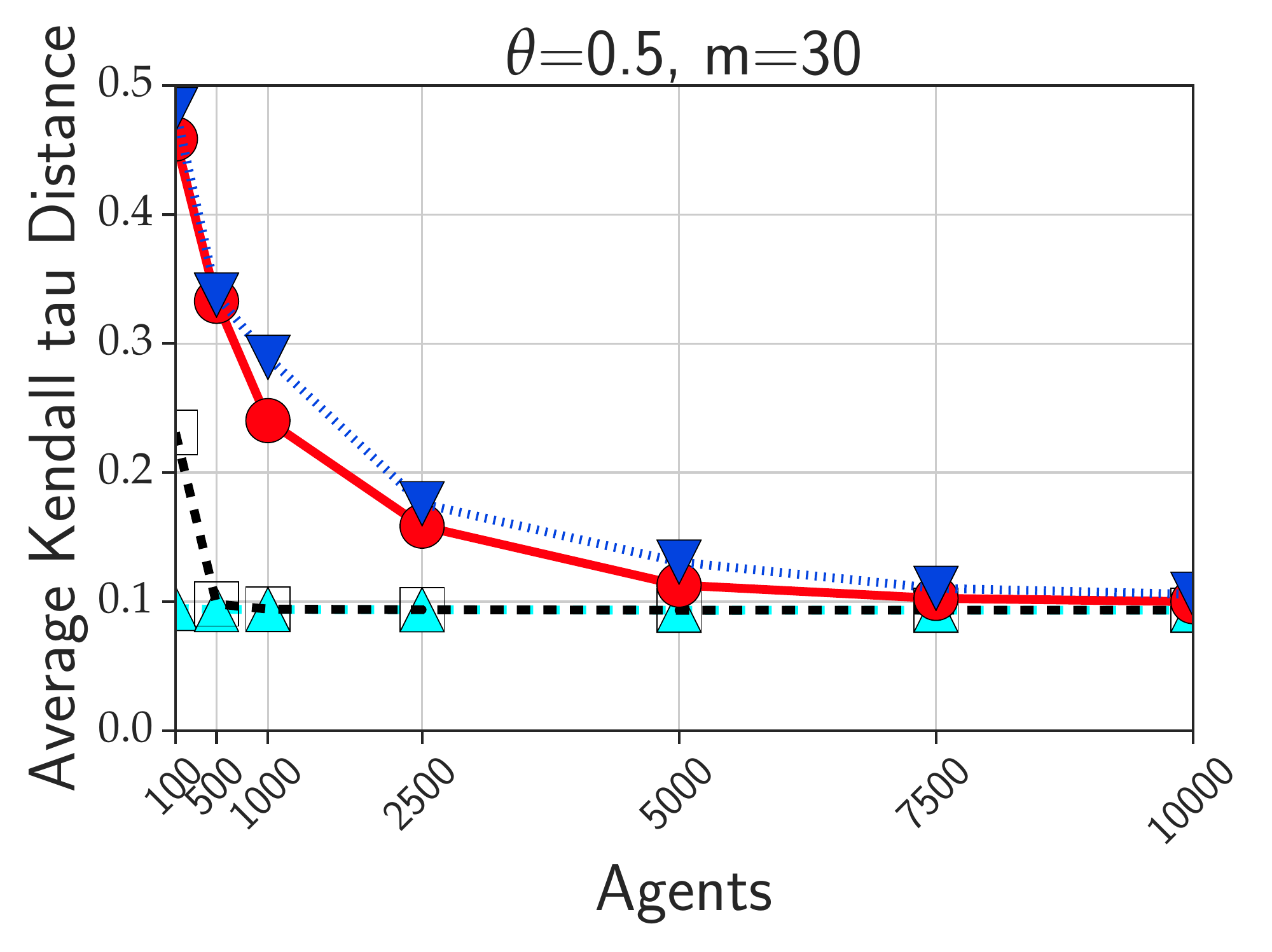}}\
\subfigure[]{\label{subfigure:fig:avgkt-varying-agents:i}
    \includegraphics[width = 0.3\linewidth]{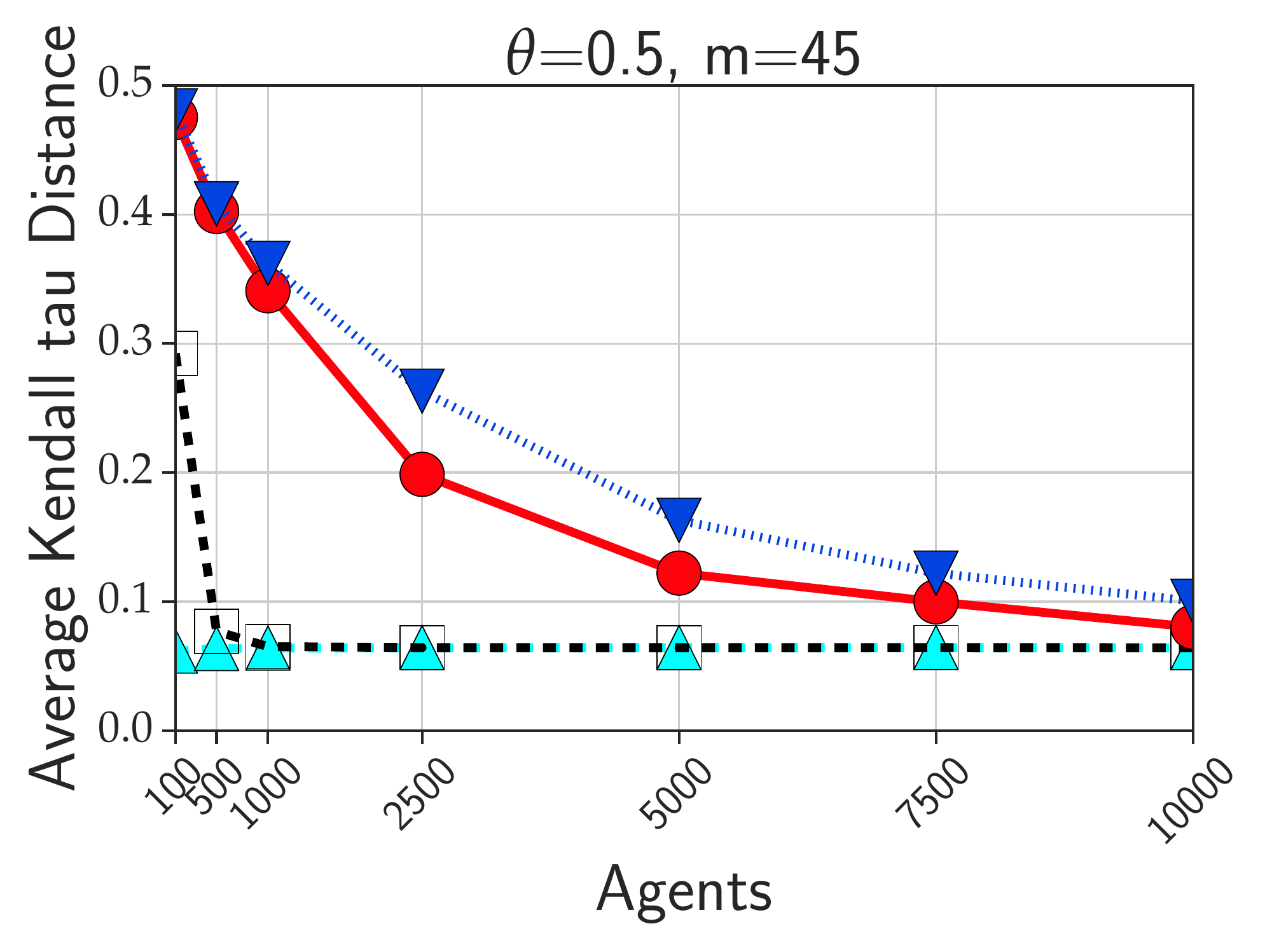}}\
\subfigure[]{\label{subfigure:fig:avgkt-varying-agents:j}
    \includegraphics[width = 0.3\linewidth]{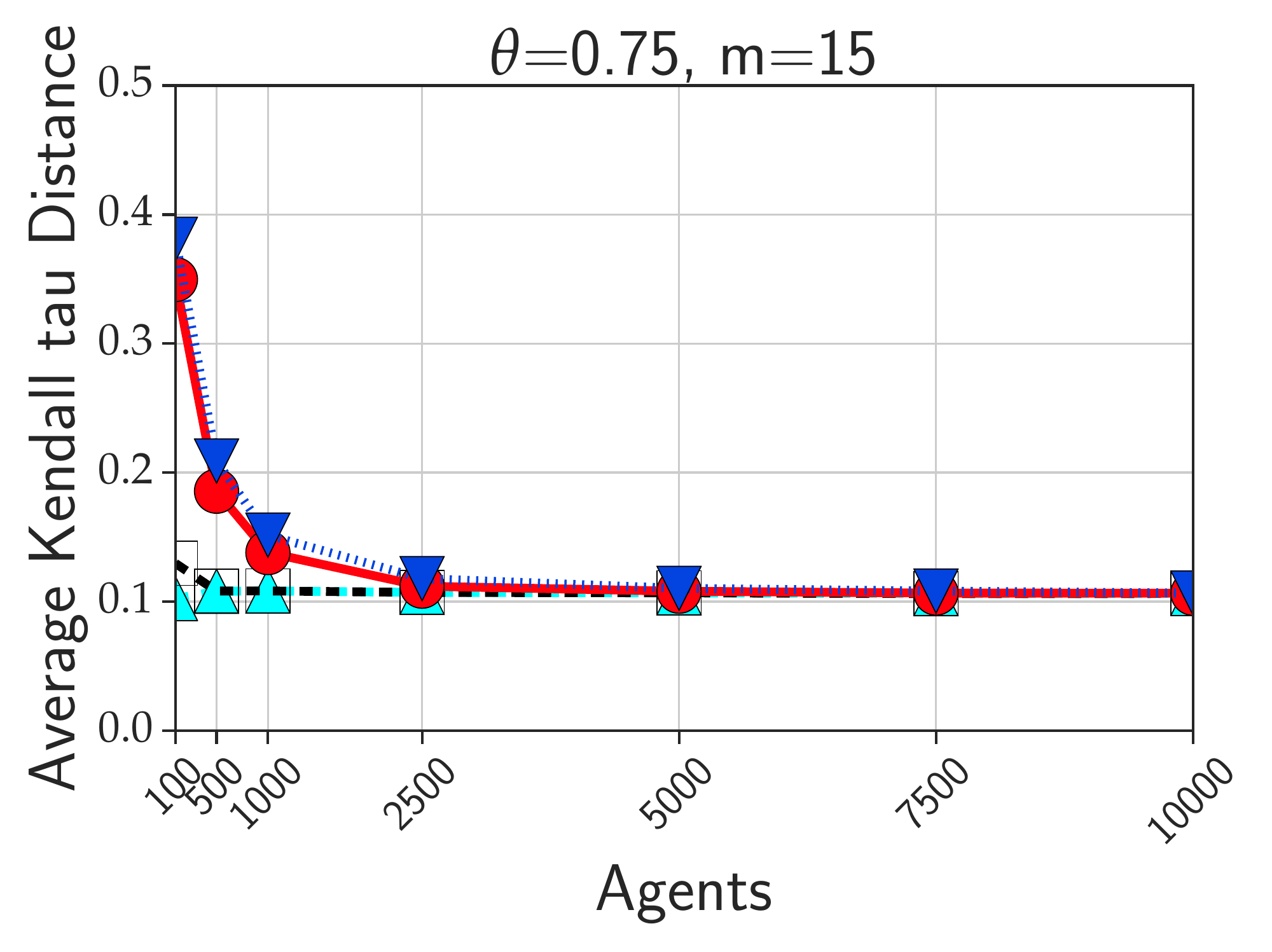}}\
\subfigure[]{\label{subfigure:fig:avgkt-varying-agents:k}
    \includegraphics[width = 0.3\linewidth]{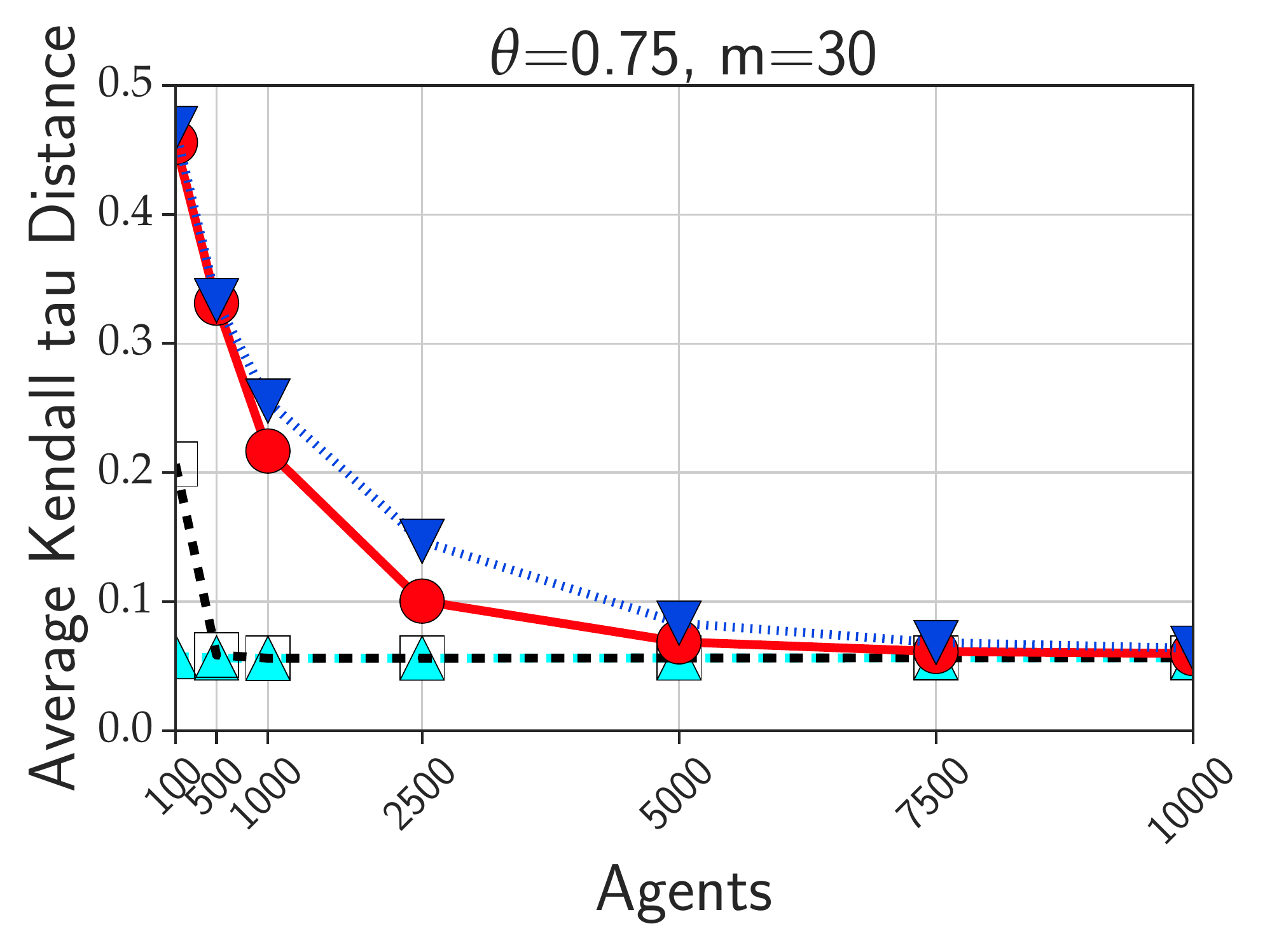}}\
\subfigure[]{\label{subfigure:fig:avgkt-varying-agents:l}
    \includegraphics[width = 0.3\linewidth]{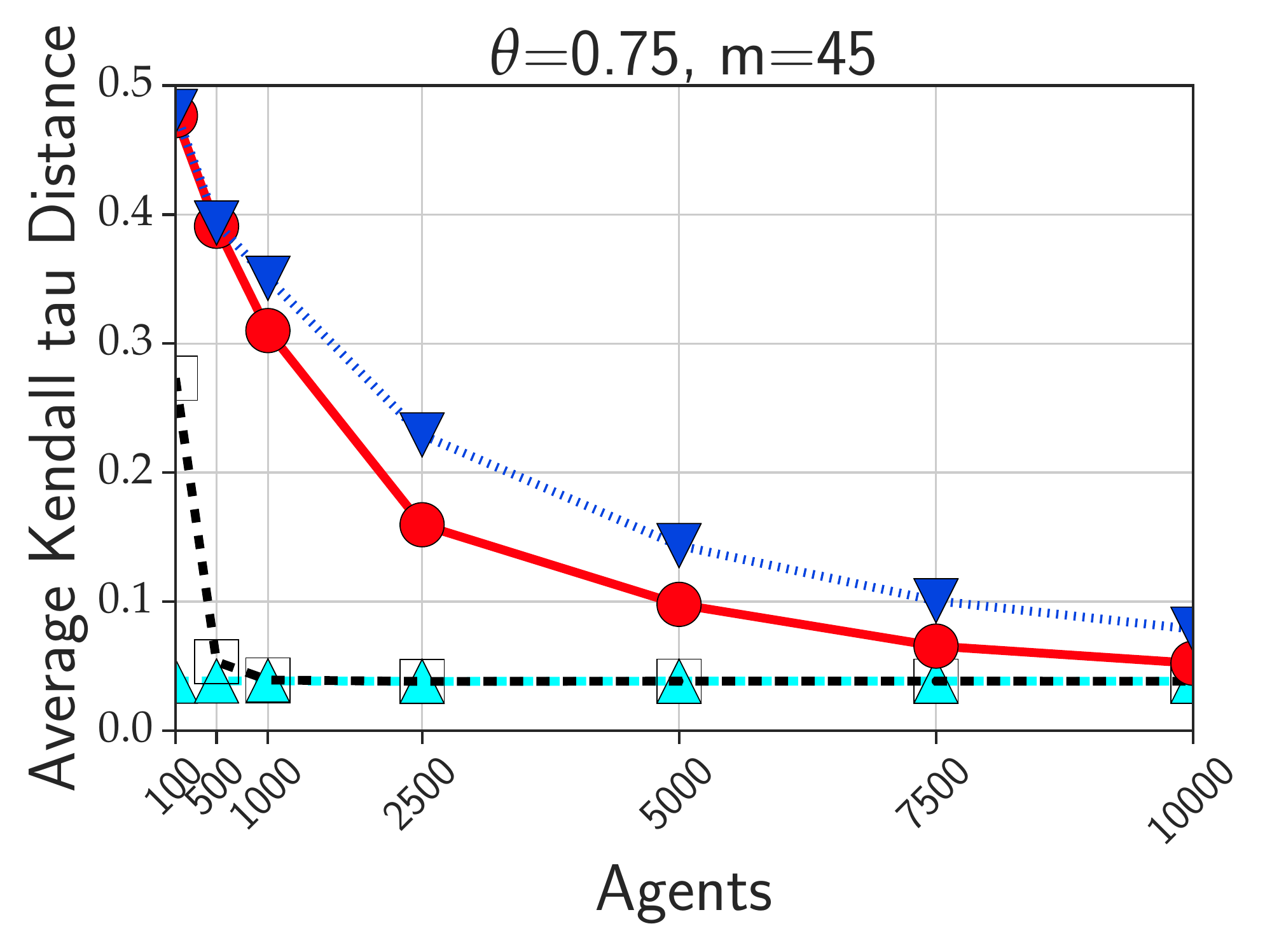}}\
\caption{Comparison of protocols in terms of average Kendall tau distance on the real-world and synthetic
datasets,
across varying the number of agents $n$.
Privacy budget $\epsilon$ is fixed at $2.0$.}
\label{fig:avgkt-varying-agents}
\end{figure*}

In~\Cref{subfigure:fig:avgkt-varying-agents:a}-\Cref{subfigure:fig:avgkt-varying-agents:c},
all the solutions are executed on the real-world datasets
in which the number of alternatives are fixed by default
and the privacy budget is set to $\epsilon=2.0$.
Firstly,
we observe that \texttt{KwikSort} algorithm
still shows the lower error bound of
the average Kendall tau distance.
For instance,
in~\Cref{subfigure:fig:avgkt-varying-agents:b},
with increasing the number of agents from $n=10$ to $n=100$,
\texttt{KwikSort} can achieve
the average Kendall tau distance below $0.4$
and get stable even with larger agent amount.
And we further observe that
when the number of agents is less than
the number of possible pairs of alternatives $\binom{m}{2}$,
\texttt{KwikSort} cannot
aggregate a sufficient pairwise comparison profile
in which each pairwise comparison $\texttt{cmp}_{\textbf{L}}(a_{j},a_{l})$
should be non-zero,
which challenges the generation
of optimal aggregate ranking.
Secondly,
the results show that
\texttt{LDP-KwikSort:RR} still outperforms
\texttt{LDP-KwikSort:Lap}.
Generally,
the involved solutions show the tendency of convergence
as the increasing number of agents.

In~\Cref{subfigure:fig:avgkt-varying-agents:d}-\Cref{subfigure:fig:avgkt-varying-agents:l},
we run all the solutions on the synthetic datasets
with $\theta \in \{0.25,0.5,0.75\}$,
varying the number of agents in a much larger range
from $n=100$ to $n=10000$,
and observe the performance of average Kendall tau distance
under $m \in \{15,30,45\}$ and $\epsilon=2.0$.
Firstly,
the results show that
with an increasing number of alternatives
from $m=15$ to $m=45$,
the lower error bound by \texttt{KwikSort} is decreasing,
but it is much harder for the \texttt{LDP-KwikSort} protocol
to achieve this bound
unless given a large number of agents.
For instance,
when given $m=15$ alternatives,
\texttt{LDP-KwikSort:RR} outperforms \texttt{LDP-KwikSort:Lap}
by $8.1\%$ at $n=2500$,
but with increasing $m$ to $30$,
\texttt{LDP-KwikSort:RR} needs more agents,
say $n=7500$,
to achieve the improvement of $9.8\%$.
Secondly,
with increasing the number of alternatives,
it will be more obvious to see the gap
between \texttt{LDP-KwikSort:RR}
and \texttt{LDP-KwikSort:Lap}.
For instance,
when setting $\theta=0.5$ and $n=2500$,
their gaps at $m=15$,
$m=30$ and $m=45$
are $2.4\%$,
$11\%$ and $32.5\%$,
respectively.

\subsection{The Impact of Dispersion Parameter}
\label{sec-ldp-ra-results-of-theta}

In the above experiments,
the dispersion parameters of synthetic dataset
$\theta=0.25$,
$\theta=0.5$ and $\theta=0.75$
are involved,
which reflects the generated rankings are closer
to the ground truth ranking.
From the results in~\Cref{fig:avgkt-varying-agents},
all solutions achieve lower average Kendall tau distance,
which demonstrates the theoretical conclusions
in~\Cref{sec-ldp-ra-utility-guarantee}.
Besides,
with increasing the dispersion parameter,
it will be more obvious to see the gap
between \texttt{LDP-KwikSort:RR}
and \texttt{LDP-KwikSort:Lap}.
For instance,
when setting $m=45$ and $n=5000$,
their gaps at $\theta=0.25$,
$\theta=0.5$ and $\theta=0.75$
are $13.5\%$,
$33.4\%$ and $46.5\%$,
respectively.

\subsection{Time Cost}
\label{sec-ldp-ra-time-cost}

Finally,
we compare the time costs of all the solutions
on three real-world datasets and five synthetic datasets.
The results are shown in~\Cref{tab:ldp-ra-time-cost-comparison}.
We observe that
\texttt{KwikSort} consumes the least execution time
and the time increases with more alternatives
because no noise introduced.
On the basis of the former,
the central model based \texttt{DP-KwikSort} algorithm
shows a small increase in time cost.
The cases of \texttt{LDP-KwikSort} protocol
present the sum of the execution time
of all the agents and the curator.
We observe that when the number of alternatives
are relatively small,
say $m=4,5$,
\texttt{LDP-KwikSort:RR} consumes less time
than \texttt{LDP-KwikSort:Lap}.
When increasing $m$ from $15$,
\texttt{LDP-KwikSort:RR} needs more time cost.

\begin{table*}[ht]
\centering
\scriptsize
\caption{Comparison of protocols in terms of time cost (in seconds)}
\label{tab:ldp-ra-time-cost-comparison}
%\begin{tabular}{c|l|l|llllllll}
\begin{tabular}{p{2.2cm}p{1.1cm}<{\centering}p{1.2cm}<{\centering}p{1.2cm}<{\centering}p{1.2cm}<{\centering}
p{1.2cm}<{\centering}p{1.2cm}<{\centering}p{1.2cm}<{\centering}p{1.2cm}<{\centering}}
\toprule
\diagbox[width=8em,trim=l]{Solution}{Dataset} & \tabincell{c}{\textsf{TurkDots}\\$795,4$} & \tabincell{c}{\textsf{TurkPuzzle}\\$793,4$} & \tabincell{c}{\textsf{SUSHI}\\$5000,10$} & \tabincell{c}{\textsf{Mallows}\\$5000,5$} & \tabincell{c}{\textsf{Mallows}\\$5000,10$} & \tabincell{c}{\textsf{Mallows}\\$5000,15$} & \tabincell{c}{\textsf{Mallows}\\$5000,30$} & \tabincell{c}{\textsf{Mallows}\\$5000,45$} \\
\midrule
\texttt{KwikSort} & 0.0001 & 0.0001 & 0.0002 & 0.0001 & 0.0001 & 0.0002 & 0.0005 & 0.0007 \\
\texttt{DP-KwikSort} & 0.0006 & 0.0006 & 0.0030 & 0.0008 & 0.0030 & 0.0076 & 0.0285 & 0.0626 \\
\texttt{LDP-KwikSort:RR} & 0.1255 & 0.1244 & 2.7076 & 2.5783 & 2.6125 & 2.8551 & 3.3530 & 4.4113 \\
\texttt{LDP-KwikSort:Lap} & 0.1751 & 0.1732 & 2.9088 & 2.9015 & 2.9026 & 2.9294 & 2.9812 & 3.1288 \\
\bottomrule
\end{tabular}
\end{table*}

\subsection{Discussion}
\label{sec-ldp-ra-discussion}

The above experimental results demonstrate
that the proposed \texttt{LDP-KwikSort}
can satisfy $\epsilon$-local differential privacy
or $\epsilon$-local individual differential privacy
while maintaining the acceptable utility of the aggregate ranking.
Particularly,
under the utility metrics such as the error rate
and the average Kendall tau distance,
solution \texttt{LDP-KwikSort:RR} generally outperforms
\texttt{LDP-KwikSort:Lap} and can achieve
the closest performance of \texttt{DP-KwikSort}.
When using the proposed protocol
in an agent scale as $n=5000$,
we recommend $\epsilon=1.0$ for the situations
with fewer alternatives
such as $m \leq 15$,
and $\epsilon=3.0$ when considering
a relatively large number of alternatives.

Due to the natural shortcoming of the local model of DP,
the observed limitation of this work relies on
the relatively large privacy budget to maintain
the acceptable utility,
compare with that of central model based solutions.
The potential optimization methods include
the consideration of personalized privacy settings.

\section{Conclusion}
\label{sec-ldp-ra-conclusion}

\emph{Rank aggregation} aims to combine different agents' preferences
over the given alternatives into an aggregate ranking
that agrees with the most with all the preferences.
In the scenario of crowdsourced data management,
since the aggregation procedure relies on a data curator,
the privacy within the agents' preference data
could be compromised when the curator is untrusted.
All existing works that guarantee differential privacy
in rank aggregation assume that the data curator is trusted.

This paper first formalizes and studies
the \emph{locally differentially private rank aggregation} (LDP-RA) problem.
Specifically,
we design the \texttt{LDP-KwikSort} protocol which
could protect the pairwise comparison
within the ranking list.
It also shows a combination of the properties from
the approximate rank aggregation algorithm \texttt{KwikSort},
the RR mechanism,
and the Laplace mechanism.
Theoretical analysis and empirical results
on the real-world and synthetic datasets
confirm that
our protocol especially
the solution \texttt{LDP-KwikSort:RR}
can achieve strong local privacy protection
while maintaining an acceptable utility.

Future work will include the following three aspects:
1) considering strategic voting behaviors and
exploring the trade-off between the soundness,
the usefulness
and the privacy preservation in LDP-RA;
2) extending our approach to support personalized privacy budget setting
for different agents;
3) the evaluation of the synthetic rank datasets
based on other mixture models
such as the Plackett-Luce model
and the general random utility model.

\clearpage
\bibliographystyle{plainnat}
\bibliography{reference}

\end{document}